\documentclass[a4paper,UKenglish,numberwithinsect]{lipics-v2021}

\usepackage{amsmath}
\usepackage{mathpartir}
\usepackage{amssymb}
\usepackage{mathtools}
\usepackage{xcolor}
\usepackage{graphicx}
\usepackage{calc}
\usepackage{xspace}

\usepackage{amsthm}
\usepackage{stmaryrd}
\usepackage{cases}

\usepackage{xfakebold}

\usepackage{cite}

\usepackage{tikz}
\usetikzlibrary{arrows,positioning,fit,calc,automata,shapes}
\tikzset{mynode/.style={rectangle,
rounded corners,draw,
text height=0.7em,
initial text=$ $,
fill=white},
align = justify}

\tikzset{mynodew/.style={rectangle,
rounded corners,draw,
text height=0.7em,
initial text=$ $,
fill=white},
align = justify}

\tikzset{mynodesh/.style={rectangle,
rounded corners,draw,
text height=0.7em,
initial text=$ $,
fill=white},
align = justify}

\tikzset{mynodel/.style={rectangle,
rounded corners,draw,
text height=0.7em,
initial text=$ $,
fill=white},
align = justify}

\tikzset{smallnode/.style={
initial text=$ $,
fill=white}}

\usepackage[capitalise]{cleveref}

\newcommand{\ag}{\mathsf{Ag}}

\newcommand{\R}{\mathcal{R}}



\makeatletter
\newcommand\bigDiamond{\mathop{\mathpalette\bigDi@mond\relax}}
\newcommand\bigDi@mond[2]{  \vcenter{\hbox{\m@th
    \scalebox{\ifx#1\displaystyle 2\else1.2\fi}{$#1\Diamond$}  }}}

\makeatletter
\def\bRightarrowfill@{\arrowfill@{\bm\Relbar}{\bm\Relbar}{\bm\Rightarrow}}
\def\xyarrow#1#2#3#4#5{\mathrel{\ext@arrow 0959{#1}{#3}{#2}{\!\!}^{#4}_{#5}}}
\def\xybRightarrow{\xyarrow\bRightarrowfill@}
\makeatother

\newcommand{\logicize}[1]{\ensuremath{\mathsf{#1}}}
\newcommand{\prop}{\logicize{Prp}}
\newcommand{\vbl}[1]{\logicize{#1}}

\newcommand{\true}{\top}
\newcommand{\false}{\bot}

\newcommand{\block}[2]{\ensuremath{\genfrac{[}{]}{0pt}{1}{#1}{#2}}}

\newcommand{\act}{\ensuremath{\mathsf{Act}}}

\newcommand{\goal}{\ensuremath{\varPhi_\mathsf{G}}}

\newcommand{\initial}{\ensuremath{s_\mathsf{I}}}

\newcommand{\M}{\ensuremath{\mathcal{M}}}
\newcommand{\A}{\ensuremath{\mathcal{A}}}
\newcommand{\pre}{\ensuremath{\mathtt{pre}}}
\newcommand{\LG}{\logicize{L}}

\newcommand{\kk}{\logicize{K}}

\newcommand{\remv}{\logicize{rmv}}
\newcommand{\ad}{\logicize{add}}
\newcommand{\nxtstg}{\logicize{nxt\_stg}}

\newcommand{\okstate}{\logicize{ok\_state}}

\renewcommand{\aa}{\vbl{a}}
\newcommand{\bb}{\vbl{b}}

\newcommand{\xx}{x}

\nolinenumbers

\title{An Undecidability Proof for the Plan Existence Problem}

\author{Antonis Achilleos}{Department of Computer Science, Reykjavik University, Iceland}{antonios@ru.is}{https://orcid.org/0000-0002-1314-333X}{}

\authorrunning{A. Achilleos}
    
        

\begin{document}

\maketitle

\begin{abstract}
The plan existence problem asks, given a goal in the form of a formula in modal logic, an initial epistemic state (a pointed Kripke model) and a set of epistemic actions, whether there exists a sequence of actions that can  be applied to reach the goal.
We prove that even in the case where the preconditions of the epistemic actions have modal depth at most 1, and there are no postconditions, the plan existence problem is undecidable.
The (un)decidability of this problem was previously unknown.
\end{abstract}

\section{Introduction}

In the field of Epistemic Planning (EP) \cite{bol-et-al-del-based-ep,bolander243gentle}, we seek to construct a \emph{plan}, that is, a sequence of actions that achieve a given \emph{goal}.
EP is an approach based on dynamic epistemic logic~\cite{dyn_epist_logic} that extends classical planning \cite{FIKES1971189} by incorporating the beliefs and knowledge of multiple agents, allowing for the modeling of complex interactions such as coordination, communication, and secrecy.

As expected, in EP, the goal is typically epistemic, in that it involves statements regarding the belief or knowledge of one or multiple agents, and it is expressed with a formula from a multi-agent epistemic logic. 
The state of the world (including the beliefs or knowledge of the agents) 
is represented by an \emph{epistemic state}, which is a pointed Kripke model. 
Finally, an \emph{epistemic action} transforms an epistemic state, by affecting the beliefs of the agents, and possibly also the values of the propositional atoms. The agents may  be uncertain about the epistemic actions, and therefore these are represented as pointed Kripke frames, where each state is called an event. 
Each event 
is further equipped with \emph{preconditions}, \emph{i.e.}~epistemic formulas that describe when the event can be applied to a world of the epistemic state; and with \emph{postconditions} that describe how the propositions of a world change value when the event is applied to that world (see \cite{bol-et-al-del-based-ep} for a more comprehensive description of EP).
Thus, a central problem in EP is the \emph{plan existence problem}: given an initial epistemic state, a goal, and a set of epistemic actions, is there a plan (a sequence of epistemic actions from the given ones) that transforms the initial epistemic state into one that satisfies the goal?

The plan existence problem is generally undecidable \cite{Bolander01012011}. Furthermore, it remains so when epistemic actions are restricted to having no postconditions  and preconditions of modal depth at most 2~\cite{on-the-impact-of-modal-depth};  and to having preconditions of modal depth at most 1 and propositional postconditions \cite{small-problems,bol-et-al-del-based-ep}.
On the other hand, it is known that plan existence is decidable when both preconditions and postconditions are propositional \cite{Yu-Wen-Liu2013multiagent,aucher:hal-01098740,doueneau2018chain}.
The case of preconditions of modal depth 1 and no postconditions, to the best of our knowledge, remains open to this day.

\paragraph*{Our contribution}
We prove that the plan existence problem with preconditions of modal depth at most 1 and no postcondition \cite{bol-et-al-del-based-ep}, is undecidable.
Our proof is by a reduction from Post's Correspondence Problem (PCP) \cite{Post1946AVO}.
We prove that the problem is undecidable in the case of at least two agents without any assumptions on the modal logic that the  agents' beliefs are based on. 
For the single-agent case, we need to distinguish between two cases. 
If the agent's knowledge behaves according to logic S5, or even if the agent has negative introspection (they know what they do not know), then it is a somewhat straightforward observation that the plan existence problem is decidable.
When the agent does not have negative introspection, we consider agents based on a modal logic that extends the basic logic K with axioms $T, B, 4$, but not both $B$ and $4$, as these imply negative introspection. 
For these cases, we prove that the plan existence problem is undecidable. 

\subparagraph*{Our reduction} makes substantial use of the fact that successive applications of epistemic actions 
can construct long paths with starting and ending states definable by a formula of modal depth 1. We need this property, as we cannot use postconditions to cycle through the states of an encoded machine, nor general preconditions to access the status of worlds further away. Thus, our constructions need to operate in a local manner, having events with preconditions that depend on the immediately accessible worlds.

In \Cref{sec:back}, we give the necessary background on epistemic logic and epistemic planning.
In \Cref{sec:K}, we prove the undecidability of the plan existence problem for a single agent based on modal logic K. This case is likely not particularly realistic for planning situations, or general, but it is sufficiently simple to help us demonstrate our reduction method.
In \Cref{sec:multi}, we extend our reduction to the case of multiple agents, and in \Cref{sec:rs} we show how to further adjust our reduction to the case of agents based on logics KTB and KT4, and therefore on any logic between K and KTB, KT4.
We conclude in \Cref{sec:conclusions} and give a straightforward decidability result for a single agent with negative introspection.

\section{Background}
\label[section]{sec:back}

Modal formulae extend propositional logic with a set of modal operators $K_i$, $i \in \ag$, often interpreted as asserting an agent $i$'s knowledge or belief of the statement it is applied on.
We use a set of propositional variables $\prop$ that can be assumed to be finite or countably infinite. The set of modal formulae is given by the following grammar:
\begin{align*}
  \varphi ::= \false \mid p \mid \neg \varphi \mid \varphi \land \varphi \mid K_i \varphi,
\end{align*}
where $p \in \prop$ and $i \in ag$. We use the usual definitions for the dual operators of the language: $\true := \neg \false$; $\varphi_1 \lor \varphi_2 := \neg (\neg \varphi_1 \land \neg \varphi_2)$; and $\bar{K} \varphi := \neg K \neg \varphi$.

A Kripke model is a triple $(W,R,V)$, where $W$ is a set of \emph{states} (or worlds), representing different possibilities about reality; $R = (R_i)_{i \in \ag}$, where $R_i \subseteq W^2$ is an \emph{accessibility relation}, where $(w,w') \in R_i $ represents that at state $w$, agent $i$ considers $w'$ a possible reality; and $V: W \to 2^\prop$ labels each state with a set of propositional variables, with the intended meaning that these are the variables that are true at that state. We usually write $w R_i w'$ instead of $(w,w') \in R_i$. The pair $(W,R)$ is a \emph{Kripke frame}.

An \emph{epistemic state}, also known as a pointed model, is a pair $(\M,w)$, where $\M = (W,R,V)$ is a Kripke model and $w \in W$ --- it can also be written as a quadruple $(W,R,V,w)$.
We interpret modal formulae on epistemic states. Let $\M = (W,R,V)$ be a Kriple model and $w \in W$. Then we define $\models$ between epistemic states and formulae in the following way.
\begin{align*}
  (\M,w) &\not \models \false \\ 
  (\M,w) &\models p &\text{ if and only if }& &p \in V(w) \\ 
  (\M,w) &\models \neg \varphi &\text{ if and only if }& &(\M,w) \not \models \varphi \\
  (\M,w) &\models \varphi_1 \land \varphi_2 &\text{ if and only if }& &(\M,w) \models \varphi_1 \text{ and }(\M,w) \models \varphi_2 \\ 
  (\M,w) &\models K_i \varphi &\text{ if and only if }& &(\M,w') \models \varphi \text{ for every } w R_i w'.
\end{align*}

In the case where $|\ag| = 1$, we write $K \varphi$ instead of $K_i \varphi$, and $R = R_i$. 
For each epistemic state $s = (\M,w) = (W,R,V,w)$ and $u \in W$, we use $(s,u)$ to denote the epistemic state $(\M,u)$.

An \emph{epistemic action} is a pair $(\A,e)$, where $\A = (E,Q,\pre)$, $(E,Q)$ is a Kripke frame, $e \in E$, and $\pre : E \to \LG$ maps each element of $E$ to a modal formula.
Elements of $E$ are called \emph{events} and for each event $f$, $\pre(f)$ is called the precondition of $f$.
We restrict ourselves to preconditions of modal depth at most $1$, that is, for each $f \in E$, $\pre(f)$ does not have any nested modal operators.

We say that an epistemic action $(\A,e)$ \emph{applies} to an epistemic state $(\M,w)$, where $\A = (E,Q,\pre)$ and $\M = (W,R,V)$, when $(\M,w) \models \pre(e)$. In that case, we define the application of $(\A,e)$ to $(\M,w)$, and we denote it as $(\M,w) \times (\A,e) = ((W',R',V'),(w,e))$, where
\begin{align*}
  W' =& \{ (u,f) \in W\times E \mid (\M,u) \models \pre(f) \}; \\
  R'_i =& \{ ((u_1,f_1),(u_2,f_2)) \in W'^2 \mid (u_1,u_2) \in R_i \text{ and } (f_1,f_2) \in Q_i \}; \text{ and} \\ 
  V'(u,f) =& V(u) \text{ for every } (u,f) \in W'.
\end{align*}

We can recursively extend the definition of the application of an epistemic action to an epistemic state for sequences of actions: if $\alpha$ is a sequence of epistemic actions, $s$ is an epistemic state,  $s \times \alpha$ is well-defined, and 
$a$ is an epistemic action that applies to $s \times \alpha$, then $s \times \alpha a = (s \times \alpha) \times a$.

We now define our (constrained) version of the plan existence problem.

\begin{definition}[Plan Existence Problem]\label[definition]{def:plan-existence-problem}
Given:
\begin{itemize}
    \item an initial epistemic state $\initial$;
    \item a goal $\varphi_\mathcal{G}$, which describes the desired knowledge or belief conditions to be achieved; and 
    \item a set of epistemic actions $\act$, where each action has preconditions of modal depth at most $1$,
\end{itemize}
the \emph{Plan Existence Problem} asks whether there exists a plan, \emph{i.e.}~a sequence of actions $\alpha = a_1a_2\cdots a_m \in \act^*$, 
such that
$\initial \times \alpha$ is defined and $\initial \times \alpha \models \goal$.
\end{definition}

\begin{remark}
Typically (\emph{e.g.}~\cite{bol-et-al-del-based-ep,bolander243gentle}), the plan existence problem is defined more generally, allowing preconditions of arbitrary modal depth and postconditions that allow an event to change the values of propositional variables. Here, we restrict our attention to the case of no postconditions and preconditions of modal depth at most $1$, as this suffices to prove undecidability.
\end{remark}

We observe that if two epistemic states $s_1, s_2$ are bisimilar ($s_1 \sim s_2$), then for every epistemic action $\alpha$, $s_1 \times \alpha \sim s_2 \times \alpha$. Therefore, plans and plan existence are preserved under bisimulation, and, for convenience we will often 
identify 
epistemic states with their bisimilar states.

\paragraph*{Frame conditions} 
The semantics we have given for modal formulas correspond to the simplest modal logic $\kk^n$ on $n$ agents.
By imposing restrictions on the accessibility relations for the frames that we allow, the resulting semantics correspond to modal logics with different behaviours --- that is, each of the frame properties that we examine corresponds to a specific modal axiom. The relationship between frame conditions and modal axioms is a well-studied phenomenon called correspondence theory (see, for example,~\cite{SAHLQVIST1975110,van1984correspondence}).

In this paper, we consider the following typical properties of the accessibility relations:
\begin{description}
  \item[reflexivity] corresponds to the factivity axiom: $K_i \varphi \to \varphi$ and is usually required when the operator $K_i$ corresponds to knowledge;
  \item[transitivity] corresponds to positive introspection: $K_i \varphi \to K_i K_i \varphi$;
  \item[symmetry] corresponds to the axiom $\varphi \to K_i \bar{K}_i \varphi$ and usually results from the Eucledian property (below) and reflexivity;
  \item[the Eucledian property]  corresponds to negative introspection: $\neg K_i \varphi \to K_i \neg K_i \varphi$.
\end{description}

\paragraph*{Post's Correspondence Problem and the Reduction Plan}

To prove that the plan existence problem is undecidable, we reduce from Post's Correspondence Problem (PCP) \cite{Post1946AVO}, a well-known undecidable problem.

\begin{definition}[Post's Correspondence Problem (PCP),~\cite{Post1946AVO}:]

Given $n \geq 1$ and $B = \{ \block{a_i}{b_i} \mid 1 \leq i \leq n\}$, where for every $i \leq n$, $a_i, b_i \in \{0,1\}^*$, the \emph{Post's Correspondence Problem} asks whether there exists a non-empty sequence of indices $i_1, i_2, \dots, i_k$ such that:
\[
a_{i_1} a_{i_2} \dots a_{i_k} = b_{i_1} b_{i_2} \dots b_{i_k}.
\]
Such a sequence is called a \emph{match} for $B$.

\end{definition}

Let $n \geq 1$ and $B = \{ \block{a_i}{b_i} \mid 1 \leq i \leq n\}$, where for every $i \leq n$, $a_i, b_i \in \{0,1\}^*$.
The following sections demonstrate how $B$ can be turned into an instance of the plan existence problem, such that $B$ has a match if and only if that instance has a plan.
\subparagraph*{The overall idea} is that a plan constructs a sequence of epistemic states in two stages. In the first stage, the application of the epistemic actions represents adding blocks to a sequence of blocks from $B$, which is encoded as two branching paths of the epistemic state, one representing the top of the sequence, and the second the bottom of the sequence. In the second stage, the application of the epistemic states verifies that the resulting sequence is a match, by iteratively removing one symbol from the end of the sequence --- \emph{i.e.}~by removing the last state from each path, if these match. If both branches can eventually be pruned to the root, the plan has succeeded in finding a match.

We first give the reduction for the case of a single agent with no restrictions on the accessibility relations.
This case allows us to demonstrate our method in a clear setting; we then adjust our construction for the cases of multiagent logics with any set of frame conditions; and to single-agent logics with frame conditions that may include reflexivity, symmetry, or transitivity, \emph{but not both} symmetry and transitivity.

\section{The Plan Existence Problem is Undecidable for Single-agent $\kk$}

We now show how to reduce PCP to the plan existence problem for the case when the accessibility relations have no restrictions and there is only one agent. 
Let $EP = (\initial,\act,\goal)$, which we define in the following.

\label{sec:K}We use the following set of propositional variables: 
\[
\prop = \{ 0, 1, 
\aa, \bb, \vbl{root}, \vbl{stg1}, \vbl{empty}, \vbl{end}, \vbl{ntF}, \vbl{lp}  \}.
\]

\subsection{The Epistemic States}

The initial epistemic state $\initial = ((W_0,R_0,V_0),w_{\vbl{root}})$, where
\begin{align*}
        W_0 = \{ &w_p, w_{\vbl{bt},\xx} &\mid p \in \prop\setminus \{0,1\}, \vbl{bt}\in \{0,1\}, \xx \in \{\aa,\bb\} \};\\
            R_0 = \{ 
        &(w_{\vbl{root}},w_{\vbl{empty}}), (w_{\vbl{root}},w_{\vbl{stg1}}), (w_{\vbl{root}},w_{\xx}), 
        \\ 
    &
    (w_\xx,w_{\vbl{bt},\xx}), 
                                                    (w_{\xx},w_{\vbl{end}}),
     (w_{\xx},w_{\vbl{ntF}}),
     \\ 
    &
     (w_{\vbl{bt},\xx},w_{\vbl{ntF}}),
     (w_{\vbl{bt},\xx},w_{\vbl{end}}),
     (w_{\vbl{bt},\xx},w_{\vbl{lp}}),
     \\ 
     &
         (w_{\vbl{0},\xx},w_{\vbl{bt},\xx}), (w_{\vbl{1},\xx},w_{\vbl{bt}})
          &\mid x \in \{\aa,\bb\}, ~ \vbl{bt} \in \{0,1\} \};     \end{align*}
        and 
    for every $p \in \prop\setminus \{0,1\}$ and $\xx \in \{\aa, \bb\}$, $V_0(w_p) = \{p\}$, $V_0(w_{0,\xx}) = \{0,\xx\}$, and $V_0(w_{1,\xx}) = \{1,\xx\}$.

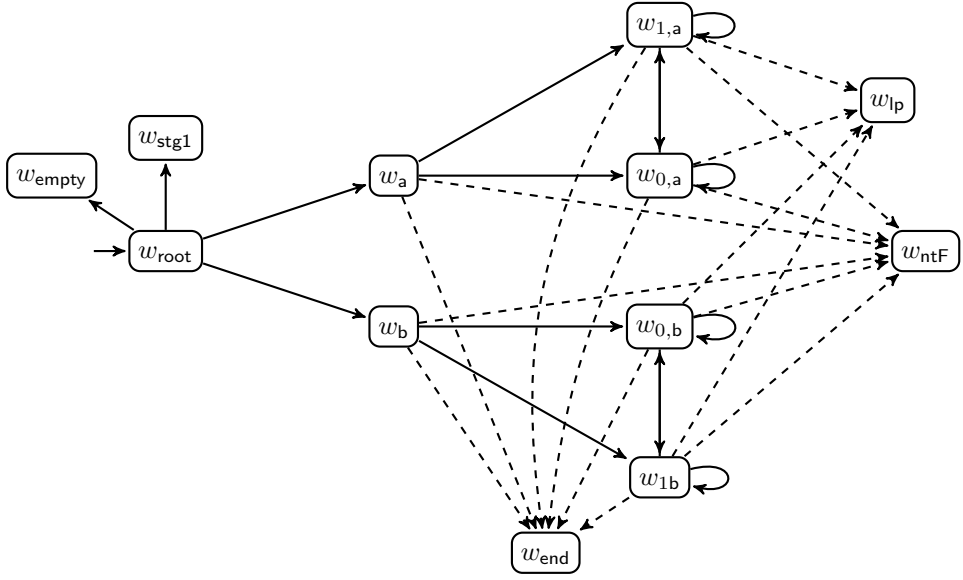
\begin{figure}
  \centering 
\begin{tikzpicture}[->,>=stealth',shorten >=1pt,auto,node distance=2.5cm, thick]

    \node[mynode,initial] (r0)   at (0, 0)    {$w_{\vbl{root}}$};
  \node[mynode] (stg)   at (0, 1.5)    {$w_{\vbl{stg1}}$};
  \node[mynode] (empty)   at (-1.5, 1)    {$w_{\vbl{empty}}$};
  \node[mynodesh] (wa)   at (3, 1)  {$w_{\aa}$};
  \node[mynodesh] (wb)   at (3, -1) {$w_{\bb}$};
  
  \node[mynodesh] (w0a)   at (6.5, 1)  {$w_{0,\aa}$};
  \node[mynodesh] (w1a)   at (6.5, 3) {$w_{1,\aa}$};
    \node[mynodesh] (wnf) at (10, 0)   {$w_{\vbl{ntF}}$};

  \node[mynodesh] (wlp) at (9.5, 2)   {$w_{\vbl{lp}}$};
  
  \node[mynodesh] (w0b)   at (6.5, -1)  {$w_{0,\bb}$};
  \node[mynodesh] (w1b)   at (6.5, -3) {$w_{1\bb}$};
  \node[mynodesh] (wend) at (5, -4)   {$w_{\vbl{end}}$};

    \draw[->] (r0) -- (stg);
  \draw[->] (r0) -- (empty);
  \draw[->] (r0) -- (wa);
  \draw[->] (r0) -- (wb);

  \draw[->] (wa) -- (w0a);
  \draw[->] (wa) -- (w1a);
  \draw[->] (wb) -- (w0b);
  \draw[->] (wb) -- (w1b);

    \draw[->] (w0a) -- (w1a);

  \draw[->] (w1a) -- (w0a);

  \draw[->] (w0a) edge[loop right] (w0a);
  \draw[->] (w1a) edge[loop right] (w1a);

  \draw[->] (w0b) -- (w1b);

  \draw[->] (w1b) -- (w0b);

  \draw[->] (w0b) edge[loop right] (w0b);
  \draw[->] (w1b) edge[loop right] (w1b);

    \draw[->,dashed] (wa) -- (wend);
  \draw[->,dashed] (wb) -- (wend);

  \draw[->,dashed] (w0a) edge[bend right=10] (wend);
  \draw[->,dashed] (w1a) edge[bend right=20] (wend);

  \draw[->,dashed] (w0a) edge   (wnf);
  \draw[->,dashed] (w1a) edge  (wnf);

  \draw[->,dashed] (w0b) -- (wend);
  \draw[->,dashed] (w1b) -- (wend);

  \draw[->,dashed] (w0b) edge  (wnf);
  \draw[->,dashed] (w1b) edge  (wnf);

  \draw[->,dashed] (wa) edge  (wnf);
  \draw[->,dashed] (wb) edge  (wnf);

  \draw[->,dashed] (w0a) edge   (wlp);
  \draw[->,dashed] (w1a) edge  (wlp);

  \draw[->,dashed] (w0b) edge  (wlp);
  \draw[->,dashed] (w1b) edge  (wlp);

\end{tikzpicture}
\caption{The initial epistemic state $\initial$. The designated state $w_{\vbl{root}}$ is marked with an arrow pointing to it. We used dashed arrows to the auxiliary states $w_{\vbl{end}}, w_{\vbl{ntF}}$ in this and in other figures, to emphasize the part of the model that encodes a sequence of blocks and uses solid arrows.}
\end{figure}

\begin{remark}
  The reader will likely realise while going through the constructions in the paper, that we could use simpler and smaller epistemic states.
  The ones we use, however, will help shorten and simplify some of our upcoming arguments, which we consider a significant gain.
\end{remark}

Intuitively, a plan is split into two stages, and the epistemic actions are split into three kinds. The first kind is applied during the first stage, and each action amounts to adding a block from $B$ at the end of the block sequence; the second kind consists of one action,  which signals the transition from the first stage to the second stage by removing state $w_{\vbl{stg1}}$ or its analogous in the current model; and the third kind is applied during the second stage to iteratively delete matching symbols from the end of the constructed block sequence.

The states $w_{\vbl{empty}}, w_{\vbl{stg1}}, w_{\vbl{end}}, w_{\vbl{ntF}}$ can be thought of as auxiliary, as they are used to keep track of useful information about the states of an epistemic state. 
As we apply actions to the epistemic states, in a state, the truth of
\begin{itemize}
  \item $\bar{K} \vbl{empty}$ encodes that we have applied no actions (and the state is the root state);
  \item $\bar{K} \vbl{stg1}$ encodes that we are on the first stage of the plan;
  \item $\bar{K} \vbl{end}$ encodes that we are at a state at the end of the encoding of the blocks; and   \item $\bar{K} \vbl{ntF}$ encodes that the plan has not attempted to remove the wrong symbol from the state during the second stage --- \emph{i.e.}, the plan has not failed yet.
\end{itemize}

\begin{definition}
Let $q_\aa, q_\bb \in \{0,1\}^*$.
We define the epistemic state $s\block{q_\aa}{q_\bb}$ to be (any epistemic state bisimilar to) the epistemic state 
$s\block{q_\aa}{q_\bb} = ((W\block{q_\aa}{q_\bb},R\block{q_\aa}{q_\bb},V\block{q_\aa}{q_\bb}),w_{\vbl{root}})$, where
    \begin{align*}
      W\block{q_\aa}{q_\bb} = \{ &w_p, w_\xx &&\mid p \in \{ \vbl{root}, \aa, \bb, \vbl{ntF} \},~ x \in \{\aa,\bb\}  \} 
      \cup \\ 
      \{ & 
      w_{\xx,j}
                                    &&\mid 
            x \in \{\aa,\bb\}  \text{ and } 1 \leq j \leq |q_\xx|
            \};
        \\ 
    R\block{q_\aa}{q_\bb} = \{
        &(w_{\vbl{root}},w_{\xx}), (w_\xx,w_{\vbl{ntF}})
    && \mid 
    x \in \{\aa,\bb\}
    \} ~\cup 
     \\ 
     \{& 
        (w_\xx,w_{\xx,1}), (w_{\xx,|q_{\xx}|},w_{\vbl{ntF}})
    && \mid 
    x \in \{\aa,\bb\}, 
    1 \leq |q_{\xx}|
        \} ~\cup 
    \\ 
    \{& 
    (w_{\xx,j},
    w_{\xx,j+1}), 
                (w_{\xx,j},w_{\vbl{ntF}})
       && \mid 
    x \in \{\aa,\bb\}, 
    1 \leq j < |q_{\xx}|
        \} 
                        ; \text{ and }
    \end{align*}
        for every $p \in \{ \vbl{root},\aa,\bb,\vbl{ntF} \}$, $V(w_p) = \{p\}$, and for every $\xx \in \{ \aa, \bb \}$ and $1 \leq j \leq |q_{\xx}|$, 
        $V(w_{\xx,j}) = \{q_{\xx}[j]\}$.
                        \end{definition}

The epistemic state $s\block{q_\aa}{q_\bb}$ is used to encode a sequence of blocks from $B$ whose top parts form $q_\aa$ and bottom parts form $q_\bb$; it is illustrated in \Cref{fig:sab}.
We also use the epistemic state $s\block{q_\aa}{q_\bb}\{0,1\}$ defined below, to encode the same sequences of blocks, while also retaining a loop of states that allows the plan to add more blocks to the sequence.

\begin{figure}
  \centering 
\begin{tikzpicture}[->,>=stealth',shorten >=1pt,auto,node distance=2.8cm, thick]

    \node[mynode,initial] (wroot)   at (0, 0)      {$w_\vbl{root}$};
    \node[mynode] (wa)   at (1, 1)    {$w_\aa$};
  \node[mynode] (wb)   at (1, -1)   {$w_\bb$};
  \node[mynode] (wj0a) at (3, 1)    {$w_{a,1}$};
  \node[mynode] (wj0b) at (3, -1)   {$w_{b,1}$};
  \node[mynode] (wj1a) at (5, 1)    {$w_{a,2}$};
  \node[mynode] (wj1b) at (5, -1)   {$w_{b,2}$};

  \node[mynode] (ejka) at (9, 1)   {$w_{a,|q_\aa|}$};
  \node[mynode] (ejkb) at (9, -1)  {$w_{b,|q_\bb|}$};

  \node[mynode] (wnf) at (6, 0)   {$w_{\vbl{ntF}}$};

      \draw[->] (wroot) -- (wa);
  \draw[->] (wroot) -- (wb);
  
  \draw[->] (wa) -- (wj0a);
  \draw[->] (wb) -- (wj0b);
  
  \draw[->] (wj0a) -- (wj1a);
  \draw[->] (wj0b) -- (wj1b);

  \draw[dotted] (wj1a) -- (ejka);
  \draw[dotted] (wj1b) -- (ejkb);

  \draw[->,line width=0.1,dashed] (wj0a) -- (wnf);
  \draw[->,line width=0.1,dashed] (wj0b) -- (wnf);

  \draw[->,line width=0.1,dashed] (wj1a) -- (wnf);
  \draw[->,line width=0.1,dashed] (wj1b) -- (wnf);

  \draw[->,line width=0.1,dashed] (ejka) -- (wnf);
  \draw[->,line width=0.1,dashed] (ejkb) -- (wnf);

  \draw[->,line width=0.1,dashed] (wa) edge (wnf);
  \draw[->,line width=0.1,dashed] (wb) -- (wnf);

\end{tikzpicture}
\caption{The epistemic state $s\block{q_\aa}{q_\bb}$. 
}\label[figure]{fig:sab}
\end{figure}

\begin{definition}
  Let $q_1, q_2 \in \{0,1\}^*$. We define the epistemic state $s\block{q_\aa}{q_\bb}\{0,1\}$ to be 
  $s\block{q_\aa}{q_\bb}\{0,1\} = ((W'\block{q_\aa}{q_\bb},R'\block{q_\aa}{q_\bb},V'\block{q_\aa}{q_\bb}),w_{\vbl{root}})$, where
  \begin{align*}
                  W'\block{q_\aa}{q_\bb} &= W\block{q_\aa}{q_\bb} & \cup \{ & w_{\vbl{stg1}}, w_{\vbl{0},\xx}, w_{\vbl{1},\xx}, w_{\vbl{end}} ,  w_{\vbl{lp}} 
      \mid \xx \in \{\aa,\bb\}
            \}; \\
            R'\block{q_\aa}{q_\bb} &= 
        R\block{q_\aa}{q_\bb} &\cup \{                   &
      (w_{\vbl{root}},w_{\vbl{stg1}}),                         (w_{\vbl{0},\xx},
      w_{\vbl{1},\xx}), 
      (w_{\vbl{1},\xx},
                        w_{\vbl{0},\xx}),
      (w_{\vbl{0},\xx},
      w_{\vbl{0},\xx}), 
      (w_{\vbl{1},\xx},
      w_{\vbl{1},\xx}),
      \\
      && 
       &
      (w_{\vbl{0},\xx},w_{\vbl{ntF}}),  
      (w_{\vbl{1},\xx},w_{\vbl{ntF}}),
      (w_{\vbl{0},\xx},w_{\vbl{end}}),  
      (w_{\vbl{1},\xx},w_{\vbl{end}}),
      \\
      && 
       &
      (w_{\vbl{0},\xx},w_{\vbl{lp}}),  
      (w_{\vbl{1},\xx},w_{\vbl{lp}})
                        \mid ~ \xx \in \{\aa,\bb\}
      \}  \\
      &&
      \cup 
      \{ 
       &(w_{\xx,|q_{\xx}|
               },w_{\vbl{bt},\xx}) ,
        (w_{\xx,|q_{\xx}|
             },w_{\vbl{end}}) 
        \mid ~ \xx \in \{\aa,\bb\}, ~ q_\xx \neq \varepsilon, ~\vbl{bt} \in \{0,1\}
      \}  \\
      &&
      \cup 
      \{ 
       & (w_{\xx},w_{\vbl{bt},\xx}), (w_{\xx},w_{\vbl{end}}) \mid ~ \xx \in \{\aa,\bb\}, ~ q_\xx = \varepsilon, ~ \vbl{bt} \in \{0,1\} \}
      ; \text{ and }
  \end{align*}
            for every $w \in W\block{q_\aa}{q_\bb}$, $V'\block{q_\aa}{q_\bb}(w) = V\block{q_\aa}{q_\bb}$, for 
      $p \in \{ \vbl{stg1},
            \vbl{end}, \vbl{lp} \}$, $V'\block{q_\aa}{q_\bb}(w_p) = \{p\}$, and for 
      $p \in \{\vbl{0},\vbl{1} \}$ and $\xx \in \{\aa,\bb\}$, $V'\block{q_\aa}{q_\bb}(w_{p,\xx}) = \{p,\xx\}$.

      \end{definition}

\begin{figure}
  \centering 
\begin{tikzpicture}[->,>=stealth',shorten >=1pt,auto,node distance=2.8cm, thick]

    \node[mynode,initial] (wroot)   {$w_\vbl{root}$};
  \node[mynode] (wst)  [above =0.7 of wroot]     {$w_{\vbl{stg1}}$};
  \node[mynodesh] (wa)   [above right = 0.7 and 0.8 of wroot]    {$w_\aa$};
  \node[mynodesh] (wb)  [below right =0.7 and 0.8 of wroot]   {$w_\bb$};
  \node[mynode] (wj0a) [right=1 of wa]   {$w_{a,1}$};
  \node[mynode] (wj0b) [right =1 of wb] {$w_{b,1}$};
  \node[mynode] (wj1a) [right=0.8 of wj0a]   {$w_{a,2}$};
  \node[mynode] (wj1b) [right =0.8 of wj0b] {$w_{b,2}$};

  \node[mynode] (ejka) [right =1.5 of wj1a]  {$w_{a,|q_\aa|}$};
  \node[mynode] (ejkb) [right =1.5 of wj1b] {$w_{b,|q_\bb|}$};

  \node[mynode] (w0a)   [above right=1 and 0.5 of ejka]  {$w_{0,\aa}$};
  \node[mynode] (w1a)  [right =1 of ejka] {$w_{1,\aa}$};
  \node[mynode] (w0b)   [right =1 of ejkb]  {$w_{0,\bb}$};
  \node[mynode] (w1b)  [below right=1 and 0.5 of ejkb] {$w_{1,\bb}$};
      \node[mynode] (wend) [below =0.7 of ejka] {$w_{\vbl{end}}$};
  
  \node[mynode] (wnf) [below = 0.7 of wj0a]   {$w_{\vbl{ntF}}$};
  
    \node[mynodesh] (wlp)   [below right = 0.5 and 0.5 of w1a]  {$w_\vbl{lp}$};

    \draw[->] (wroot) -- (wst);
  \draw[->] (wroot) -- (wa);
  \draw[->] (wroot) -- (wb);
  
  \draw[->] (wa) -- (wj0a);
  \draw[->] (wb) -- (wj0b);

  \draw[->] (wj0a) -- (wj1a);
  \draw[->] (wj0b) -- (wj1b);

  \draw[dotted] (wj1a) -- (ejka);
  \draw[dotted] (wj1b) -- (ejkb);

  \draw[->] (ejka) -- (w0a);
  \draw[->] (ejka) -- (w1a);
  \draw[->] (ejkb) -- (w0b);
  \draw[->] (ejkb) -- (w1b);

  \draw[<->] (w0a) -- (w1a);
  \draw[->] (w0a) edge[loop above] (w0a);
  \draw[->] (w1a) edge[loop below] (w1a);

  \draw[->,dashed] (w0a) -- (wend);
  \draw[->,dashed] (w1a) -- (wend);

   \draw[<->] (w0b) -- (w1b);
  \draw[->] (w0b) edge[loop above] (w0b);
  \draw[->] (w1b) edge[loop below] (w1b);

  \draw[->,dashed] (w0b) -- (wend);
  \draw[->,dashed] (w1b) -- (wend);

      \draw[->,dashed] (ejka)  -- (wend);
  \draw[->,dashed] (ejkb) -- (wend);

  \draw[->,dashed] (w0a) 
  --
    (wnf);
  \draw[->,dashed] (w1a) 
    --
  (wnf);

  \draw[->,dashed] (w0b) 
  --
    (wnf);
  \draw[->,dashed] (w1b) 
    --
  (wnf);

  \draw[->,dashed] (w0a) 
  --
    (wlp);
  \draw[->,dashed] (w1a) 
    --
  (wlp);

  \draw[->,dashed] (w0b) 
  --
    (wlp);
  \draw[->,dashed] (w1b) 
    --
  (wlp);

  \draw[->,line width=0.1,dashed] (wj0a) -- (wnf);
  \draw[->,line width=0.1,dashed] (wj0b) -- (wnf);

  \draw[->,line width=0.1,dashed] (wj1a) -- (wnf);
  \draw[->,line width=0.1,dashed] (wj1b) -- (wnf);

  \draw[->,line width=0.1,dashed] (ejka) -- (wnf);
  \draw[->,line width=0.1,dashed] (ejkb) -- (wnf);

  \draw[->,line width=0.1,dashed] (wa) -- (wnf);
  \draw[->,line width=0.1,dashed] (wb) -- (wnf);
\end{tikzpicture}
\caption{The epistemic state $s\block{q_\aa}{q_\bb} \{0,1\}$. 
}
\end{figure}
  
Observe that $\initial \neq s\block{\varepsilon}{\varepsilon}\{0,1
\}$, as state $w_{\vbl{empty}}$ is absent from $s\block{\varepsilon}{\varepsilon}\{0,1\}$.

As we will see, a successful plan will generate a sequence of epistemic states of the form 
$s\block{q_\aa}{q_\bb}\{0,1\}$ or $s\block{q_\aa}{q_\bb}$. 

\subsection{The Epistemic Actions}

We now proceed to define the epistemic actions that we use. 
We use the following shorthand formulas:
\begin{align*}
                            \logicize{symb} & = 0 \lor 1 
    ,
  &
  \logicize{tail} & = 
  (
      \aa \lor \bb) \land K \neg \logicize{symb}
  \\ 
  \logicize{last} & = \bar{K} \vbl{end} \land K \neg \vbl{lp}
                                        ,~~~
    &
  \logicize{failed}
  & =
                        (\aa \lor \bb 
                ) 
        \land 
        K \neg \vbl{ntF}
                        ,~~~
  \text{ and}
  \\
  \logicize{loop}_\xx & = \xx \land \bar{K} \vbl{lp},
            \end{align*}
for every $\xx \in \{\aa,\bb\}$.

Note that the states in $s \block{q_\aa}{q_\bb} \{0,1\}$ that satisfy $\logicize{loop}_\xx$ are exactly $w_{0,\xx}, w_{1,\xx} 
$.
The states that satisfy $\logicize{last}$ are $w_{\xx,|q_\xx|}$, if $q_\xx \neq \varepsilon$, and $w_\xx$, if $q_\xx = \varepsilon$, where $\xx \in \{\aa,\bb\}$.
The states in $s \block{q_\aa}{q_\bb}$ that satisfy $\logicize{tail}$ are exactly $w_{\xx,|q_\xx|}$, if $q_\xx \neq \varepsilon$, and $w_\xx$, if $q_\xx = \varepsilon$.
No state in $s \block{q_\aa}{q_\bb}$ or $s \block{q_\aa}{q_\bb} \{0,1\}$ satisfies $\logicize{failed}$, which will be used to mark a state that witnesses that the plan has failed (see \Cref{lem:failed-states-fail}).

\subsubsection{Adding Blocks}
For every $1 \leq i \leq n$, we define 
$
\ad_i = (\mathcal{AD}_i,e_s), ~~ \text{ where }~ \mathcal{AD}_i = (E_i,R_i,\pre), $
        \begin{align*}
                        E_i = \{& e_s,  
                                    e_\xx,  
                                                      e_{\xx,lst},
                                    e_{st}, e_{\vbl{end}}, e_{\vbl{ntF}} ,
                  e_{01,\xx}
                                    &\mid& \xx \in \{\aa,\bb\} 
                  \} \\ 
                  & \cup \\ 
                  \{ & 
                  e_{\xx,j}
                  &\mid& x \in \{\aa,\bb\}, ~~ 0 < j \leq  |\xx_i| 
                  \}; 
                  \\ 
                                          R_i = \{&                             (e_s,e_{st}), (e_s,e_{\xx}),  
                (e_{s},e_{\xx,lst}), 
                \\ 
                &
                (e_\xx,e_{\xx}), 
                (e_\xx,e_{\vbl{ntF}}), 
                (e_\xx,e_{\xx,lst}),\\ 
                &
                (e_{\xx,lst},e_{\vbl{ntF}}), 
                                \\ 
                &
                (e_{01,\xx},
                                e_{01,\xx}),
                                (e_{01,\xx},e_{\vbl{end}}), 
                                                (e_{01,\xx},e_{\vbl{ntF}})
                                &\mid &
                \xx \in \{\aa,\bb\} 
                \}\\
                & \cup \\ 
                \{ & 
                (e_{\xx,lst},e_{\xx,1}), 
                                                (e_{\xx,j},e_{\vbl{ntF}}), 
                (e_{\xx,|\xx_i|},e_{\vbl{end}}) &\mid &
                \xx \in \{\aa,\bb\}, ~~ 0 < j \leq |\xx_i|
                                \}\\ 
                \{ & 
                                                                (e_{\xx,j},e_{\xx,j+1})
                                                                                                                &\mid &
                x \in \{\aa,\bb\}, ~~ 0 < j  < |x_i| 
                \}\\ 
               & \cup \\  
                \{&
                  (e_{\xx,lst},e_{01,\xx}), (e_{\xx,lst},e_{\vbl{end}})
                  &\mid &   
                  x \in \{\aa,\bb\}, ~~  |x_i| = 0 
                \};
                          \end{align*}
            and the precondition for each event in $E_i$ is given below:
            \begin{align*}
                        pre(e_s) =~& \vbl{root} \land \bar{K} \vbl{stg1} 
                        &
                        pre(e_{st}) =~& \vbl{stg1} \\ 
                                    pre(e_\xx) =~& 
                                      \xx 
                          \land \neg \vbl{last}
                                                &
            pre(e_{\xx,lst}) =~& 
                          \xx 
                            \land \vbl{last}
            \\ 
                                    pre(e_{\vbl{ntF}}) =~& \vbl{ntF}
                        &
                                    pre(e_{\xx,j}) =~& \xx_i[j] \land \logicize{loop}_\xx
                                                                                                                                                \\ 
                                    pre(e_{\vbl{end}}) =& \vbl{end}
                        &
                                    pre(e_{01,\xx}) =& \logicize{loop}_\xx \lor \vbl{lp} 
                                                                                    ,  
                                \end{align*}
        for each $ \xx \in \{\aa,\bb\}$ and $ j \leq |\xx_i|$. 

        Action $\ad_i$ corresponds to adding a block at the end of a sequence of blocks. The precondition of $e_s$ ensures that $\bar{K}\vbl{stg1}$ is true at the epistemic state, and therefore we are still in the first stage. The action is illustrated in \Cref{fig:addK}.

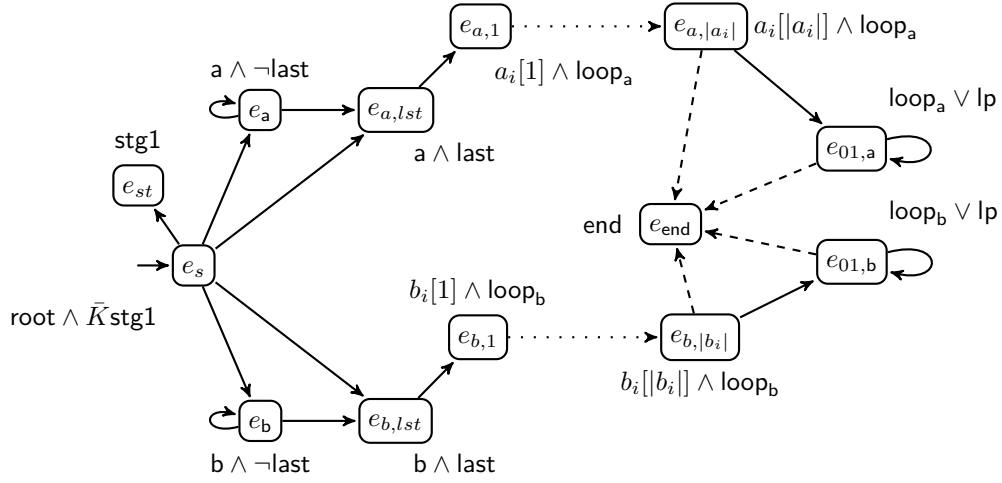
\begin{figure}
  \centering 
\begin{tikzpicture}[->,>=stealth',shorten >=1pt,auto,node distance=3cm, thick]

    \node[mynodesh,initial] (es)   {$e_s$};
  \node[mynodesh, above left=0.5cm and 0.1cm of es] (est)  {$e_{st}$};
  \node[mynodesh, above right=1.5cm and 0.3cm of es] (ea)    {$e_\aa$};
  \node[mynodesh, below right=1.5cm and 0.3cm of es] (eb)   {$e_\bb$};
  \node[mynode, right=1cm of ea] (ej0a)    {$e_{a,lst}$};
  \node[mynode, right=1cm of eb] (ej0b)   {$e_{b,lst}$};
  \node[mynode, above right=0.5cm and 0.2cm of ej0a] (ej1a)   {$e_{a,1}$};
  \node[mynode, above right=0.5cm and 0.2cm of ej0b] (ej1b)  {$e_{b,1}$};
  
  \node[mynode, right=2cm of ej1a] (ejka)   {$e_{a,|a_i|}$};
  \node[mynode, right=2cm of ej1b] (ejkb)  {$e_{b,|b_i|}$};

  \node[mynodesh, right=1cm of ejkb, yshift=2.5cm] (epsilon01a)  {$e_{01,\aa}$};
  \node[mynodesh, right=1cm of ejkb, yshift=1cm] (epsilon01b)  {$e_{01,\bb}$};

  \node[right=1cm of ejkb, yshift=1.5cm] (epsilon01) { };

  \node[mynodesh, left=1.5cm of epsilon01] (eend)  {$e_{\vbl{end}}$};

  \node[below left = 0.5cm of es] at (es)   {$\vbl{root} \land \bar{K}\vbl{stg1}$};
  \node[above = 0.3cm of est] at (est)  {$\vbl{stg1}$};
  \node[above = 0.3cm of ea] at (ea)   {$\aa \land \neg \vbl{last} 
      $};
  \node[below = 0.3cm of eb] at (eb)   {$\bb \land \neg \vbl{last} 
    $};
  \node[below right = 0.3cm and 0.1cm of ej0a] at (ej0a)    {$\aa \land \vbl{last} 
    $};
  \node[below right = 0.3cm and 0.1cm of ej0b] at (ej0b)   {$\bb \land \vbl{last} 
    $};
  \node[below right = 0.3cm and 0.05cm of ej1a] at (ej1a)   {$a_i[1] \land \logicize{loop}_\aa$};
  \node[above = 0.3cm of ej1b] at (ej1b)  {$b_i[1] \land \logicize{loop}_\bb$};
  
  \node[right = 0.5cm of ejka] at (ejka)   {$a_i[|a_i|] \land \logicize{loop}_\aa$};
  \node[below = 0.3cm of ejkb] at (ejkb)  {$b_i[|b_i|] \land \logicize{loop}_\bb$};

  \node[above right = 0.5cm of epsilon01a] at (epsilon01a)  {$\logicize{loop}_\aa \lor \vbl{lp}$};
  \node[above right = 0.5cm of epsilon01b] at (epsilon01b)  {$\logicize{loop}_\bb \lor \vbl{lp}$};

  \node[left = 0.5cm of eend] at (eend)  {$\vbl{end}$};

    \draw[->] (es) -- (est);
  \draw[->] (es) -- (ea); 
  \draw[->] (es) -- (eb);

  \draw[->] (ea) edge[loop left] (ea);
  \draw[->] (eb) edge[loop left] (eb);
  
  \draw[->] (es) -- (ej0a);
  \draw[->] (es) -- (ej0b);

  \draw[->] (ea) -- (ej0a);
  \draw[->] (eb) -- (ej0b);
  
  \draw[->] (ej0a) -- (ej1a);
  \draw[->] (ej0b) -- (ej1b);

  \draw[->, loosely dotted] (ej1a) -- (ejka);
  \draw[->, loosely dotted] (ej1b) -- (ejkb);
  
  \draw[->] (ejka) -- (epsilon01a);
  \draw[->] (ejkb) -- (epsilon01b);
  
  \draw[->] (epsilon01a) edge[loop right] (epsilon01a);
  \draw[->] (epsilon01b) edge[loop right] (epsilon01b);

  \draw[dashed] (ejka) -- (eend);
  \draw[dashed] (ejkb) -- (eend);
  \draw[dashed] (epsilon01a) -- (eend);
  \draw[dashed] (epsilon01b) -- (eend);

\end{tikzpicture}
\caption{The epistemic action $\ad_i$, where the designated event, $e_s$ is marked with an arrow. Each event's precondition appears next to the event. To avoid overloading the figure, we omit event $e_{\vbl{ntF}}$.}
\label{fig:addK}
\end{figure}

\begin{lemma}\label[lemma]{lem:adding-first-block}
  Let $1 \leq i \leq n$.  
  Then, $\initial \times \ad_i = s\block{\varepsilon}{\varepsilon}\{ 0,1 \} \times \ad_i$.
  \end{lemma}

  \begin{proof}
    Observe that $s\block{\varepsilon}{\varepsilon}\{ 0,1 \} $ is identical to $\initial$ except that it lacks state $w_{\vbl{empty}}$.
    Since there is not any event in $\ad_i$ whose precondition is satisfied at $w_{\vbl{empty}}$, the lemma follows.
  \end{proof}

\begin{lemma}\label[lemma]{lem:adding-blocks}
  Let $s = s\block{q_\aa }{q_\bb }\{ 0,1 \}$ and $1 \leq i \leq n$.  
  Then, $s \times \ad_i = s\block{q_\aa a_i }{q_\bb b_i }\{ 0,1 \}$.
\end{lemma}

\begin{proof}
    Let $((W,R,V),(w_{\vbl{root}},e_s)) = s \times \ad_i$. We show that $s \times \ad_i \sim s\block{q_\aa a_i }{q_\bb b_i }\{ 0,1 \}$.
  For simplicity, we refer to all accessibility relations as $R$.
                                                                                                                            
  Let $W_1 = \{w_{\vbl{root}}, w_{\vbl{stg1}}, w_{\vbl{ntF}}, w_{\vbl{end}}, w_{\xx} \mid \xx \in \{\aa,\bb\} \} \cup \{ w_{\xx,j} \mid \xx \in \{\aa,\bb\}, 0 < j \leq |q_\xx| \}$;
  and 
  let $W_2 (\xx) = \{ w_{\vbl{bt},\xx}, w_{\vbl{lp}} \}.$ We note that $W \block{q_\aa}{q_\bb} = W_1 \cup W_2$.

  Let $E_1 = \{ e_s, e_{st}, e_{\vbl{ntF}}, e_{\vbl{end}}, e_\xx, e_{\xx,lst} \mid \xx \in \{\aa,\bb\} \}$; 
  $E_2(\xx) = \{ e_{\xx,j} \mid  0 < j \leq |\xx_i|\}$; and 
  $E_3(\xx) = \{ e_{01,\xx}  \}$. $E_1, E_2(\aa), E_2(\bb), E_3(\aa), E_3(\bb)$ partition $E_i$.

  We observe that for every $w \in W\block{q_\aa}{q_\bb}$ and $e \in E_i$: 
  \begin{enumerate}
    \item if $w \models \pre(e)$ and $e \in E_1$, then $w \in W_1$;
    \item if $w \models \pre(e)$ and $e \in E_2(\xx)$, then $w \in W_2(\xx)$ (specifically, $w = w_{0,\xx}$ or $w = w_{1,\xx}$);
    \item if $w \models \pre(e)$ and $e \in E_3(\xx)$, then $w \in W_2(\xx)$.
    \item if $w \in W_1$, then there exists a unique $e(w) \in E$, such that $w \models \pre(e(w))$, and furthermore, $\pre(w) \in E_1$ for every $w \in W_1$;
    \item if $e \in E_2 (\xx)$, then there exists a unique $w(e) \models \pre(e)$, and furthermore, $w(e) \in W_2(\xx)$;
    \item if $w \in W_2(\xx)$ and $e \in E_3(\xx)$, then $w \models \pre(e)$.
  \end{enumerate}
  Let $WE_1 = \{ (w,e(w)) \mid w \in W_1 \}$; $WE_2(\xx) = \{ (w(e),e) \mid e \in E_2(\xx) \}$; and $WE_3(\xx) = W_2(\xx) \times E_3(\xx)$.
  From items 1, 2, 3, we see that the states of $s \times \ad_i$ are exactly $WE_1 \cup WE_2(\aa) \cup WE_2(\bb) \cup WE_3(\aa) \cup WE_3(\bb)$.
  
  It is not hard to observe that $(s \times \ad_i |_{WE_1}, (w,e(w)) \sim (s |_{W_1}, w)$ for every $w \in W_1$ (the two structures are actually isomorphic through the projection relation). But then, $s |_{W_1} = s \block{q_\aa \aa_i}{q_\bb \bb_i}\{ 0,1 \} |_{W_1}$, so 
  \[
  (s \times \ad_i |_{WE_1}, (w,e(w)) \sim (s \block{q_\aa \aa_i}{q_\bb \bb_i}\{ 0,1 \} |_{W_1}, w) 
  .\]

  We also see that for every $(w(e_{\xx,j}),e_{\xx,j}) \in WE_2(\xx)$, 
  \[
  (s \times \ad_i |_{WE_2(\xx)}, (w(e_{\xx,j}),e_{\xx,j})) \sim (s \block{q_\aa \aa_i}{q_\bb \bb_i}\{ 0,1 \} |_{\{ w_{\xx,|q_\xx|+h} \mid 0 < h \leq |\xx_i|\}},  w_{\xx,|q_\xx|+j}).
  \]
  Let $w_\times(e_{\xx,j}) = w_{\xx,|q_\xx|+j}$ for each $e_{\xx,j} \in E_2(\xx)$.

  Finally, 
  \[(s \times \ad_i |_{WE_3(\xx)},(w,e_{01,\xx}))) \sim (s |_{W_2(\xx)},w)  = (s \block{q_\aa \aa_i}{q_\bb \bb_i}\{ 0,1 \} |_{W_2(\xx)},w)
  \] 
  for every $w \in W_2(\xx)$.

  Let $\R = \{  
    ((w,e(w)) , w ) \mid w \in W_1
  \}
  \cup
  \{
    ((w(e),e),w_\times(e)) \mid e \in E_2(\xx)  , ~ \xx \in \{\aa,\bb\}
  \}
  \cup 
  \{
    ((w,e_{01,\xx}),w) \mid w \in W_2(\xx), ~ \xx \in \{\aa,\bb\}
  \}$.

  To verify that $\R$ is a bisimulation relation, from the above, it suffices to verify the accessibility pairs that cross between the sets 
  $WE_1 , WE_2(\aa) , WE_2(\bb) , WE_3(\aa) , WE_3(\bb)$ on $s \times \ad_i$, and the accessibility pairs that cross between the sets 
  $W_1 , W_2(\aa) , W_2(\bb), \{ w_{\xx,|q_\xx|+h} \mid 0 < h \leq |\xx_i|\}$ on $s \block{q_\aa \aa_i}{q_\bb \bb_i}\{0,1\}$.
  But these are few and easy to verify. Observe that $\R$ is a function.

  \begin{description}
  \item [Let $(w,e(w)) \in WE_1$.] From the above, we have that $(w,e(w)) \R w$. If $(w,e(w)) R (u,v) \notin WE_1$, then $w R u \notin W_1$ and $w \in W_1$, so $u = w_{\vbl{bt},\xx}$ for some $\xx \in \{\aa,\bb\}$ and $\vbl{bt} \in \{0,1\}$. Then, $e(w) = e_{\xx,lst}$ and $v = e_{\xx,1}$ or $v = e_{01,\xx}$.
  In the first case, $(u,v) \R w_{\xx,|q_\xx|+1}$ and $w = w_{\xx,|q_\xx|} R w_{\xx,|q_\xx|+1}$ in $s \block{q_\aa \aa_i}{q_\bb \bb_i}\{0,1\}$;
  in the second case, $(u,v) \R u$, $\xx_i = \varepsilon$, and $w R u$ in $s \block{q_\aa \aa_i}{q_\bb \bb_i}\{0,1\}$ (as in $s \block{q_\aa}{q_\bb}\{0,1\}$).

  On the other hand, if $w R u \notin W_1$ in $s \block{q_\aa \aa_i}{q_\bb \bb_i}$, then $u = w_{\vbl{bt},\xx}$ and $\xx_i = \varepsilon$, for some $\xx \in \{\aa,\bb\}$ and $\vbl{bt} \in \{0,1\}$; or $u = w_{\xx,|q_\xx|+1}$ for some $\xx \in \{\aa,\bb\}$.
  In the first case, let $(u',v') = (w_{\vbl{bt},\xx},e_{01,\xx})$; in the second case, let $(u',v') = (w(e_{\xx,1}),e_{\xx,1})$.
  In both cases, it is not hard to see that $(w,e(w)) R (u',v') \R u$.

  \item [Let $(w(e),e) \in WE_2(\xx)$.] From the above, we have that $(w,e(w)) \R w_\times(e)$. If $(w(e),e) R (u,v) \notin WE_2(\xx)$, then: 
    \begin{enumerate}
      \item $e = e_{\xx,|\xx_i|}$, $w(e) = w_{01,\xx}$, and $(u,v) = (w_{01,\xx},e_{01,\xx}) \R u$, in which case we also have $w_\times = w_{\xx,|q_\xx|+|\xx_i|} R u$ in $s \block{q_\aa \aa_i}{q_\bb \bb_i}\{0,1\}$; or
      \item $(u,v) = (w_{\vbl{p}},e_{\vbl{p}})$, where $p \in \{ \vbl{end}, \vbl{ntF}\}$, in which case 
            we also have that $w_\times(e) R u$ and $(u,v) \R u$.
    \end{enumerate}
    On the other hand, if $w_\times(e) R u \notin WE_2(\xx)$, then either 
    \begin{enumerate}
      \item $u = w_{01,\xx}$, in which case $w_\times = w_{\xx,|q_\xx|+|\xx_i|}$, $e = e_{\xx,|\xx_i|}$, and $w(e) = w_{01,\xx}$, so let $v = e_{01,\xx}$, and therefore, $(w(e),e) R (u,v) = (w_{01,\xx},e_{01,\xx}) \R u$; or
      \item $u = w_{\vbl{p}}$, where $p \in \{ \vbl{end}, \vbl{ntF}\}$, in which case let  $(u,v) = (w_{\vbl{p}},e_{\vbl{p}})$, so
            we also have that $(w(e),e) R (u,v) \R u$.
    \end{enumerate}

  \item [Let $(w,e_{01,\xx}) \in WE_3(\xx)$.] From the above, we have that $(w,e_{01,\xx}) \R w$. 
        If $(w,e_{01,\xx}) R (u,v) \notin WE_3(\xx)$, then it can only be the case that 
        $(u,v) = (w_{\vbl{p}},e_{\vbl{p}})$, where $p \in \{ \vbl{end}, \vbl{ntF}\}$, and therefore, 
                we have that  $w R u$ and $(u,v) \R u$.
        Similarly, if $w R u \notin W_2(\xx)$, then   it can only be the case that $u = w_{\vbl{p}}$, where $p \in \{ \vbl{end}, \vbl{ntF}\}$, in which case let  $(u,v) = (w_{\vbl{p}},e_{\vbl{p}})$, so
            we also have that  $(w,e_{01,\xx}) R (u,v) \R u$.
  \end{description}

We can conclude that $\R$ is a bisimulation, and therefore the proof of the lemma is complete.
\end{proof}

\begin{corollary}\label[corollary]{cor:adding-blocks}
  Let $\alpha = \ad_{i_1} \ad_{i_2} \cdots \ad_{i_k} \in \{ \ad_i \mid 1 \leq i \leq n \}^*$. 
  Then, $\initial \times \alpha = s\block{a_{i_1} a_{i_2} \cdots a_{i_k}}{b_{i_1} b_{i_2} \cdots b_{i_k}}\{ 0,1\}$.
\end{corollary}

        \subsubsection{Moving to the Second Stage}
        We define epistemic action 
        \[\nxtstg = 
                (\{e_{nx}
                \}, \{(e_{nx},e_{nx})
                \}, e_{nx})
        ,
        \text{ where }
        \]
            \begin{align*}
      pre(e_{nx}) &= 
        (K \neg \vbl{empty} \land \bar{K} \vbl{stg1}) \lor 
                \left( 
        (\aa \lor \bb) \land K \neg \vbl{lp} \right)       \lor \vbl{ntF}
                            \end{align*}
                                                                
    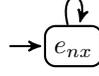
\begin{figure}
      \centering 
    \begin{tikzpicture}[->,>=stealth',shorten >=1pt,auto,node distance=2.8cm, thick]
    
            \node[mynodesh,initial] (es)   at (0, 0)      {$e_{nx}$};

            \draw[->] (es) edge [loop above] (es);
                        
    \end{tikzpicture}
    \caption{The epistemic action $\nxtstg$ corresponds to moving the plan to the second stage.}
    \end{figure}

        Action $\nxtstg$ removes the states $w_{\vbl{stg1}}, w_{\vbl{end}}, w_{0}(\xx), w_{1}(\xx), w_{\vbl{lp}}$ from $s\block{q_\aa}{q_\bb}\{0,1\}$, 
                initiating the second stage of the plan. 

        \begin{lemma}\label[lemma]{lem:next-stage_short}
          Let $s = s\block{q_\aa }{q_\bb }\{ 0,1 \}$.  
          Then, $s \times \nxtstg = s\block{q_\aa }{q_\bb }$.
        \end{lemma}

        \begin{proof}
          Simply observe that the states left in $s\block{q_\aa }{q_\bb }$ are exactly the ones that satisfy the precondition of $e_{nx}$.
                                                                                                                                                                                                                                                                                                                                                                                                                                                                                                                                                                                                      \end{proof}
        
        Similarly, it is straightforward to see that the following lemma holds.

        \begin{lemma}\label[lemma]{lem:next-stage_no_addK}
          Let $s$ be an epistemic state where $\nxtstg$ applies.  
          Then, $s \times \nxtstg \models K \neg \vbl{stg1}$.
        \end{lemma}

        We call the epistemic states of the form $\initial \times \alpha$, where $\alpha \in \{\ad_i \mid 1 \leq i \leq n\}^*$, first-stage states. 

\begin{lemma}\label[lemma]{lem:K-is-preservedK}
  Let $s$ be an epistemic state, $f$ an epistemic action, and $K \varphi$ an epistemic formula of modal depth $1$.
  If $s \models K \varphi$, then $s \times f \models K \varphi$.
\end{lemma}

\begin{proof}
  Let $s = ((W,R,V),w)$ and $f = (E,R',e)$. It suffices to prove that for every $w R w'$ and $e R' e'$, $(s\times f,(w',e')) \models \varphi$. But $\varphi$ is propositional and therefore,  $(s\times f,(w',e')) \models \varphi$ if and only if $((W,R,V),w') \models \varphi$. The fact that $w R w'$ and $s \models K \varphi$ yield that $((W,R,V),w') \models \varphi$ and the proof is complete. 
\end{proof}

\begin{lemma}
  For every epistemic state $s 
    $ and $i \leq n$, neither $\ad_i$ nor $\nxtstg$ apply to $s \times \nxtstg$. Furthermore, if $\ad_i$ (and $\nxtstg$) does not apply to $s$, then for every epistemic action $f$ that applies to $s$, neither $\ad_i$ nor $\nxtstg$ apply to $s \times f$.
\end{lemma}

\begin{proof}
  Observe that 
  $\ad_i$ does not apply to an epistemic state $s'$ if and only if 
  $s' \models K \neg \vbl{stg1}$.
  Furthermore, by \Cref{lem:next-stage_no_addK}, 
  $s \times \nxtstg  \models K \neg \vbl{stg1}$, which implies the first statement.
  The second statement is a consequence of \Cref{lem:K-is-preservedK}.  
\end{proof}

\begin{corollary}
  Let 
  $s = \initial \times \alpha$, where $\alpha$ is a sequence of epistemic actions 
    that includes $\nxtstg$. Then, neither $\nxtstg$ nor $\ad_i$ 
    apply to $s$. 
    \end{corollary}

Therefore, epistemic action $\nxtstg$ can only be applied once in a plan, and it marks the end of the applications of the actions $\ad_i$.

        \subsubsection{Verifying one Symbol at Each Step}
        Finally, for each $\vbl{bt} \in \{0,1\}$, we define the epistemic action 
        \[ \remv_{\vbl{bt}} = (E_{\vbl{bt}},R_{\vbl{bt}},e^{\vbl{bt}}), ~~\text{ where}  \]
        \begin{align*}
                E_{\vbl{bt}} =& \{ 
                                                  e^{\vbl{bt}}
                    ,  e^{\vbl{bt}}_{\vbl{ntF}}, e^{\vbl{bt}}_{\text{fail}} \};\\ 
                        R_{\vbl{bt}} =& \{ 
                                                            (e^{\vbl{bt}},e^{\vbl{bt}}), 
                                                                                    (e^{\vbl{bt}},e_{\text{fail}}^{\vbl{bt}}),
            (e^{\vbl{bt}},e_{\vbl{ntF}}^{\vbl{bt}})
                                                \};\\
                                                    pre(e^{\vbl{bt}}) =&  
        (\vbl{root} \land K \neg \vbl{stg1}) \lor 
        ( (\aa \lor \bb 
                ) \land \neg \logicize{tail} )
                ;
        \\ 
                                                                                                        pre(e_{\text{fail}}^{\vbl{bt}}) =& 
          (\logicize{tail} \land \neg \vbl{bt}
        ) 
        \lor 
        \vbl{failed}
                                                        ; ~~~ \text{ and}
                \\ 
        pre(e_{\vbl{ntF}}^{\vbl{bt}}) =& \vbl{ntF}
        .
            \end{align*}

        For each $\vbl{bt} \in \{0,1\}$, $\remv_{\vbl{bt}}$ only activates during the second stage and if the epistemic state is of the 
    right form, \emph{i.e.}
        $s\block{q_\aa\vbl{bt}}{q_\bb\vbl{bt}}$, the resulting state is $s\block{q_\aa}{q_\bb}$. If the epistemic state is not of the right form, then $e_{\text{fail}}^{\vbl{bt}}$ activates and ensures that a last state can no longer access $w_{\vbl{ntF}}$, thus marking the plan as a failed one.

\begin{figure}
  \centering 
\begin{tikzpicture}[->,>=stealth',shorten >=1pt,auto,node distance=2.8cm, thick]

      \node[mynode,initial] (eloop)  at (3, 0)      {$e^{\vbl{bt}}$};
  \node[mynode] (efail)   at (6, 1)    {$e_{\text{fail}}^{\vbl{bt}}$};
  \node[mynode] (entf)   at (6, -1)    {$e_{\vbl{ntF}}^{\vbl{bt}}$};

  \node (pre-es)   at (0, -0.6)      {$(\vbl{root} \land K \neg \vbl{stg1}) \lor 
  ((\aa \lor \bb 
    ) \land \neg \logicize{tail})
  $};
        \node (pre-ea)   at (8, 0.4)    {$(\logicize{tail} \land \neg \vbl{bt}
        ) 
        \lor 
          \vbl{failed}
                                                  $};
  \node (pre-fl)   at (6.3, -1.5)    {$\vbl{ntF}
  $};
  
      \draw[->] (eloop) edge [loop above] (eloop);
  \draw[->] (eloop) -- (efail);
  \draw[->] (eloop) -- (entf);
    
\end{tikzpicture}
\caption{The epistemic action $\remv_{\vbl{bt}}$ corresponds to removing the bit $\vbl{bt}$ from the end of the encoding.}
\end{figure}
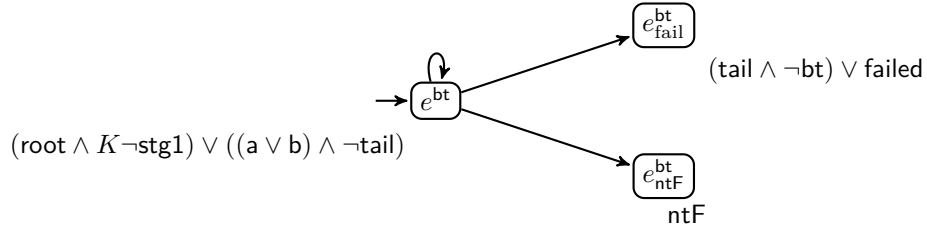

\begin{lemma}\label[lemma]{lem:remove-only-at-the-end}
  Let $s$ be an epistemic state where $\remv_{\vbl{bt}}$ applies.  
  Then, 
    neither 
  $\ad_i$ nor $\nxtstg$ applies to $s$.
\end{lemma}
\begin{proof}
  The lemma results from the observation that the precondition of $e^{\vbl{bt}}
    $ is consistent with neither of the preconditions for $e_s$ and $e_{nx}$.
\end{proof}

\begin{lemma}\label[lemma]{lem:removing-bits-successfully}
  Let $q_\aa, q_\bb \neq \varepsilon$ and
                    let $s = s\block{q_\aa \vbl{bt}}{q_\bb \vbl{bt}}$. 
    Then, $s \times \remv_{\vbl{bt}} = s\block{q_\aa }{q_\bb }$.
  \end{lemma}

\begin{proof}
      We first observe that the precondition of $e^{\vbl{bt}}_{\text{fail}}$ is not satisfied by any state in $s\block{q_\aa \vbl{bt}}{q_\bb \vbl{bt}}$; therefore, $e^{\vbl{bt}}_{\text{fail}}$ does not activate in the product update.
  Furthermore, the precondition of $e^{\vbl{bt}}_{\vbl{ntF}}$ is satisfied by $w_{\vbl{ntF}}$ only.
  Finally, every other state in $s\block{q_\aa \vbl{bt}}{q_\bb \vbl{bt}}$, except for $w_{\vbl{ntF}}$, $w_{\aa,|q_\aa|+1}$, $w_{\bb,|q_\bb|+1}$, satisfies the precondition of $e^{\vbl{bt}}$. 
  Let $((W,R,V),(w_{\vbl{root}},e_s^{\vbl{bt}})) = s \times \remv_{\vbl{bt}}$. We show that $s \times \remv_{\vbl{bt}} \sim s\block{q_\aa }{q_\bb }$.
  Let 
    \begin{align*}
      \R = \{& 
      \left((w_{\vbl{root}},e^{\vbl{bt}}) , w_{\vbl{root}} \right), \\
      & \left((w_{\xx},e^{\vbl{bt}}) , w_{\xx} \right), \\
      & \left((w_{\xx,j},e^{\vbl{bt}}) , w_{\xx,j} \right), \\
                  & \left((w_{\vbl{ntF}},e^{\vbl{ntF}}) , w_{\vbl{ntF}} \right)
                                                                  ~~ &
      \mid &~~ 
                  0 < j \leq |q_\xx|, ~ x \in \{\aa,\bb\}
            \}
  \end{align*}
  It is now straightforward to verify that $\R$ is a bisimulation
  --- in fact, $\R$ is an isomorphism between the two models.
        \end{proof}

\begin{definition}\label[definition]{def:failed-state}
Let $s = ((W,R,V),s_0)$ be an epistemic state. We say that $s$ is a \emph{failed state} when there is a sequence of states $s_0 R s_1 R \cdots R s_k$, where $k \geq 1$, such that: 
\begin{enumerate}
  \item
$((W,R,V),s_0) \models \vbl{root} \land K \neg \vbl{stg1}$;
  \item
$((W,R,V),s_1) \models \xx $ for some $\xx \in \{\aa,\bb\}$;
  \item  
  $((W,R,V),s_i) \models \vbl{symb} $ for every $2 \leq i \leq k$;
and 
  \item
$((W,R,V),s_k) \models 
\vbl{failed}$.
\end{enumerate}
\end{definition}

A failed state is the result of applying $\remv_{\vbl{bt}}$ to an epistemic state that is not of the form $s\block{q_\aa \vbl{bt}}{q_\bb \vbl{bt}}$.
In that case, the action $\remv_{\vbl{bt}}$ marks that it failed to remove that last bit from (at least) one of the state's branches. Furthermore, as the following  \Cref{lem:removing-bt-when-hash} demonstrates, once an epistemic state has been marked as a failed one, it will remain marked as such throughout the application of the plan.

\begin{lemma}\label[lemma]{lem:removing-bt-when-hash}
  If $\remv_{\vbl{bt}}$ applies to epistemic state $s$, and
  \begin{itemize}
    \item 
      $s = s\block{q_\aa}{q_\bb}$ and 
            $q_\aa = \varepsilon$ or $q_\bb = \varepsilon$; 
            \item 
      $s = s\block{q_\aa \vbl{bt}'}{q_\bb}$ or $s = s\block{q_\aa }{q_\bb \vbl{bt}'}$, and 
    $\vbl{bt}' \neq \vbl{bt}$; or  
    \item 
      $s$ is a failed state;
  \end{itemize}
  then 
  $s \times \remv_{\vbl{bt}}$ is a failed state.
\end{lemma}

\begin{proof}
  We prove the lemma for the case where $s = s\block{q_\aa \vbl{bt}'}{q_\bb}$, as the other cases are similar.
  There exist some states $w_{\vbl{root}} = r_0 ~ R ~ w_{\aa} = r_1 ~R~\cdots ~R~ r_k = w_{a,|q_\aa|+1}$ in $s$.
  We then have that in $s$:
  \begin{enumerate}
    \item 
  $r_0 \models \vbl{root} \land K \neg \vbl{stg1}$;
    \item
  $r_1 \models \aa $;
    \item
  $r_i \models \vbl{symb} $ for every $2 \leq i \leq k$; and 
    \item
  $r_k \models 
   \vbl{tail} \land \vbl{bt}'$ in $s$.
  \end{enumerate}
  It is also straightforward to see that for each $0<i<k$, $r_i \models \neg \vbl{tail}$.
  Therefore, it is not hard to see that $s \times \remv_{\vbl{bt}}$ contains the sequence of states
  \[  
    (r_0,e^{\vbl{bt}}) ~ R ~ (r_1,e^{\vbl{bt}})  ~R~\cdots ~R~ (r_k,e^{\vbl{bt}}_{\text{fail}}) 
  \]
  that satisfy the conditions in \Cref{def:failed-state} for $s \times \remv_{\vbl{bt}}$ to be a failed state.
  \end{proof}

\begin{corollary}\label[corollary]{cor:plan-sequence}
  Let $\alpha$ be a plan (a sequence of epistemic actions) that applies to epistemic state $s_I$.
  Then, $\alpha$ is of the form $\ad_{i_1}\ad_{i_2}\cdots \ad_{i_k}$ or 
  \[\ad_{i_1}\ad_{i_2}\cdots \ad_{i_k} ~\nxtstg~ \remv_{\vbl{bt}_1}\remv_{\vbl{bt}_2}\cdots \remv_{\vbl{bt}_l}.\]
\end{corollary}
\begin{proof}
  The corollary is a consequence of \Cref{lem:next-stage_no_addK,lem:remove-only-at-the-end,lem:adding-blocks,lem:removing-bits-successfully,lem:removing-bt-when-hash}.
\end{proof}

\subsection{Completing the Reduction Plan}

To complete the definition of $EP$, let 
\begin{align*}
  \act =& \{ \ad_i, \nxtstg, \remv^{\vbl{bt}} \mid 1 \leq i \leq n \text{ and } \vbl{bt} \in \{ 0,1 \}  \}, \quad \text{ and}\\
  \varPhi_G =& K \neg \vbl{empty} \land K 
    \left(
        \vbl{tail} \land 
    \neg \logicize{failed}
            \right)
    .
\end{align*}

\begin{lemma}\label[lemma]{lem:failed-states-fail}
  Let $s$ be a failed epistemic state. Then, $s \not \models \varPhi_G$.
\end{lemma}
\begin{proof}
  A consequence of the definitions of a failed state (\Cref{def:failed-state}) and of $\varPhi_G$.
\end{proof}

\begin{lemma}\label[lemma]{lem:successfully-removing-yields-match}
  Let $q_\aa, q_\bb \in \{0,1\}^*$ and let $rm = \remv_{\vbl{bt}_l}\remv_{\vbl{bt}_{l-1}}\cdots \remv_{\vbl{bt}_1}$.
  If $s\block{q_\aa}{q_\bb} \times rm \models \varPhi_G$, then $q_\aa = q_\bb = \vbl{bt}_1\vbl{bt}_2\cdots \vbl{bt}_l$. 
        \end{lemma}
\begin{proof}
  We proceed by induction on $l$.
    If $rm = \varepsilon$, then $s\block{q_\aa}{q_\bb} = s\block{q_\aa}{q_\bb} \times rm \models \varPhi_G$, and therefore, $q_\aa = q_\bb = \varepsilon$.
  If $rm = \remv_{\vbl{bt}_1} rm' $, then, by \Cref{lem:removing-bits-successfully,lem:removing-bt-when-hash}, either $s\block{q_\aa}{q_\bb} \times \remv_{\vbl{bt}_1}$ is a failed state, or it is $s\block{q_\aa'}{q_\bb'}$, where $q_\aa = q_\aa' \vbl{bt}_l$ and $q_\bb = q_\bb' \vbl{bt}_l$.
    The first case results in a contradiction, due to \Cref{lem:removing-bt-when-hash,lem:failed-states-fail}.
  We conclude the proof by the inductive hypothesis on $l-1$.
      \end{proof}

We can conclude with our main theorem.

\begin{theorem}\label[theorem]{thm:reduction-K}
  $B$ has a match if and only if $EP$ has a plan.
\end{theorem}
\begin{proof}
  To prove the first direction, let $i_1, i_2, \ldots, i_k$ be a match for $B$. Let 
  $m_1 m_2 \cdots m_\ell = a_1 a_2 \cdots a_k = b_1 b_2 \cdots b_k$, where for every $1 \leq i \leq \ell$, $m_i \in \{0,1\}$. 
  Let 
  \begin{align*}
  \ad_m &= \ad_{i_1} \ad_{i_2} \cdots \ad_{i_k};\\ 
  \remv_m &=  \remv_{m_\ell}  
    \cdots  \remv_{m_1}; \text{ and} \\ 
  \alpha &= \ad_m ~\nxtstg~ \remv_{m}.
  \end{align*}
  We now demonstrate that $\alpha$ is a plan that applies to $s_I$ and that $s_{I} \times \alpha \models \varPhi_G$.

  By \Cref{cor:adding-blocks}, $s_{I} \times \ad_m = s\block{m_1 m_2 \cdots m_\ell}{m_1 m_2 \cdots m_\ell} \{ 0,1 \}$; then, \Cref{lem:next-stage_short} yields that $s_{I} \times \ad_m ~\nxtstg =  s\block{m_1 m_2 \cdots m_\ell}{m_1 m_2 \cdots m_\ell} \{ 0,1 \} \times \nxtstg = s\block{m_1 m_2 \cdots m_\ell}{m_1 m_2 \cdots m_\ell}$; and finally, \Cref{lem:removing-bits-successfully} yields that $s_I \times \alpha = s\block{m_1 m_2 \cdots m_\ell}{m_1 m_2 \cdots m_\ell} \times \remv_{m} = s\block{\varepsilon}{\varepsilon}\models \varphi_G$.

  For the converse direction, let $\alpha$ be a plan that applies to $s_I$, such that $s_I \times \alpha \models \varPhi_G$.
  By \Cref{cor:plan-sequence}, $\alpha$ is of the form 
  $\ad_{i_1}\ad_{i_2}\cdots \ad_{i_k}$ or 
  \[\ad_{i_1}\ad_{i_2}\cdots \ad_{i_k} ~\nxtstg~ \remv_{\vbl{bt}_l}\remv_{\vbl{bt}_{l-1}}\cdots \remv_{\vbl{bt}_1}.\]
  If $\alpha$ is of the form 
  $\ad_{i_1}\ad_{i_2}\cdots \ad_{i_k}$, then $s_I \times \alpha = s\block{a_{i_1}\cdots a_{i_k}}{b_{i_1}\cdots b_{i_k}}\{0,1\} \not \models \varPhi_G$, contradicting our assumptions.
  Therefore, $\alpha$ is of the form 
   \[\ad_{i_1}\ad_{i_2}\cdots \ad_{i_k} ~\nxtstg~ \remv_{\vbl{bt}_l}\remv_{\vbl{bt}_{l-1}}\cdots \remv_{\vbl{bt}_1}.\]
   We prove that $i_1,i_2,\ldots,i_k$ is a match for $B$. 
     Let 
    $m_1 m_2 \cdots m_\ell = a_1 a_2 \cdots a_k$ and 
    $m'_1 m'_2 \cdots m'_{\ell'} = b_1 b_2 \cdots b_k$.
    By \Cref{cor:adding-blocks}, $s_{I} \times \ad_{i_1}\ad_{i_2}\cdots \ad_{i_k} = s\block{m_1 m_2 \cdots m_\ell}{m'_1 m_2 \cdots m_{\ell'}} \{ 0,1\}$; then, \Cref{lem:next-stage_short} yields that 
    \[
      s_{I} \times \ad_{i_1}\ad_{i_2}\cdots \ad_{i_k} ~\nxtstg =  s\block{m_1 m_2 \cdots m_\ell}{m'_1 m'_2 \cdots m'_{\ell'}} \{ 0,1,\# \} \times \nxtstg = s\block{m_1 m_2 \cdots m_\ell}{m'_1 m'_2 \cdots m'_{\ell'}}.
      \]
     Finally, by \Cref{lem:successfully-removing-yields-match}, we conclude that $a_1 a_2 \cdots a_k = m_1 m_2 \cdots m_\ell = m'_1 m'_2 \cdots m'_{\ell'} = b_1 b_2 \cdots b_k$, and therefore $i_1,i_2,\ldots,i_k$ is a match for $B$. 
\end{proof}

\begin{corollary}\label[corollary]{cor:planning-undec-K}
  The plan existence problem for arbitrary frames is undecidable, even for the case of a single agent.
\end{corollary}

\section{The Case of Multiple Agents}
\label{sec:multi}

For the case of multiple agents, regardless of the constraints on the accessibility relations, we can similarly prove that the plan existence problem is undecidable. To adjust our reduction, in our construction we alternate between the accessibility relations of two agents, $1$ and $2$, to avoid any side-effects due to the constraints on the accessibility relations, and to give direction to our models when symmetric accessibility relations do not.
For the sake of simplicity, we assume that $1$ and $2$ are the only agents and that accessibility relations are equivalence relations, but the arguments are the same for any number of agents, and for more relaxed accessibility relation conditions.

We follow the same approach as in \Cref{sec:K}, though we note that 
due to the requirements that the accessibility relations are symmetric and transitive, simply closing the epistemic states and actions with respect to symmetry and transitivity would have unintended consequences, as some states  --- for example, states $w_{\vbl{ntF}}$ and $w_{\vbl{end}}$ --- are accessible from multiple states that should not be connected. 
Therefore, we will use multiple copies of these states to avoid unintended connections. We note that in the case of accessibility relations with no closure conditions, we can ensure that states $w_{\vbl{ntF}}$ and $w_{\vbl{end}}$ have no accessible states, and therefore keeping one or multiple copies of these results in bisimilar models. However, when the accessibility relation is an equivalence relation, that is no longer the case.

We modify the definition of the epistemic states and actions in the following way. We use an additional propositional variable $\#$ that acts as a separator symbol between two bits in a block, ensuring an even number of transitions between two bit states.
The initial state is relatively simple --- the epistemic actions will ensure to unfold the epistemic states and alternate the accessibility relations from agent 1 to agent 2 and back.

As illustrated in \Cref{fig:init-state-multi}, the initial epistemic state $\initial = ((W_0,(R^1_0,R_0^2),V_0),w_{\vbl{root}})$, where
\begin{align*}
    W_0 = \{ &
    w_{p} \mid p \in {\{ \vbl{root}, \vbl{empty}, \vbl{stg1}, \aa, \bb\}}\} \cup
    \{ w_{p}(\xx) \mid p \in {\{ \vbl{0}, \vbl{1}, \vbl{\#}, \vbl{ntF}, \vbl{end}
        \}}, ~ \xx \in {\{\aa,\bb\}} \};\\
    R_0^1 = \{ &
      w_{\vbl{root}},w_{\vbl{empty}}, w_{\vbl{stg1}}, w_{\aa}, w_{\bb}
      \}^2 
      \cup
      \{
        w_{\vbl{0}}(\aa), w_{\vbl{1}}(\aa), w_{\vbl{\#}}(\aa), w_{\vbl{ntF}}(\aa),w_{\vbl{end}}(\aa)       \}^2 \\ &
      \cup
      \{
        w_{\vbl{0}}(\bb), w_{\vbl{1}}(\bb), w_{\vbl{\#}}(\bb), w_{\vbl{ntF}}(\bb),w_{\vbl{end}}(\bb)
              \}^2
                                                                                                                                                          ;     \\ 
    R_0^2 = 
    \{  &
    w_{\aa}, 
        w_{\vbl{0}}(\aa), w_{\vbl{1}}(\aa), w_{\vbl{\#}}(\aa), w_{\vbl{ntF}}(\aa),w_{\vbl{end}}(\aa)
    \}^2 \\ &
    \cup
    \{
    w_{\bb}, 
        w_{\vbl{0}}(\bb), w_{\vbl{1}}(\bb), w_{\vbl{\#}}(\bb), w_{\vbl{ntF}}(\bb),w_{\vbl{end}}(\bb)
      \}^2 \cup \{ w_\vbl{root} \}^2 \cup \{ w_\vbl{empty} \}^2 \cup \{ w_\vbl{stg1} \}^2
                                                                                                                                                                                  \end{align*}
    and 
    for every $p \in \{ \vbl{root}, \vbl{empty}, \vbl{stg1}, \aa, \bb\}\}$, $V_0(w_p) = \{p\}$; and for every $p \in \{ \vbl{0}, \vbl{1}, \vbl{\#}, \vbl{end}, \vbl{ntF}
        \}$ and $\xx \in \{\aa,\bb\}$, $V_0(w_p(\xx)) = \{p,\xx\}$.

\begin{figure}
  \centering 
\begin{tikzpicture}[->,>=stealth',shorten >=1pt,auto,node distance=2.5cm, thick]

  \node[mynode,initial] (r0)   at (0, 0)    {$w_{\vbl{root}}$};
  \node[mynode] (stg)   [above = 1 of r0]    {$w_{\vbl{stg1}}$};
  \node[mynode] (empty)  [below= 1 of r0]    {$w_{\vbl{empty}}$};
  \node[mynodesh] (wa)   [above right = 0.5 and 1 of r0]  {$w_{\aa}$};
  \node[mynodesh] (wb)   [below right = 0.5 and 1 of r0] {$w_{\bb}$};
  \node[mynode] (w0a)   [above right = 0.5 and 1 of wa]  {$w_0(\aa)$};
  \node[mynode] (w1a)   [right = 1 of wa] {$w_1(\aa)$};
  \node[mynode] (whasha) [right = 1 of w0a]   {$w_{\#}(\aa)$};
  \node[mynode] (wenda) [right = 1 of w1a]   {$w_{\vbl{end}}(\aa)$};
  \node[mynode] (wnfa) [right = 1 of wenda]   {$w_{\vbl{ntF}}(\aa)$};
  
  \node[mynode] (w0b)   [right = 1 of wb]  {$w_0(\bb)$};
  \node[mynode] (w1b)   [below right = 0.5 and 1 of wb] {$w_1(\bb)$};
  \node[mynode] (whashb) [right = 1 of w0b]   {$w_{\#}(\bb)$};
  \node[mynode] (wendb) [right = 1 of w1b]   {$w_{\vbl{end}}(\bb)$};
  \node[mynode] (wnfb) [right = 1 of wendb]   {$w_{\vbl{ntF}}(\bb)$};

    \node[inner sep=10pt,draw,dashed,rounded corners,fit=(r0) (stg) (empty) (wa) (wb) ] {};
  \node (12a) [rounded corners,fit=(w0a) (w1a) (whasha) (wenda) (wnfa)] {} ;
  \node [draw,dashed,rounded corners,fit=(w0a) (w1a) (whasha) (wenda) (wnfa) 
    ] {} ;
  \node[draw,rounded corners,fit= (wa) (w0a) (w1a) (whasha) (wenda) (wnfa) (12a) ] {};
  \node (12b) [rounded corners,fit=(w0b) (w1b) (whashb) (wendb) (wnfb)] {} ;
  \node [draw,dashed,rounded corners,fit=(w0b) (w1b) (whashb) (wendb) (wnfb) 
    ] {} ;
  \node[draw,rounded corners,fit= (wb) (w0b) (w1b) (whashb) (wendb) (wnfb) (12b)] {};

\end{tikzpicture}
\caption{Initial state for the two-agent case. The dashed rectangles represent the equivalence classes of the accessibility relation for agent 1, and the solid ones represents the equivalence classes for agent 2. If a state is not contained in any rectangle, it means that it forms its own equivalence class.
}
\label[figure]{fig:init-state-multi}
\end{figure}
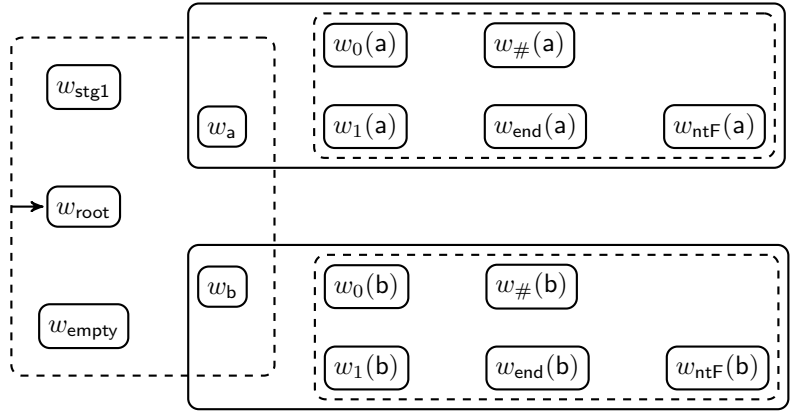

In this multi-agent case, as we apply actions to the epistemic states, in a state, the truth of
\begin{itemize}
  \item $\vbl{root} \land \bar{K}_1 \vbl{empty}$ encodes that we have applied no actions (and the state is the root state);
  \item $\vbl{root} \land \bar{K}_1 \vbl{stg1}$ encodes that we are on the first stage of the plan;
  \item $\# \land \bar{K}_2 \vbl{end}$ encodes that we are at a state at the end of the encoding of the blocks (or at $w_{\#}$); and 
  \item $\okstate = ((0 \lor 1) \land  \bar{K}_1 \vbl{ntF}) \lor ((\aa \lor \bb) \land \neg 0 \land \neg 1 \land  \bar{K}_2 \vbl{ntF})$ encodes that the plan has not attempted to remove the wrong symbol from the state during the second stage.
\end{itemize}
As the reader notices, we are more explicit in the multiagent case, to avoid side-effects due to the constraints on the accessibility relations.

\begin{definition}
Let $q_\aa, q_\bb \in \{0,1\}^*$.
We define the epistemic state $s\block{q_\aa}{q_\bb}$ to be 
$s\block{q_\aa}{q_\bb} = ((W\block{q_\aa}{q_\bb},(R_1\block{q_\aa}{q_\bb},R_2\block{q_\aa}{q_\bb}),V\block{q_\aa}{q_\bb}),w_{\vbl{root}})$, where
    \begin{align*}
      W\block{q_\aa}{q_\bb} = 
      \{ &
      w_{\vbl{root}}, w_{\aa}, w_{\bb}, w_{\vbl{ntF}}(\aa), w_{\vbl{ntF}}(\bb)
            \} 
      \cup \\ 
      \{ & w_{\xx,j}, 
                  w_{\xx,j,\#},  w_{\vbl{ntF}}(\xx,j),
      w_{\vbl{ntF}}(\xx,j,\#),
            \mid 
            x \in \{\aa,\bb\}  \text{ and } 1 \leq j \leq |q_\xx|
            \};
        \\ 
    R_1\block{q_\aa}{q_\bb} = \{
        &w_{\vbl{root}},w_{\xx}
        \mid 
    x \in \{\aa,\bb\}
    \}^2 ~\cup 
     \bigcup_{
     \substack{ 
     \xx \in \{\aa,\bb\}\\ 
     1 \leq j \leq |q_{\xx}|
     }
     }
     \{
    w_{\xx,j},w_{\xx,j,\#}, 
                    w_{\vbl{ntF}}(\xx,j) 
                \}^2 
                                                \\ 
    R_2\block{q_\aa}{q_\bb} = 
    &\bigcup_{
     \xx \in \{\aa,\bb\}
     }
     (
    \{
         w_\xx,w_{\vbl{ntF}}(\xx)
    \} \cup 
    \{ 
        w_{\xx,1}
    \mid q_\xx \neq \varepsilon
    \})^2 \cup 
    \\ & 
    \bigcup_{
     \substack{ 
     \xx \in \{\aa,\bb\}\\ 
     1 \leq j < |q_{\xx}|
     }
     }
    \{
    w_{\xx,j,\#},w_{\xx,j+1}, 
    w_{\vbl{ntF}}(\xx,j,\#) 
            \}^2 \cup 
    &\\ 
    &\bigcup_{
      \substack{
     \xx \in \{\aa,\bb\} \\
     q_{\xx} \neq \varepsilon
     }
     }
    \{ 
    w_{\xx,|q_{\xx}|,\#},w_{\vbl{ntF}}(\xx,|q_{\xx}|,\#)
    \}^2 
    ; \text{ and }
    \end{align*}
    for every $\xx \in \{ \aa,\bb \}$, $V(w_\vbl{root}) = \{\vbl{root}\}$, $V(w_\xx) = \{\xx\}$ and $V(w_{\vbl{ntF}}(\xx)) = \{\vbl{ntF},\xx\}$, and for every $1 \leq j \leq |q_{\xx}|$, 
        $V(w_{\xx,j}) = \{q_\xx[j],\xx\}$, $V(w_{\xx,j,\#}) = \{\#,\xx\}$, 
             and
     $V(w_{\vbl{ntF}}(\xx,j)) = V(w_{\vbl{ntF}}(\xx,j,\#)) = 
         \{\vbl{ntF},\xx\}$.
\end{definition}

We then define $s\block{q_\aa}{q_\bb} - \# $ to be the submodel of $s\block{q_\aa}{q_\bb}$ induced by $W\block{q_\aa}{q_\bb} \setminus \{w_{\aa,|q_{\aa}|,\#},w_{\bb,|q_{\bb}|,\#}\}$, as in \Cref{sec:K}.
The definition of the epistemic state $s\block{q_\aa}{q_\bb}\{0,1,\#\}$ can similarly be adjusted from \Cref{sec:K} 
(see \Cref{fig:epistemicstatemulti}).

\begin{definition}
  Let $q_\aa, q_\bb \in \{0,1\}^*$ and let 
  \[Tl_{q_\xx} = \begin{cases*}\{ w_{\vbl{0}}(\xx), w_{\vbl{1}}(\xx), w_{\vbl{\#}}(\xx), w_{\vbl{ntF}}(\xx,|q_\xx|,\#), w_{\vbl{end}}(\xx)\}, & if $q_\xx \neq \varepsilon$;\\
    \{ w_{\vbl{0}}(\xx), w_{\vbl{1}}(\xx), w_{\vbl{\#}}(\xx), w_{\vbl{ntF}}(\xx), w_{\vbl{end}}(\xx)\}, & if $q_\xx = \varepsilon$;
  \end{cases*} \quad \text{ for }~\xx \in \{\aa,\bb\}.\]  
  We define the epistemic state 
    $s\block{q_\aa}{q_\bb}\{0,1,\#\} = ((W'\block{q_\aa}{q_\bb},(R_1'\block{q_\aa}{q_\bb},R_2'\block{q_\aa}{q_\bb}),V'\block{q_\aa}{q_\bb}),w_{\vbl{root}})$, where
  \begin{align*}
    W'\block{q_\aa}{q_\bb} =& W\block{q_\aa}{q_\bb} &\cup~ &\{  w_{\vbl{stg1}}\} \cup Tl_{q_\aa} \cup Tl_{q_\bb} 
        ;\\
    R_1'\block{q_\aa}{q_\bb} =& 
      R_1\block{q_\aa}{q_\bb} &\cup~  &\{ w_{\vbl{root}},w_{\vbl{stg1}}, w_\aa, w_\bb \}^2  
                  \cup 
       Tl_{q_\aa} ^2 \cup Tl_{q_\bb} ^2 \\
      R_2'\block{q_\aa}{q_\bb} =& 
      R_2\block{q_\aa}{q_\bb} &\cup~  &\{ w_{\vbl{stg1}}\}^2  \cup 
            \bigcup_{\xx \in \{\aa,\bb\}} (Tl_{q_\xx} \cup \{ w_{\xx,|q_\xx|,\#} \mid  q_\xx \neq \varepsilon \} \cup \{ w_{\xx} \mid  q_\xx = \varepsilon \}  )^2;
      \end{align*}
      for every $w$ in $W\block{q_\aa}{q_\bb}$, $V'\block{q_\aa}{q_\bb}(w) = V\block{q_\aa}{q_\bb}(w)$, $V'\block{q_\aa}{q_\bb}(w_{\vbl{stg1}}) = \{\vbl{stg1}\}$, and for every $p \in \{0,1,\#,\vbl{end}
            \}$ and $\xx \in \{\aa,\bb\}$, $V'\block{q_\aa}{q_\bb}(w_{p}(\xx)) = \{p,\xx\}$.
    \end{definition}

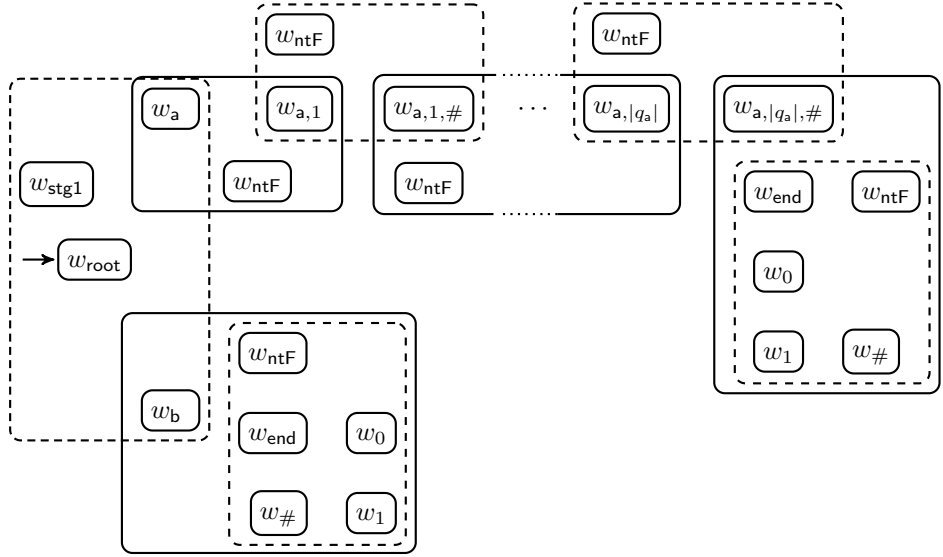
\begin{figure}
  \centering 
\begin{tikzpicture}[-,>=stealth',shorten >=1pt,auto,node distance=1cm, thick]

    \newcommand{\nameofwroot}{$w_\vbl{root}$}
  \node[mynode,initial] (wroot)   at (0, 0)      {\nameofwroot};
  \newcommand{\nameofwst}{$w_{\vbl{stg1}}$};
  \node[mynode] (wst)  at (-0.5, 1)      {\nameofwst};
  \newcommand{\nameofwa}{$w_\aa$ };
  \node[mynode] (wa)   at (1, 2)    {\nameofwa};
  \newcommand{\nameofwb}{$w_\bb$ };
  \node[mynode] (wb)   at (1, -2)   {\nameofwb};
  \newcommand{\nameofwjZa}{$w_{\aa,1}$};
  \node[mynode] (wjZa) at (2.7, 2)    {\nameofwjZa};
      \newcommand{\nameofwjOa}{$w_{\aa,1,\#}$};
  \node[mynode] (wjOa) at (4.4, 2)   {\nameofwjOa};
    
  \node (elips) at (5.8, 2)   {$\cdots$};
  
\newcommand{\nameofwjka}{$w_{\aa,|q_\aa|}$};  
\node[mynode] (wjka) at (7,2)   {\nameofwjka};
    
  \newcommand{\nameofwjkash}{$w_{\aa,|q_\aa|,\#}$};
  \node[mynode] (wjkash) at (9, 2)   {\nameofwjkash};
    
  \newcommand{\nameofwenda}{$w_{\vbl{end}}$};
  \node[mynode] (wenda)  [below = 0.5 of wjkash]  {\nameofwenda};
  \newcommand{\nameofwZa}{$w_0$};
  \node[mynode] (wZa)  [below=0.5 of wenda]  {\nameofwZa};
  \newcommand{\nameofwOa}{$w_1$};
  \node[mynode] (wOa) [below=0.5 of wZa] {\nameofwOa};
  \newcommand{\nameofwhasha}{$w_{\#}$};
  \node[mynode] (whasha) [right=0.5 of wOa]  {\nameofwhasha};

  \node[mynode] (ntfwb) [above right= 0.2 and 0.5 of wb]  {$w_{\vbl{ntF}}$};
  \newcommand{\nameofwendab}{$w_{\vbl{end}}$};
  \node[mynode] (wendb)  [below  = 0.5 of ntfwb]  {\nameofwendab};
  \newcommand{\nameofwZb}{$w_0$};
  \node[mynode] (wZb)  [right=0.5 of wendb]  {\nameofwZb};
  \newcommand{\nameofwOb}{$w_1$};
  \node[mynode] (wOb) [below=0.5 of wZb] {\nameofwOb};
  \newcommand{\nameofwhashb}{$w_{\#}$};
  \node[mynode] (whashb) [left=0.5 of wOb]  {\nameofwhashb};

  \node[mynode] (ntfwjkash) [right=0.5 of wenda]  {$w_{\vbl{ntF}}$};
  \node[mynode] (ntfwa) [below right=0.4 and 0.3 of wa]  {$w_{\vbl{ntF}}$};

  \foreach \y in {wjOa} {
    \node[mynode] (ntf\y) [below=0.4 of \y]  {$w_{\vbl{ntF}}$};
  }

  \foreach \y in {wjZa,wjka} {
    \node[mynode] (ntf\y) [above =0.4  of \y]  {$w_{\vbl{ntF}}$};
  }

    \node[draw,rounded corners,densely dashed,fit= (wroot) (wst) (wa) (wb)] {};

  \node[draw,rounded corners,fit= (wa) (ntfwa) (wjZa)] {};

          \node[draw,rounded corners,dashed,fit= (wjZa) (ntfwjZa) (wjOa)] {};
  \node[draw,rounded corners,dashed,fit= (wjka) (ntfwjka) (wjkash)] {};
  
   \node [fit= (wjOa)(ntfwjOa)(wjka)]  (long_up) {};

  \draw[-,thick, dotted] 
        ($(long_up.north east)!0.4!(long_up.north west)$) -- 
        ($(long_up.north east)!0.6!(long_up.north west)$);

  \draw[-,thick, dotted] 
        ($(long_up.south east)!0.4!(long_up.south west)$) -- 
        ($(long_up.south east)!0.6!(long_up.south west)$);

  \draw[-,thick,rounded corners] 
        ($(long_up.north east)!0.6!(long_up.north west)$) -- 
        (long_up.north west)  -- (long_up.south west) -- 
        ($(long_up.south east)!0.6!(long_up.south west)$);

  \draw[-,thick,rounded corners] 
        ($(long_up.south east)!0.4!(long_up.south west)$) --
        (long_up.south east) -- (long_up.north east) -- 
        ($(long_up.north east)!0.4!(long_up.north west)$);

  \node (boxa) [dashed,fit= (wZa) (wOa) (ntfwjkash) (whasha) (wenda)] {};
  \node [draw,rounded corners,dashed,fit= (wZa) (wOa) (ntfwjkash) (whasha) (wenda) 
    ] {};
  \node [draw,rounded corners,fit= (wjkash) (boxa)] {};
  
  \node (boxb) [dashed,fit= (wZb) (wOb) (ntfwb) (wb) (wendb)] {};
  \node [draw,rounded corners,dashed,fit= (wZb) (wOb) (ntfwb) (whashb) (wendb) 
    ] {};
  \node [draw,rounded corners,fit= (wb) (boxb)] {};

\end{tikzpicture}
\caption{The epistemic state $s\block{q_\aa}{q_\bb} \{0,1,\#\}$ for the case where $q_\aa \neq \varepsilon$ and $q_\bb = \varepsilon$. Dashed lines represent the equivalence classes for agent 1, while solid lines represent the equivalence classes for agent 2. We omit marking the singleton equivalence classes. We write $w_{p}$ instead of $w_{p}(\cdots)$.} 
\label[figure]{fig:epistemicstatemulti}
\end{figure}

\subsection{The Epistemic Actions}

We redefine our shorthand formulas:
\begin{align*}
  \logicize{ag1} & = 
  \vbl{root} \lor 
  \vbl{0} \lor \vbl{1},
  &
  \logicize{ag2} & = (\aa \lor \bb) 
  \land \neg (0 \lor 1 \lor \vbl{ntF} \lor \vbl{end}),
    \\
                      \logicize{symb} & = 0 \lor 1 \lor \# 
    ,
    &
  \logicize{last} & = 
    \logicize{ag2} \land \bar{K}_2 \vbl{end}
    ,
    \\ 
  \logicize{tail}(i) & = 
   \vbl{ag}\vbl{i} \land K_i \neg \vbl{ag}(3-\vbl{i})
   ,
   & 
                                \end{align*}
for every $i \in \{1,2\}$. 

We now define the epistemic actions that we use.

\subsubsection{Adding Blocks}

For every $1 \leq i \leq n$, we define 
\(\ad_i = (E_i,(R_1,R_2),e_s), 
\)
  where
        \begin{align*}
            E_i = \{& e_s,  
                                    e_\aa, e_\bb, 
                  e_{st}, e_{\aa\{\}}, e_{\bb,\{\}}, 
                  e_{\aa,lst}, e_{\bb,lst},                   e_{\aa}^\varepsilon, e_{\bb}^\varepsilon,
                  e_{\aa smb}, e_{\bb smb}, \\&
                  e_{\vbl{ntF}}^{\xx,lst},
                  e_{\vbl{ntF}}^{\aa,1} , 
                  e_{\vbl{ntF}}^{\bb,1} ,
                  e_{\vbl{ntF}}^{\aa,2} , 
                  e_{\vbl{ntF}}^{\bb,2} 
                                                      \} 
                  \\ & \cup 
                  \{
                  e_{\xx,j}, e_{\xx,j,\#}, 
                  e_{\vbl{ntF}}(\xx,j)
                                                                                                            \mid x \in \{\aa,\bb\}, ~~ 0 < j \leq  |x_i| 
                  \} \\ & \cup 
                  \{
                  e_{\vbl{ntF}}(\xx,j,\#)
                                                                                                            \mid x \in \{\aa,\bb\}, ~~ 0 < j <  |x_i| 
                  \}; 
                  \\ 
            R_1 = &\{                 e_s, e_{st}, e_{\aa}, e_{\bb}, 
                e_{\aa}^\varepsilon, e_{\bb}^\varepsilon\}^2 
                                 \cup
                  \{ 
                  e_{\aa smb}, e_{\aa,lst}
                  \}^2
                                    \cup
                  \{ 
                  e_{\bb smb}, e_{\bb,lst}
                  \}^2
                  \\ & \cup
                  \bigcup_{
                    \substack{\xx \in \{\aa,\bb\}
                  \\
                  0 < j \leq |\xx_i|
                  }}
                  \{ 
                  e_{\xx, j}, e_{\xx,j,\#}, e_{\vbl{ntF}}(\xx,j) 
                  \}^2
                  \\ & \cup
                  \bigcup_{
                    \substack{\xx \in \{\aa,\bb\}
                  \\
                  0 < j < |\xx_i|
                  }}
                  \{ 
                  e_{\vbl{ntF}}(\xx,j,\#)
                  \}^2
                  \\ & \cup
                  \{e_{\vbl{ntF}}^{\xx,lst}\}^2 \cup
                  \{e_{\vbl{ntF}}^{\aa,1}\}^2 \cup
                  \{e_{\vbl{ntF}}^{\aa,2}\}^2 \cup
                  \{e_{\vbl{ntF}}^{\bb,1}\}^2 \cup
                  \{e_{\vbl{ntF}}^{\bb,2}\}^2 \cup
                  \{e_{\aa,\{\}}\}^2 \cup
                  \{e_{\bb,\{\}}\}^2; \\ 
            R_2 = & \{ 
                  e_\aa, e_{\vbl{ntF}}(\aa), e_{\aa,smb}\}^2 
                   \cup
                  \{ 
                  e_\bb, e_{\vbl{ntF}}(\bb), e_{\bb,smb}\}^2 
                  \\ & \cup
                  (\{ 
                    e_\aa^\varepsilon, e_{\aa,lst}, e_{\vbl{ntF}}^{\aa,2}, e_{\vbl{ntF}}^{\aa,lst} 
                     \}
                     \cup \{ e_{\aa,1} \mid \aa_i \neq \varepsilon \}
                     \cup \{ e_{\aa, \{\}} \mid \aa_i = \varepsilon  \}
                     )^2
                  \\ & \cup
                  (\{ 
                    e_\bb^\varepsilon, e_{\bb,lst}, e_{\vbl{ntF}}^{\bb,2}, e_{\vbl{ntF}}^{\bb,lst} 
                     \}
                     \cup \{ e_{\bb,1} \mid \bb_i \neq \varepsilon \}
                     \cup \{ e_{\bb, \{\}} \mid \bb_i = \varepsilon  \}
                     )^2 
                  \\ & \cup
                     \bigcup_{
                    \substack{\xx \in \{\aa,\bb\}
                  \\
                  0 < j < |\xx_i|  
                  }}(
                  \{ 
                  e_{\xx, j+1}, e_{\xx,j,\#},  e_{\vbl{ntF}}(\xx,j,\#)
                  \}^2 \cup 
                   \{ e_{\xx,|\xx_i|,\#}, e_{\xx, \{\}} \}^2)
                   \\ & \cup
                   \{ e_s \}^2
                   \cup \{ e_{st} \}^2
                   \cup \bigcup_{\substack{
                    \xx \in \{\aa,\bb\} \\ 0< j \leq |\xx_i|
                   }} \{ e_{\vbl{ntF}}(\xx,j)  \}^2
                   ;
            \end{align*}
            and the precondition for each event in $E_i$ is given below:
            \begin{align*}
                        \pre(e_s) =~& \vbl{root} \land \bar{K}_1 \vbl{stg1} 
            \\ 
                                    \pre(e_{st}) =~& \vbl{stg1} 
            \\ 
                                    \pre(e_\xx) =~& 
            \xx \land \vbl{ag2} \land 
            \neg \#
                                                \land \neg \vbl{last}
                                    \\ 
            \pre(e_{\xx,lst}) =~& 
            \xx \land 
             \# \land \vbl{last} \land K_1 \neg \vbl{end} 
            \\ 
            \pre(e_{\xx}^\varepsilon) =~& 
             \xx \land \vbl{ag2} \land 
            \neg \# \land \vbl{last} \land K_1 \neg \vbl{end}
            \\
                                    \pre(e_{\vbl{ntF}}^{\xx,1}) = \pre(e_{\vbl{ntF}}^{\xx,2}) =~& \xx \land \vbl{ntF} \land K_1 \neg \vbl{end} 
            \\ 
            \pre(e_{\vbl{ntF}}^{\xx,lst}) =~& \xx \land \vbl{ntF} \land \bar{K}_1 \vbl{end}
            \\
            \pre(e_{\xx,smb}) =~& \xx \land \vbl{symb} \land \neg \vbl{last} 
            \\
            \pre(e_{\xx,\{\}}) =~&  \xx \land \bar{K}_1 
                                    \vbl{end} \land (\vbl{symb}\lor \vbl{ntF} \lor \vbl{end})
            \\
            \pre(e_{\vbl{ntF}}(\xx,j)) = \pre(e_{\vbl{ntF}}(\xx,j,\#)) =~& 
            \xx \land 
            \vbl{ntF} \land \bar{K}_1 \vbl{end}
            \\
            \pre(e_{\xx,j}) =~& \xx \land x_i[j] \land \bar{K}_1  \vbl{end} \\
            \pre(e_{\xx,j,\#}) =~& \xx \land \# \land \bar{K}_1  \vbl{end} 
                                                                                                                                                                                                                                                                                                                                                                                    ,  
                                \end{align*}
        for each $ x \in \{\aa,\bb\}$ and $ j \leq |x_i|$. The action $\ad_i$ is illustrated in \Cref{fig:add-multi}.

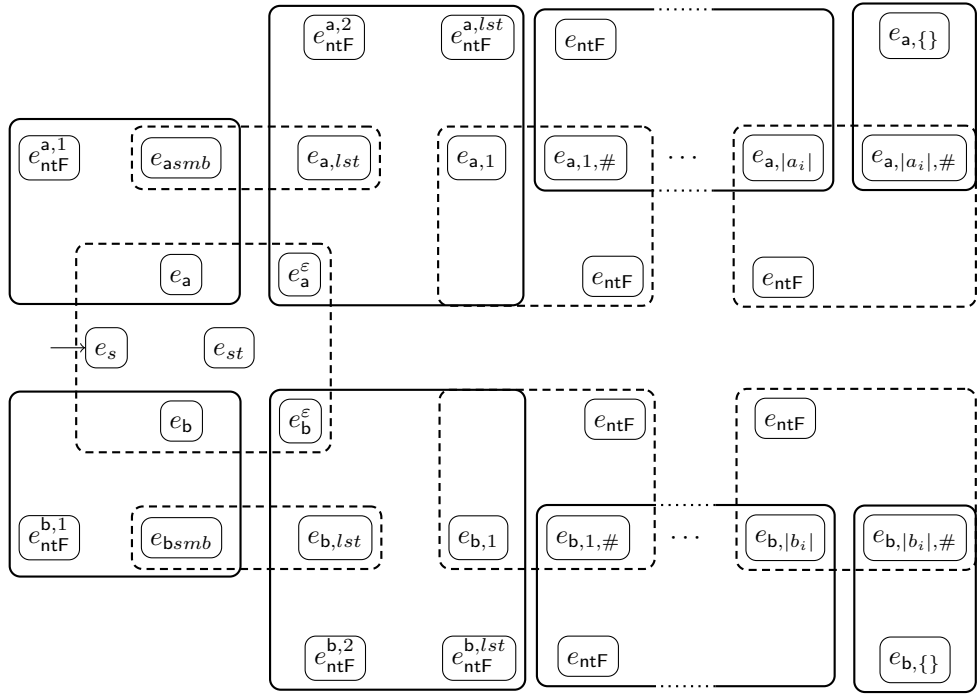
\begin{figure}
  \centering 
\begin{tikzpicture}
  \node [mynodesh,initial] (es)  {$e_s$};
  \node[mynodesh]  (elpa) [above right=0.6 of es] {$e_{\aa}$};
  \node[mynodesh]  (est) [right=of es] {$e_{st}$};
  \node [mynode] (easb) [above=of elpa] {$e_{\aa smb}$};
    \node[mynodesh]  (elpb) [below right=0.6 of es] {$e_{\bb}$};
  \node [mynode] (ebsb) [below=of elpb] {$e_{\bb smb}$};
  
  \node[mynodesh]  (elsta) [right=of elpa] {$e_{\aa}^\varepsilon$};
  \node[mynodesh]  (elstb) [right=of elpb] {$e_{\bb}^\varepsilon$};

  \node[mynode]  (elstahsh) [right=of easb] {$e_{\aa,lst}$};
  \node[mynode]  (elstbhsh) [right=of ebsb] {$e_{\bb ,lst}$};

  \node[mynodesh]  (ea1) [right=of elstahsh] {$e_{\aa,1}$};
  \node[mynodesh]  (eb1) [right=of elstbhsh] {$e_{\bb,1}$};
  \node[mynode]  (ea1sh) [right=0.5 of ea1] {$e_{\aa,1,\#}$};
  \node[mynode]  (eb1sh) [right=0.5 of eb1] {$e_{\bb,1,\#}$};

  \node  (adots) [right= 0.4 of ea1sh] {$\cdots$};
  \node  (bdots) [right= 0.4 of eb1sh] {$\cdots$};

  \node[mynode]  (eak) [right=0.4 of adots] {$e_{\aa,|a_i|}$};
  \node[mynode]  (ebk) [right=0.4 of bdots] {$e_{\bb,|b_i|}$};
  \node[mynode]  (eaksh) [right= 0.5 of eak] {$e_{\aa,|a_i|,\#}$};
  \node[mynode]  (ebksh) [right=0.5  of ebk] {$e_{\bb,|b_i|,\#}$};

    \node[mynodesh]  (elstahshntf) [above=of elstahsh] {$e_{\vbl{ntF}}^{\aa,2}$};
  \node[mynodesh]  (elstbhshntf) [below=of elstbhsh] {$e_{\vbl{ntF}}^{\bb,2}$};

  \node[mynode]  (elstntfa) [right=of elstahshntf] {$e_{\vbl{ntF}}^{\aa,lst}$};
  \node[mynode]  (elstntfb) [right=of elstbhshntf] {$e_{\vbl{ntF}}^{\bb,lst}$};

  \node[mynodesh]  (ea1ntf) [below right=of ea1] {$e_{\vbl{ntF}}$};
  \node[mynodesh]  (eb1ntf) [above right=of eb1] {$e_{\vbl{ntF}}$};
  \node[mynodesh]  (ea1shntf) [above=of ea1sh] {$e_{\vbl{ntF}}$};
  \node[mynodesh]  (eb1shntf) [below=of eb1sh] {$e_{\vbl{ntF}}$};

  \node[mynodesh]  (eakntf) [below=of eak] {$e_{\vbl{ntF}}$};
  \node[mynodesh]  (ebkntf) [above=of ebk] {$e_{\vbl{ntF}}$};
  \node[mynodesh]  (eatail) [above=of eaksh] {$e_{\aa,\{\}}$};
  \node[mynodesh]  (ebtail) [below=of ebksh] {$e_{\bb,\{\}}$};

  \node [mynodesh]  (eantf) [left = 0.8 of easb] {$e_{\vbl{ntF}}^{\aa,1}$};
  \node [mynodesh]  (ebntf) [left=0.8 of ebsb] {$e_{\vbl{ntF}}^{\bb,1}$};

  \node (loop1) [fit= (es)(elpa)(elpb)(elsta)(elstb),draw=black,thick,rounded corners,densely dashed] {};
  \node (loop2) [fit= (easb)(elstahsh),draw=black,thick,rounded corners,densely dashed] {};
  \node (loop2phant) [fit= (easb),draw=none] {};
  \node (loop3) [fit= (elpa)(loop2phant)(eantf),draw=black,thick,rounded corners] {};

  \node (loop2b) [fit= (ebsb)(elstbhsh),draw=black,thick,rounded corners,densely dashed] {};
  \node (loop2bphant) [fit= (ebsb),draw=none] {};
  \node (loop3b) [fit= (elpb)(loop2bphant)(ebntf),draw=black,thick,rounded corners] {};

  \node  [fit= (ea1)(ea1ntf)(ea1sh),draw=black,thick,rounded corners,densely dashed] {};
      \node [fit= (ea1sh)(ea1shntf)(eak)]  (long_up) {};

  \draw[thick, dotted] 
        ($(long_up.north east)!0.4!(long_up.north west)$) -- 
        ($(long_up.north east)!0.6!(long_up.north west)$);

  \draw[thick, dotted] 
        ($(long_up.south east)!0.4!(long_up.south west)$) -- 
        ($(long_up.south east)!0.6!(long_up.south west)$);

  \draw[thick,rounded corners] 
        ($(long_up.north east)!0.6!(long_up.north west)$) -- 
        (long_up.north west)  -- (long_up.south west) -- 
        ($(long_up.south east)!0.6!(long_up.south west)$);

  \draw[thick,rounded corners] 
        ($(long_up.south east)!0.4!(long_up.south west)$) --
        (long_up.south east) -- (long_up.north east) -- 
        ($(long_up.north east)!0.4!(long_up.north west)$);

  \node  [fit= (eak)(eakntf)(eaksh),draw=black,thick,rounded corners,densely dashed] {};
  \node [fit= (elstahsh)(elsta)(ea1)(elstahshntf)(elstntfa),draw=black,thick,rounded corners] {};

  \node [fit= (eb1)(eb1ntf)(eb1sh),draw=black,thick,rounded corners,densely dashed] {};
    \node [fit= (eb1sh)(eb1shntf)(ebk)] (long_down) {};

  \draw[thick, dotted] 
        ($(long_down.north east)!0.4!(long_down.north west)$) -- 
        ($(long_down.north east)!0.6!(long_down.north west)$);

  \draw[thick, dotted] 
        ($(long_down.south east)!0.4!(long_down.south west)$) -- 
        ($(long_down.south east)!0.6!(long_down.south west)$);

  \draw[thick,rounded corners] 
        ($(long_down.north east)!0.6!(long_down.north west)$) -- 
        (long_down.north west)  -- (long_down.south west) -- 
        ($(long_down.south east)!0.6!(long_down.south west)$);

  \draw[thick,rounded corners] 
        ($(long_down.south east)!0.4!(long_down.south west)$) --
        (long_down.south east) -- (long_down.north east) -- 
        ($(long_down.north east)!0.4!(long_down.north west)$);

  \node [fit= (ebk)(ebkntf)(ebksh),draw=black,thick,rounded corners,densely dashed] {};
  \node [fit= (elstbhsh)(elstb)(eb1)(elstbhshntf)(elstntfb),draw=black,thick,rounded corners] {};

    \node [fit= (eaksh)(eatail),draw=black,thick,rounded corners] {};
    \node [fit= (ebksh)(ebtail),draw=black,thick,rounded corners] {};

\end{tikzpicture}
\caption{The epistemic action $\ad_i$ for the case when $\aa_i,\bb_i\neq \varepsilon$. The designated event, $e_s$, is marked with an arrow. 
}
\label[figure]{fig:add-multi}
\end{figure}

\begin{lemma}\label[lemma]{lem:adding-blocks-multi}
  Let $s = s\block{q_\aa }{q_\bb }\{ 0,1,\# \}$ and $1 \leq i \leq n$.  
  Then, $s \times \ad_i = s\block{q_\aa a_i }{q_\bb b_i }\{ 0,1,\# \}$.
\end{lemma}

\begin{proof}
    Let $((W,(R_1,R_2),V),(w_{\vbl{root}},e_s)) = s \times \ad_i$. We show that $s \times \ad_i \sim s\block{q_\aa a_i }{q_\bb b_i }\{ 0,1,\# \}$.
  For simplicity, we refer to all accessibility relations as $(R_1,R_2)$.
  The set of states in $s \times \ad_i$ is:  
  \begin{align*}
       \{& 
      (w_{\vbl{root}},e_s) ,  
       (w_{\vbl{stg1}},e_{st}) \} \\
      \cup \{& (w_{\vbl{end}}(\xx),e_{\xx,\{\}}), 
       (w_{\vbl{0}}(\xx),e_{\xx,\{\}}) , \\
      & (w_{\vbl{1}}(\xx),e_{\xx,\{\}}) ,  
       (w_{\#}(\xx),e_{\xx,\{\}}) ,  
                  && \mid ~ \xx \in \{\aa,\bb\}
      \} &\\
            \cup  \{ &
      (w_{\vbl{ntF}}(\xx),e_{\xx,\{\}})
      && \mid ~ \xx \in \{\aa,\bb\}, ~ q_\xx = \xx_i = \varepsilon 
      \} \\
      \cup  \{ 
    & (w_{\vbl{ntF}}(\xx,|q_\xx|,\#),e_{\xx,\{\}})
           && \mid 
      \xx\in \{\aa,\bb\},  
      ~ q_\xx \neq \varepsilon
      \} 
       \\
            \cup  \{ 
      & (w_\xx,e_{\xx}^\varepsilon)
      && \mid ~ \xx \in \{\aa,\bb\}, ~ q_\xx = \varepsilon
      \} \\ 
      \cup  \{ &
      (w_{\vbl{ntF}}(\xx),e_{\vbl{ntF}}^{\xx,lst})
      && \mid ~ \xx \in \{\aa,\bb\}, ~ q_\xx = \varepsilon 
      \} \\ 
            \cup \{ &
      (w_\xx,e_\xx), 
      (w_{\vbl{ntF}}(\xx),e_{\vbl{ntF}}^{\xx,1}), \\
      & (w_{\xx,j},e_{\xx ,smb}),
      \\
      & (w_{\vbl{ntF}}(\xx,j),e_{\vbl{ntF}}^{\xx,1}), \\
      & (w_{\xx,|q_\xx|,\#},e_{\xx ,lst}), \\
      & (w_{\vbl{ntF}}(\xx,|q_\xx|,\#),e_{\vbl{ntF}}^{\xx,lst}) 
      && \mid ~ \xx \in \{\aa,\bb\}, ~  0 < j \leq |q_\xx|
      \} \\ \cup  \{ &
      (w_{\xx,j,\#},e_{\xx ,smb}),
      \\ 
      & (w_{\vbl{ntF}}(\xx,j,\#),e_{\vbl{ntF}}^{\xx,2})
      && \mid ~ \xx \in \{\aa,\bb\}, ~  0 < j < |q_\xx|
      \} \\ 
             \cup  \{
      & (w_{\xx_i[h]},e_{\xx,h}) ,  \\
      & (w_{\#}(\xx),e_{\xx,h,\#})
      && \mid 
      \xx\in \{\aa,\bb\},  ~ 0 < h \leq |x_i|
      \} 
      \\ \cup  \{ 
    & (w_{\vbl{ntF}}(\xx,|q_\xx|,\#),e_{\vbl{ntF}}(\xx,h,\#)) 
           && \mid 
      \xx\in \{\aa,\bb\},  ~ 0 < h < |x_i|, ~ q_\xx \neq \varepsilon
      \} 
      \\  \cup  \{ 
    & (w_{\vbl{ntF}}(\xx),e_{\vbl{ntF}}(\xx,h,\#)) 
           && \mid 
      \xx\in \{\aa,\bb\},  ~ 0 < h < |x_i|, ~ q_\xx = \varepsilon
      \} 
      .
  \end{align*}
  To see this, observe that for every event $e$ and state $w \in W\block{q_\aa}{q_\bb}$, if $w \models \pre(e)$, then $(w,e)$ appears above.
    
    We split our epistemic actions and epistemic states into 
    parts.
  Let  
  \begin{align*}
  E_1 = \{&       e_s,  
                                    e_\aa, e_\bb, 
                  e_{st}, 
                  e_{\aa,lst}, e_{\bb,lst},                   e_{\aa}^\varepsilon, e_{\bb}^\varepsilon,
                  e_{\aa smb}, e_{\bb smb},                   e_{\vbl{ntF}}^{\xx,lst},
                  e_{\vbl{ntF}}^{\aa,1} , 
                  e_{\vbl{ntF}}^{\bb,1} ,
                  e_{\vbl{ntF}}^{\aa,2} , 
                  e_{\vbl{ntF}}^{\bb,2} 
                                                      \} \\ 
                  &\cup \{ e_\xx^1 \mid \xx \in \{\aa,\bb\},~ \xx_i \neq \varepsilon \}
                  \cup 
                  \{ e_{\xx,\{\}} \mid \xx \in \{\aa,\bb\},~ \xx_i = \varepsilon \}
                  ;
                  \\
  E_2 = \{&       e_{\xx,j}, e_{\xx,j,\#}, 
                  e_{\vbl{ntF}}(\xx,j),
                                                                                                                              \mid x \in \{\aa,\bb\}, ~~ 0 < j \leq  |x_i| 
                  \} 
                  \\ & \cup 
                  \{
                  e_{\vbl{ntF}}(\xx,j,\#)
                                                                                                            \mid x \in \{\aa,\bb\}, ~~ 0 < j <  |x_i| 
                  \}
                  ; \text{ and} \\
  E_3 = \{&      e_{\xx\{\}} \mid \xx \in \{ \aa, \bb \} \} 
  \cup \{ e_{\xx,|\xx_i|,\#} \mid \xx \in \{\aa,\bb\},~ \xx_i \neq \varepsilon \}
                                      \end{align*}

  We observe that the states of $s \times \ad_i$ and $s \block{q_\aa \aa_i}{q_\bb \bb_i} \{ 0, 1, \# \}$ can similarly be arranged into corresponding parts.
  Let 
  \begin{align*}
    W_1 = & \{ 
      w_{\vbl{root}}, w_{\aa}, w_{\bb}, w_{\vbl{ntF}}(\aa), w_{\vbl{ntF}}(\bb)
            \} 
      \cup \\ 
      &\{ w_{\xx,j}, 
                  w_{\xx,j,\#},  w_{\vbl{ntF}}(\xx,j),
                  \mid 
            x \in \{\aa,\bb\}  \text{ and } 1 \leq j \leq |q_\xx|
            \}
      \cup \\ 
      &\{  
      w_{\vbl{ntF}}(\xx,j,\#),
            \mid 
            x \in \{\aa,\bb\}  \text{ and } 1 \leq j < |q_\xx|
            \}
      \\ 
                  &\cup \{ w_{\xx_i[1]}(\xx) \mid \xx \in \{\aa,\bb\},~ \xx_i \neq \varepsilon \}
                  \cup 
                  \bigcup_{\xx_i = \varepsilon} Tl_{q_\xx}
                        ; \\ 
    W_2(\xx)  = & Tl_{q_\xx} 
        ; \\ 
    W'_2(\xx)  = & Tl_{q_\xx \xx_i}  
    \cup \{ w_{\xx,|\xx_i|,\#} \mid \xx \in \{\aa,\bb\},~ \xx_i \neq \varepsilon \}
        ;~\text{ and}  \\ 
    W_3(\xx) = & \{  w_{\xx,j}, 
                  w_{\xx,j,\#},  w_{\vbl{ntF}}(\xx,j),
                  \mid 
            x \in \{\aa,\bb\}  \text{ and } |q_\xx| < j \leq |q_\xx| + |\xx_i|
            \}
      \cup \\ 
      & \{  
      w_{\vbl{ntF}}(\xx,j,\#),
            \mid 
            x \in \{\aa,\bb\}  \text{ and } |q_\xx|< j < |q_\xx| + |\xx_i|
            \}.
  \end{align*}

  We consider the pairs of generated submodels on the respective restrictions of $s \times \ad_i$ and $s \block{q_\aa \aa_i}{q_\bb \bb_i} \{ 0, 1, \# \}$ on:
  \begin{enumerate}
    \item $W_1 \times E_1$ and $W_1 $;
    \item $W_2(\xx) \times E_2$ and $W_3(\xx)$; 
    \item $W_2(\xx) \times E_3$ and $W'_2(\xx)$.
  \end{enumerate}
  These sets cover all states --- and more importantly, all accessibility relation pairs --- of the respective models. We will prove that the above pairs are bisimilar. Then, it is not hard to see that, since these sets cover all states and accessibility relation pairs, the complete epistemic states are also bisimilar.
    
  In the first case, observe that each state $w \in W_1$ satisfies exactly one precondition from the ones of $E_1$. Let $e(w) \in E_1$, such that $s, w \models \pre(e(w))$.
  Furthermore, we see that for every $w_1, w_2 \in W_1$, if $w_1 R_j w_2$, we also have that $e(w_1) R_j e(w_2)$; therefore, it is immediate that 
  \[
    (s \times \ad_i |_{ W_1 \times E_1}, (w,e(w))) \sim (s |_{W_1},w) = (s \block{q_\aa \aa_i}{q_\bb \bb_i}\{0,1,\#\} |_{W_1},w),
  \]
      for every $w \in W_1$.

 In the second case, we observe that $W_2(\xx)$ is a clique for both accessibility relations, and for each $e \in E_2$, there is a unique $w(e) \in W_2(\xx)$ such that $s, w(e) \models \pre(e)$ (and $w(-)$ is surjective on $W_2(\xx)$).
 For each $e_{\xx,j}, e_{\xx,j,\#}, 
                  e_{\vbl{ntF}}(\xx,j) \in E_2$, let 
                  $w_\times (e_{\xx,j}) = w_{\xx,|q_\xx|+j}$; 
                  $w_\times (e_{\xx,j,\#}) = w_{\xx,|q_\xx|+j,\#}$; and
                  $w_\times (e_{\vbl{ntF}}(\xx,j)) = w_{\vbl{ntF}}(\xx,j)$.
                  And for each $e_{\vbl{ntF}}(\xx,j) \in E_2$, let
                  $w_\times (e_{\vbl{ntF}}(\xx,j,\#)) = w_{\vbl{ntF}}(\xx,j,\#)$.
  It is then, not hard to see that 
    \[
  (s \times \ad_i
    |_{
  W_2(\xx) \times E_2
  }, (w(e),e))
      \sim 
  (s \block{q_\aa \aa_i}{q_\bb \bb_i}\{0,1,\#\} |_{W_3(\xx)},w_\times(e)),
  \]  
  for every $e \in E_2$.
    
  Finally, for the third case, 
  let $E_3' = E_3 \setminus \{e_{\xx,|\xx_i|,\#}\}$.
  As $E_3'$  includes a single event whose precondition is satisfied in all of $W_2(\xx)$, 
    it is not hard to see that 
  \[ 
    (s \times \ad_i |_{W_2(\xx) \times E_3'}, (w,e_{\xx,\{\}})) \sim (s |_{W_2(\xx)},w) \sim (s \block{q_\aa \aa_i}{q_\bb \bb_i}\{0,1,\#\} |_{Tl_{q_\xx \xx_i}},w^+),
  \]
  where if $q_\xx = \varepsilon$,  $\xx_i \neq \varepsilon$, and $w = w_\vbl{ntF} (\xx)$, then $w^+ = w_\vbl{ntF} (\xx,|\xx_i|,\#)$; if $q_\xx \neq \varepsilon$,  $\xx_i \neq \varepsilon$, and $w = w_\vbl{ntF} (\xx,|q_\xx|,\#)$, then $w^+ = w_\vbl{ntF} (\xx,|q_\xx|+|\xx_i|,\#)$; and otherwise, $w^+ = w$.
        The only interesting case is when $q_\xx = \varepsilon$ and $\xx_i \neq \varepsilon$, but then it suffices to observe that $w_\vbl{ntF} (\xx)$ and $w_\vbl{ntF} (\xx,|\xx_i|,\#)$ are bisimilar.
    Then, after reintroducing $(w_{\#}(\xx),e_{\xx,|\xx_i|,\#})$ and $w_{\xx,|\xx_i|,\#}$, we observe that 
  \[ 
    (s \times \ad_i |_{W_2(\xx) \times E_3}, (w,e_{\xx,\{\}})) \sim (s \block{q_\aa \aa_i}{q_\bb \bb_i}\{0,1,\#\} |_{Tl_{q_\xx \xx_i} W'_2(\xx)},w^+),
  \]
  and 
  \[ 
    (s \times \ad_i |_{W_2(\xx) \times E_3}, (w_{\#}(\xx),e_{\xx,|\xx_i|,\#}) \sim (s \block{q_\aa \aa_i}{q_\bb \bb_i}\{0,1,\#\} |_{Tl_{q_\xx \xx_i} W'_2(\xx)},w_{\xx,|\xx_i|,\#}).
  \]

  This completes the proof, as we can easily construct a bisimulation between 
  $s \times \ad_i $ and $ s\block{q_\aa a_i }{q_\bb b_i }\{ 0,1,\# \}$ by taking the union of the three bisimilarity relations between the respective submodels.
  As the bisimilarity relation is a bisimulation and every accessibility relation appears in at least one submodel, all bisimulation conditions are satisfied.
      \end{proof}

\begin{corollary}\label[corollary]{cor:adding-blocks-multi}
  Let $\alpha = \ad_{i_1} \ad_{i_2} \cdots \ad_{i_k} \in \{ \ad_i \mid 1 \leq i \leq n \}^*$. 
  Then, $s_0 \times \alpha = s\block{a_{i_1} a_{i_2} \cdots a_{i_k}}{b_{i_1} b_{i_2} \cdots b_{i_k}}\{ 0,1,\# \}$.
\end{corollary}

        \subsubsection{Moving to the Second Stage}
        In the multi-agent case, we define the define epistemic action 
        \[\nxtstg = 
                (\{e_{nx}\}, (\{(e_{nx},e_{nx})\},\{(e_{nx},e_{nx})\}), e_{nx})
        ,
        \text{ where }
        \]
            \[
      \pre(e_{nx}) =  (\neg\vbl{root} \lor ( K_1 \neg \vbl{empty} \land \bar{K}_1 \vbl{stg1})) 
      \land
      \vbl{stg1} \land
      \neg (\bar{K}_1 0 \land \bar{K}_1 1) \lor \vbl{ntF}
    \]
                                                                                    
    \begin{figure}
      \centering 
    \begin{tikzpicture}[->,>=stealth',shorten >=1pt,auto,node distance=2.8cm, thick]
    
            \node[mynode,initial] (es)       {$e_{nx}$};

      \node (pre-es)   [right= 0.5 of es]     {$(\neg\vbl{root} \lor ( K_1 \neg \vbl{empty} \land \bar{K}_1 \vbl{stg1})) 
      \land
      \neg \vbl{stg1} \land
      \neg (\bar{K}_1 0 \land  \bar{K}_1 1) \lor \vbl{ntF}$};

    \end{tikzpicture}
    \caption{next stage}
    \end{figure}

        Action $\nxtstg$ removes $w_{\vbl{stg1}}, w_{\vbl{end}}, w_{0}(\aa), w_{1}(\aa), w_{\#}(\aa), w_{0}(\bb), w_{1}(\bb), w_{\#}(\bb)$, initiating the second stage. 

        \begin{lemma}\label[lemma]{lem:next-stage_short_multi}
          Let $s = s\block{q_\aa }{q_\bb }\{ 0,1,\# \}$.  
          Then, $s \times next\_ stage = s\block{q_\aa }{q_\bb }$.
        \end{lemma}

        \begin{proof}
          Simply observe that the states that do not satisfy $pre(e_{nx})$ in $s$ are exactly $w_{\vbl{stg1}}$, $w_{\vbl{end}}$, $w_{0}(\aa)$, $w_{1}(\aa)$, $w_{\#}(\aa)$, $w_{0}(\bb)$, $w_{1}(\bb)$, $w_{\#}(\bb)$.
        \end{proof}
                                                                        
        Similarly, it is straightforward to see that the following lemma holds.

        \begin{lemma}\label[lemma]{lem:next-stage_no_add-multi}
          Let $s$ be an epistemic state where $next\_ stage$ applies.  
          Then, $s \times next\_ stage \models K_1 \neg \vbl{stg1}$.
        \end{lemma}

        We call the epistemic states of the form $s_0 \times \alpha$, where $\alpha \in \{\ad_i \mid 1 \leq i \leq n\}^*$, first-stage states. 

\begin{lemma}\label[lemma]{lem:K-is-preserved-multi}
  Let $s$ be an epistemic state, $f$ an epistemic action, and $K_i \varphi$ an epistemic formula of modal depth $1$.
  If $s \models K_i \varphi$, then $s \times f \models K_i \varphi$.
\end{lemma}

\begin{proof}
  The reasoning is the same as in the proof of \Cref{lem:K-is-preservedK}. 
\end{proof}

\begin{lemma}
  For every epistemic state $s 
    $ and $i \leq n$, neither $\ad_i$ nor $\nxtstg$ apply to $s \times \nxtstg$. Furthermore, if $\ad_i$ and $\nxtstg$ do not apply to $s$, then for every epistemic action $f$ that applies to $s$, neither $\ad_i$ nor $\nxtstg$ apply to $s \times f$.
\end{lemma}

\begin{proof}
  Observe that 
  $\ad_i$ does not apply to an epistemic state $s'$ if and only if 
  $s' \models K_1 \neg \vbl{stg1}$.
  Furthermore, by \Cref{lem:next-stage_no_add-multi}, 
  $s \times \nxtstg  \models K_1 \neg \vbl{stg1}$, which implies the first statement.
  The second statement is a consequence of \Cref{lem:K-is-preserved-multi}.  
\end{proof}

\begin{corollary}\label[corollary]{cor:verif-at-end-multi}
  Let 
  $s = \initial \times \alpha$, where $\alpha$ is a sequence of epistemic actions 
    that includes $\nxtstg$. Then, neither $\nxtstg$ nor $\ad_i$ 
    apply to $s$. 
    \end{corollary}

Therefore, epistemic action $\nxtstg$ can only be applied once in a plan, and it marks the end of the applications of the actions $\ad_i$.

        \subsubsection{Verifying one Symbol at Each Step}
        Finally, for each $\vbl{bt} \in \{0,1,\#\}$, we define the epistemic action
        \[ \remv_{\vbl{bt}} = (\{e^{\vbl{bt}}\},(\{(e^{\vbl{bt}},e^{\vbl{bt}})\},\{(e^{\vbl{bt}},e^{\vbl{bt}})\}),e^{\vbl{bt}}), ~~\text{ where}  \]
        \begin{align*}
        \pre(e^{\vbl{bt}}) =& (\vbl{root} \land K_1 \neg \vbl{stg1})
        \lor \left(\vbl{ntF} \land \bigvee_{i = 1,2} (\bar{K}_i \vbl{ag1} \land \bar{K}_i \vbl{ag2})\right) \\ &
        \lor \bigvee_{i=1,2} (\vbl{ag}i \land \neg (\logicize{tail}(i) \land \vbl{bt} \land \bar{K}_i \vbl{ntF}) ).
             \end{align*}

        For each $\vbl{bt} \in \{0,1,\#\}$, $\remv_{\vbl{bt}}$ only activates during the second stage and if the epistemic state is of the 
    right form, \emph{i.e.}
        $s\block{q_\aa\vbl{bt}}{q_\bb\vbl{bt}}$, the resulting state is $s\block{q_\aa}{q_\bb}$. If the epistemic state is not of the right form, then $\pre(e^{\vbl{bt}})$  ensures that a last state can no longer access $w_{\vbl{ntF}}$, thus marking this attempt as a failed one.

\begin{lemma}\label[lemma]{lem:remove-only-at-the-end-multi}
  Let $s$ be an epistemic state. It cannot be the case that $\remv_{\vbl{bt}}$ applies to $s$ and so does $\ad_i$ or $\nxtstg$.
\end{lemma}
\begin{proof}
  The lemma results from the observation that the precondition of $e^{\vbl{bt}}
    $ is consistent with neither of the preconditions for $e_s$ and $e_{nx}$.
\end{proof}

\begin{lemma}\label[lemma]{lem:removing-bits-successfully-multi}
  Let $q_\aa, q_\bb \neq \varepsilon$.
  \begin{enumerate}
    \item 
  Let $s = s\block{q_\aa}{q_\bb}$. 
    Then, $s \times \remv_{\vbl{\#}} = s\block{q_\aa }{q_\bb } - \#$.
        \item 
  Let $s = s\block{q_\aa \vbl{bt}}{q_\bb \vbl{bt}} - \#$. 
    Then, $s \times \remv_{\vbl{bt}} = s\block{q_\aa }{q_\bb }$.
  \end{enumerate}
\end{lemma}

\begin{proof}
  For both statements, it suffices to observe that the states that do not satisfy $\pre(e^{\vbl{bt}})$ in $s$ are exactly 
  $w_{\xx,|q_\xx|,\#}$ and $w_{\vbl{ntF}}(\xx,|q_\xx|,\#)$ for statement 1, and 
  $w_{\xx,|q_\xx|}$ and $w_{\vbl{ntF}}(\xx,|q_\xx|)$ for statement 2, where $\xx \in \{\aa,\bb\}$.
\end{proof}

Let $s = ((W,R,V),s_0)$ be an epistemic state. We say that $s$ is a \emph{failed state} when there is a sequence of states $s_0 R_1 s_1 R_2 \cdots R_{j} s_k$, where $j=1$ if $k$ is odd and $j=2$ if $k$ is even, such that 
$k \geq 1$; 
$((W,R,V),s_0) \models \vbl{root} \land K_1 \neg \vbl{stg1}$;
$((W,R,V),s_1) \models \xx $ for some $x \in \{\aa,\bb\}$;
for each $0<i<k$ and 
$j \in \{1,2\}$, 
if 
$((W,R,V),s_i) \models \vbl{ag}j $, then $((W,R,V),s_{i+1}) \models \vbl{ag}(3-j) $;
and 
$((W,R,V),s_k) \models \vbl{bt} \land K \neg \vbl{ntF}$ for some $\vbl{bt} \in \{\aa,\bb,0,1,\#\}$.

\begin{lemma}\label[lemma]{lem:removing-bt-when-hash-multi}
  If $\remv_{\vbl{bt}}$ applies to epistemic state $s$, and
  \begin{itemize}
    \item 
      $s = s\block{q_\aa}{q_\bb}$ and $\vbl{bt} \neq \#$ or $q_\aa = \varepsilon$ or $q_\bb = \varepsilon$; 
            \item 
      $s = s\block{q_\aa \vbl{bt}'}{q_\bb} - \#$ or $s = s\block{q_\aa }{q_\bb \vbl{bt}'} - \#$ and 
    $\vbl{bt}' \neq \vbl{bt}$; or  
    \item 
      $s$ is a failed state;
  \end{itemize}
  then 
  $s \times \remv_{\vbl{bt}}$ is a failed state.
\end{lemma}

\begin{proof}
  We prove the lemma for the case where $s = s\block{q_\aa \vbl{bt}'}{q_\bb} - \#$, as the other cases are similar.
  There exist some states $w_{\vbl{root}} = r_0 ~ R_1 ~ w_{\aa} = r_1 ~R_2~\cdots ~R_i~ r_k = w_{\aa,|q_\aa|+1}$ in $s$, such that $w_{\aa,|q_\aa|+1} \models \vbl{tail} \land \vbl{bt}'$ in $s$, and
  for each $0<i<k$ and 
  $j \in \{1,2\}$,
    if
$r_i \models \vbl{ag}j $, then $r_{i+1} \models \vbl{ag}(3-j) $.
  Since $\remv_{\vbl{bt}}$ applies to $s$, $r_0 \models \pre(e^{\vbl{bt}}
    )$; furthermore, it is not hard to verify that for every $ 0 < i \leq k$, $r_i \models \pre(e^{\vbl{bt}})$; 
  and that $w_{\vbl{ntF}}(\aa,|q_\aa|,\#) \not \models \pre(e^{\vbl{bt}})$.
    Therefore, there exists a sequence of states 
  \[  
    (r_0,e^{\vbl{bt}}
        ) ~ R_1 ~ (r_1,e^{\vbl{bt}})  ~R_2~\cdots ~R_i~ (r_k,e^{\vbl{bt}}) 
  \]
  in $s \times \remv_{\vbl{bt}}$
  that satisfies the conditions for a failed state.
\end{proof}

\begin{corollary}\label[corollary]{cor:plan-sequence-multi}
  Let $\alpha$ be a plan (a sequence of epistemic actions) that applies to epistemic state $s_I$.
  Then, $\alpha$ is of the form $\ad_{i_1}\ad_{i_2}\cdots \ad_{i_k}$ or 
  \[\ad_{i_1}\ad_{i_2}\cdots \ad_{i_k} ~\nxtstg~ \remv_{\vbl{bt}_1}\remv_{\vbl{bt}_2}\cdots \remv_{\vbl{bt}_l}.\]
\end{corollary}
\begin{proof}
  The corollary is a consequence of \Cref{lem:next-stage_no_add-multi,cor:verif-at-end-multi,lem:adding-blocks,lem:removing-bits-successfully,lem:removing-bt-when-hash-multi}.
\end{proof}

To complete the definition of $EP$, let 
\[
  \act = \{ \ad_i, \nxtstg, remove^{\vbl{bt}} \mid 1 \leq i \leq n \text{ and } \vbl{bt} \in \{ 0,1,\# \}  \},
\]
and $\varPhi_G = K_1 \neg \vbl{empty} \land K_1 (\neg(a \lor b) \lor 
(
    \vbl{tail}(2) \land \bar{K}_2 \vbl{ntF}) )$.

\begin{lemma}\label[lemma]{lem:failed-states-fail-multi}
  Let $s$ be a failed epistemic state. Then, $s \not \models \varPhi_G$.
\end{lemma}
\begin{proof}
  A consequence of the definition of a failed state and of $\varPhi_G$.
\end{proof}

\begin{lemma}\label[lemma]{lem:successfully-removing-yields-match-multi}
  Let $q_\aa, q_\bb \in \{0,1\}^*$ and let $rm = \remv_{\vbl{bt}_1}\remv_{\vbl{bt}_2}\cdots \remv_{\vbl{bt}_l}$.
  If $s\block{q_\aa}{q_\bb} \times rm \models \varPhi_G$, then $q_\aa = q_\bb$ and 
  \[
    rm = \remv_{\#}\remv_{q_\aa[|q_\aa|]}\remv_{\#}\remv_{q_\aa[|q_\aa|-1]}\cdots \remv_{\#}\remv_{q_\aa[1]}.
  \]
\end{lemma}
\begin{proof}
  We proceed by induction on $rm$. If $rm = \varepsilon$, then $s\block{q_\aa}{q_\bb} = s\block{q_\aa}{q_\bb} \times rm \models \varPhi_G$, and therefore, $q_\aa = q_\bb = \varepsilon$.
  If $rm = \remv_{\vbl{bt}_1} rm' $, then, by \Cref{lem:removing-bits-successfully,lem:removing-bt-when-hash-multi}, either $s\block{q_\aa}{q_\bb} \times \remv_{\vbl{bt}_1}$ is a failed state, or it is $s\block{q_\aa}{q_\bb}-\#$, where $q_\aa,q_\bb\neq \varepsilon$ and $\vbl{bt}_1 = \#$.
  The first case results in a contradiction, due to \Cref{lem:removing-bt-when-hash-multi,lem:failed-states-fail-multi}.

  If $rm' = \varepsilon$, then $s\block{q_\aa}{q_\bb} \times rm = s\block{q_\aa}{q_\bb}-\# \not \models \varPhi_G$, which contradicts the lemma's assumptions. Therefore, $rm' = \remv_{\vbl{bt}_2} rm''$. Using similar reasoning as above, 
    $q_\aa = q_\aa' \vbl{bt}_2$ and $q_\bb = q_\bb' \vbl{bt}_2$. We conclude the proof by the inductive hypothesis on $rm''$.
\end{proof}

We can conclude with our main theorem.

\begin{theorem}\label[theorem]{thm:reduction}
  $B$ has a match if and only if $EP$ has a plan.
\end{theorem}
\begin{proof}
  To prove the first direction, let $i_1, i_2, \ldots, i_k$ be a match for $B$. Let 
  $m_1 m_2 \cdots m_\ell = a_1 a_2 \cdots a_k = b_1 b_2 \cdots b_k$, where for every $1 \leq i \leq \ell$, $m_i \in \{0,1\}$. 
  Let 
  \begin{align*}
  \ad_m &= \ad_{i_1} \ad_{i_2} \cdots \ad_{i_k};\\ 
  \remv_m &= \remv_{\#} \remv_{m_\ell} \remv_{\#} 
    \cdots \remv_{\#} \remv_{m_1}; \text{ and} \\ 
  \alpha &= \ad_m ~\nxtstg~ \remv_{m}.
  \end{align*}
  We now demonstrate that $\alpha$ is a plan that applies to $s_I$ and that $s_{I} \times \alpha \models \varPhi_G$.

  By \Cref{cor:adding-blocks-multi}, $s_{I} \times \ad_m = s\block{m_1 m_2 \cdots m_\ell}{m_1 m_2 \cdots m_\ell} \{ 0,1,\# \}$; then, \Cref{lem:next-stage_short_multi} yields that \[s_{I} \times \ad_m ~\nxtstg =  s\block{m_1 m_2 \cdots m_\ell}{m_1 m_2 \cdots m_\ell} \{ 0,1,\# \} \times \nxtstg = s\block{m_1 m_2 \cdots m_\ell}{m_1 m_2 \cdots m_\ell};\] and finally, \Cref{lem:removing-bits-successfully} yields that \[s_I \times \alpha = s\block{m_1 m_2 \cdots m_\ell}{m_1 m_2 \cdots m_\ell} \times \remv_{m} = s\block{\varepsilon}{\varepsilon}\models \varphi_G.\]

  For the converse direction, let $\alpha$ be a plan that applies to $s_I$, such that $s_I \times \alpha \models \varPhi_G$.
  By \Cref{cor:plan-sequence-multi}, $\alpha$ is of the form 
  $\ad_{i_1}\ad_{i_2}\cdots \ad_{i_k}$ or 
  \[\ad_{i_1}\ad_{i_2}\cdots \ad_{i_k} ~\nxtstg~ \remv_{\vbl{bt}_1}\remv_{\vbl{bt}_2}\cdots \remv_{\vbl{bt}_l}.\]
  If $\alpha$ is of the form 
  $\ad_{i_1}\ad_{i_2}\cdots \ad_{i_k}$, then $s_I \times \alpha = s\block{a_{i_1}\cdots a_{i_k}}{b_{i_1}\cdots b_{i_k}} \not \models \varPhi_G$, contradicting our assumptions.
  Therefore, $\alpha$ is of the form 
   \[\ad_{i_1}\ad_{i_2}\cdots \ad_{i_k} ~\nxtstg~ \remv_{\vbl{bt}_1}\remv_{\vbl{bt}_2}\cdots \remv_{\vbl{bt}_l}.\]
   We prove that $i_1,i_2,\ldots,i_k$ is a match for $B$. 
     Let 
    $m_1 m_2 \cdots m_\ell = a_1 a_2 \cdots a_k$ and 
    $m'_1 m'_2 \cdots m'_{\ell'} = b_1 b_2 \cdots b_k$.
    By \Cref{cor:adding-blocks}, $s_{I} \times \ad_{i_1}\ad_{i_2}\cdots \ad_{i_k} = s\block{m_1 m_2 \cdots m_\ell}{m'_1 m_2 \cdots m_{\ell'}} \{ 0,1,\# \}$; then, \Cref{lem:next-stage_short} yields that 
    \[
      s_{I} \times \ad_{i_1}\ad_{i_2}\cdots \ad_{i_k} ~\nxtstg =  s\block{m_1 m_2 \cdots m_\ell}{m'_1 m'_2 \cdots m'_{\ell'}} \{ 0,1,\# \} \times \nxtstg = s\block{m_1 m_2 \cdots m_\ell}{m'_1 m'_2 \cdots m'_{\ell'}}.
      \]
     Finally, by \Cref{lem:successfully-removing-yields-match-multi}, we conclude that $a_1 a_2 \cdots a_k = m_1 m_2 \cdots m_\ell = m'_1 m'_2 \cdots m'_{\ell'} = b_1 b_2 \cdots b_k$, and therefore $i_1,i_2,\ldots,i_k$ is a match for $B$. 
\end{proof}

\section{A Single Agent with Reflexivity or Symmetry}
\label{sec:rs}

We now proceed to adjust our reduction to the case of a single agent with an accessibility relation that is reflexive or symmetric --- but not transitive.
In fact, we will give a construction using accessibility relations that are both reflexive and symmetric.

We use the following set of propositional variables: 
\[
\prop = \{ 0, 1, \#_1, \#_2, \vbl{a}, \vbl{b}, \vbl{root}, \vbl{stg1}, \vbl{empty}, \vbl{end}, \vbl{ntF} , \vbl{lp}  \}.
\] 

 Instead of one separator symbol as in the multi-agent case, we use two ($\#_1$ and $\#_2$) to help give the epistemic states an orientation, which we need for our encoding to work as intended.

\subsection{The Epistemic States}

The initial epistemic state $\initial = ((W_0,R_0,V_0),w_{\vbl{root}})$, where
\begin{align*}
        W_0 = \{ &
    w_{p} \mid p \in {\{ \vbl{root}, \vbl{empty}, \vbl{stg1}, \aa, \bb, \vbl{lp}\}}\} 
    \\ \cup &
    \{ w_{p}(\xx) \mid p \in {\{ \vbl{0}, \vbl{1}, \vbl{\#}_1, \vbl{\#}_2, \vbl{ntF}, \vbl{end}
        \}}, ~ \xx \in {\{\aa,\bb\}} \};\\
                R_0 &\text{ is the reflexive and symmetric closure of: }
    \\ 
    R_0^-
    = \{ 
        &(w_{\vbl{root}},w_{\vbl{empty}}), (w_{\vbl{root}},w_{\vbl{stg1}}), (w_{\vbl{root}},w_{\xx}), 
    \\ 
    &
        (w_\xx,w_{\vbl{bt}}(\xx)), (w_{\vbl{bt}}(\xx),w_{\#_1}(\xx)), 
    \\ 
    &
            (w_{\#_1}(\xx),w_{\#_2}(\xx)), 
                (w_{\#_2}(\xx),w_{\vbl{bt}}(\xx))
    ,\\ 
    &
      (w_{\#_2}(\xx),w_{\vbl{end}}(\xx)) 
     , (w_{\xx},w_{\vbl{end}}(\xx)),\\ 
     &
     (w_{\xx},w_{\vbl{ntF}}(\xx)),
     (w_{\vbl{bt}}(\xx),w_{\vbl{ntF}}(\xx)),
     \\ 
    &
      (w_{\#_1}(\xx),w_{\vbl{ntF}}(\xx)),
      (w_{\#_2}(\xx),w_{\vbl{ntF}}(\xx)),
      \\ 
    &
    (w_{\vbl{bt}}(\xx),w_{\vbl{lp}}),
              (w_{\#_1}(\xx),w_{\vbl{lp}}),
      (w_{\#_2}(\xx),w_{\vbl{lp}})
    &\mid \xx \in \{\aa,\bb\}, ~ \vbl{bt} \in \{0,1\} \};     \end{align*}
        and 
    for every $w_p, w_{p'}(\xx) \in W_0$, $V_0(w_p) = \{p\}$ and $V_0(w_{p'}(\xx)) = \{p',\xx\}$.
    The epistemic state $\initial$ is illustrated in \Cref{figure:initialRS}.

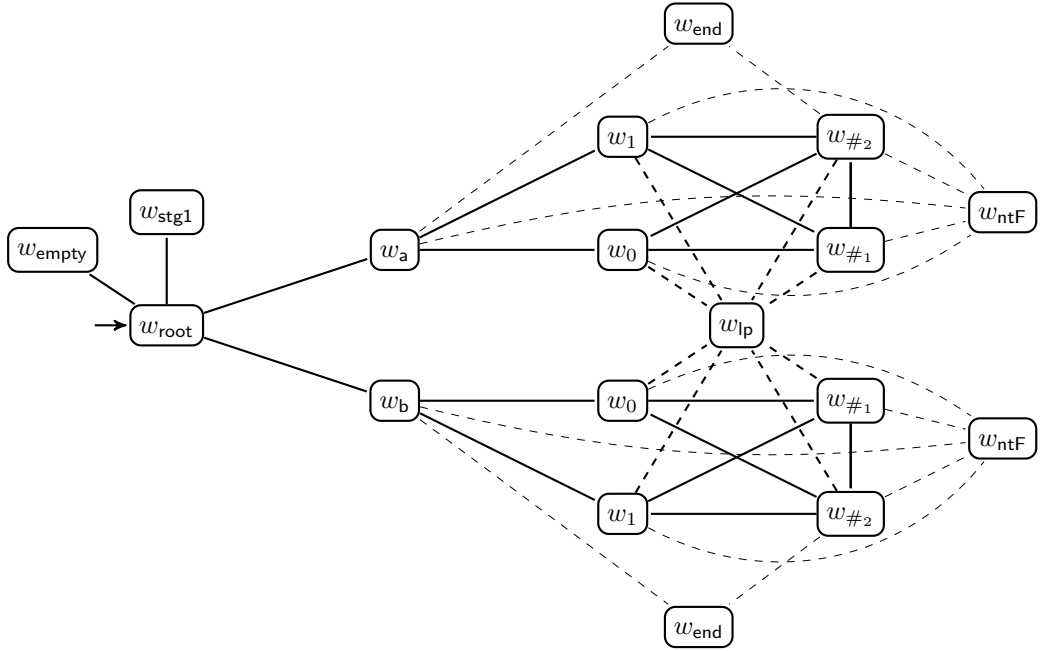
\begin{figure}
  \centering 
\begin{tikzpicture}[->,>=stealth',shorten >=1pt,auto,node distance=2.5cm, thick]

    \node[mynode,initial] (r0)   at (0, 0)    {$w_{\vbl{root}}$};
  \node[mynode] (stg)   at (0, 1.5)    {$w_{\vbl{stg1}}$};
  \node[mynode] (empty)   at (-1.5, 1)    {$w_{\vbl{empty}}$};
  \node[mynode] (wa)   at (3, 1)  {$w_{\vbl{a}}$};
  \node[mynode] (wb)   at (3, -1) {$w_{\vbl{b}}$};

  \node[mynode] (w0a)   at (6, 1)  {$w_0$};
  \node[mynode] (w1a)   at (6, 2.5) {$w_1$};
  \node[mynode] (whash1a) at (9, 1)   {$w_{\#_1}$};
  \node[mynode] (whash2a) at (9, 2.5)   {$w_{\#_2}$};
  \node[mynode] (wenda) at (7, 4)   {$w_{\vbl{end}}$};
  \node[mynode] (wnfa) at (11, 1.5)   {$w_{\vbl{ntF}}$};

  \node[mynode] (w0b)   at (6, -1)  {$w_0$};
  \node[mynode] (w1b)   at (6, -2.5) {$w_1$};
  \node[mynode] (whash1b) at (9, -1)   {$w_{\#_1}$};
  \node[mynode] (whash2b) at (9, -2.5)   {$w_{\#_2}$};
  \node[mynode] (wendb) at (7, -4)   {$w_{\vbl{end}}$};
  \node[mynode] (wnfb) at (11, -1.5)   {$w_{\vbl{ntF}}$};

    \draw[-] (r0) -- (stg);
  \draw[-] (r0) -- (empty);
  \draw[-] (r0) -- (wa);
  \draw[-] (r0) -- (wb);

  \draw[-] (wa) -- (w0a);
  \draw[-] (wa) -- (w1a);
  \draw[-] (wb) -- (w0b);
  \draw[-] (wb) -- (w1b);

  \draw[-] (w0a) -- (whash1a);
  \draw[-] (w1a) -- (whash1a);

  \draw[-] (w0b) -- (whash1b);
  \draw[-] (w1b) -- (whash1b);

  \draw[-] (whash1a) -- (whash2a);
  \draw[-] (whash1a) -- (whash2a);

  \draw[-] (whash2a) -- (w0a);
  \draw[-] (whash2a) -- (w1a);

  \draw[-,line width=0.1,dashed] (whash2a) -- (wenda);

  \draw[-] (whash1b) -- (whash2b);
  \draw[-] (whash1b) -- (whash2b);

  \draw[-] (whash2b) -- (w0b);
  \draw[-] (whash2b) -- (w1b);

  \draw[-,line width=0.1,dashed] (whash2b) -- (wendb);

  \draw[-,line width=0.1,dashed] (wa) -- (wenda);
  \draw[-,line width=0.1,dashed] (wb) -- (wendb);

  \draw[-,line width=0.1,dashed] (w0a) edge[bend right=30] (wnfa);
  \draw[-,line width=0.1,dashed] (w1a) edge[bend left=40] (wnfa);

  \draw[-,line width=0.1,dashed] (w0b) edge[bend left=30] (wnfb);
  \draw[-,line width=0.1,dashed] (w1b) edge[bend right=40] (wnfb);

  \draw[-,line width=0.1,dashed] (wa) edge[bend left=10] (wnfa);
  \draw[-,line width=0.1,dashed] (wb) edge[bend right=10] (wnfb);

  \draw[-,line width=0.1,dashed] (whash1a) -- (wnfa);
  \draw[-,line width=0.1,dashed] (whash2a) -- (wnfa);

  \draw[-,line width=0.1,dashed] (whash1b) -- (wnfb);
  \draw[-,line width=0.1,dashed] (whash2b) -- (wnfb);

  \node[mynodesh] (wlp) at (7.5, 0)   {$w_{\vbl{lp}}$};

    \draw[-,dashed] (w0a) edge   (wlp);
  \draw[-,dashed] (w1a) edge  (wlp);

  \draw[-,dashed] (w0b) edge  (wlp);
  \draw[-,dashed] (w1b) edge  (wlp);

  \draw[-,dashed] (whash1a) edge   (wlp);
  \draw[-,dashed] (whash2a) edge  (wlp);

  \draw[-,dashed] (whash1b) edge  (wlp);
  \draw[-,dashed] (whash2b) edge  (wlp);

\end{tikzpicture}
\caption{The initial epistemic state $\initial$ for the case of a single agent with a reflexive and symmetric accessibility relation. We do not draw self-loops, but every state is accessible from itself. We use arcs instead of arrows, as the accessibility relation is (reflexive and) symmetric.}
\label[figure]{figure:initialRS}
\end{figure}

\begin{definition}
Let $q_a, q_b \in \{0,1\}^*$.
We define the epistemic state $s\block{q_a}{q_b}$ to be 
the epistemic state 
$s\block{q_a}{q_b} = ((W\block{q_a}{q_b},R\block{q_a}{q_b},V\block{q_a}{q_b}),w_{\vbl{root}})$, where
    \begin{align*}
      W\block{q_a}{q_b} = \{ &w_p, w_\xx, w_\vbl{ntF}(\xx) &&\mid p \in \{ \vbl{root}, \vbl{a}, \vbl{b} \},~ \xx \in \{\aa,\bb\}  \} ~
      \cup \\ 
      \{ & w_{\xx,j}, 
            w_{\xx,j,\#_1}, 
      w_{\xx,j,\#_2}
            &&\mid 
            ~ \xx \in \{\aa,\bb\}  \text{ and } 1 \leq j \leq |q_\xx|
            \};
        \\ 
    R\block{q_a}{q_b} &\text{ is the reflexive, symmetric closure of } \\  \{
        &(w_{\vbl{root}},w_{\xx}), (w_\xx,w_{\vbl{ntF}}(\xx))
    && \mid 
    \xx \in \{\aa,\bb\}
    \} ~\cup 
     \\ 
     \{& 
        (w_\xx,w_{\xx,1}), 
                            &\\
        &(w_{\xx,j},w_{\xx,j,\#_1}), (w_{\xx,j,\#_1},w_{\xx,j,\#_2}), 
    &\\ 
                    &(w_{\xx,j},w_{\vbl{ntF}}), (w_{\xx,j,\#_1},w_{\vbl{ntF}}), 
    \\ &
    (w_{\xx,j,\#_2},w_{\vbl{ntF}})
                                                                                        && \mid 
    \xx \in \{\aa,\bb\}, 
    1 \leq j \leq |q_{\xx}|
        \} ~\cup 
     \\ 
     \{& 
    (w_{\xx,j,\#_2},w_{\xx,j+1})
                                && \mid 
    \xx \in \{\aa,\bb\}, 
    1 \leq j < |q_{\xx}|
        \}  
                        ; \text{ and }
    \end{align*}
        for every $\xx \in \{\aa,\bb\}$ and $w_p \in W\block{q_\aa}{q_\bb}$, 
        $V(w_p) = \{p\}$, $V(w_{\vbl{ntF}}(\xx)) = \{\vbl{ntF},\xx\}$, and for every $1 \leq j \leq |q_{\xx}|$, 
        $V(w_{\xx,j}) = \{q_\xx[j],\xx\}$, $V(w_{\xx,j,\#_1}) = \{\#_1,\xx\}$, and $V(w_{\xx,j,\#_2}) = \{\#_2,\xx\}$.
    \end{definition}

\begin{figure}
  \centering 
\begin{tikzpicture}[->,>=stealth',shorten >=1pt,auto,node distance=2.8cm, thick]

    \node[mynode,initial] (wroot)   at (0, 0)      {$w_\vbl{root}$};
    \node[mynode] (wa)   at (1, 1)    {$w_a$};
  \node[mynode] (wb)   at (1, -1)   {$w_b$};
  \node[mynode] (wj0a) at (2, 2)    {$w_{a,1}$};
  \node[mynode] (wj0b) at (2, -0.1)   {$w_{b,1}$};
  \node[mynode] (wj1a) at (3, 1)   {$w_{a,1,\#_1}$};
  \node[mynode] (wj1b) at (3, -1)  {$w_{b,1,\#_1}$};
  \node[mynode] (wj1ash) at (4, 0.1)   {$w_{a,1,\#_2}$};
  \node[mynode] (wj1bsh) at (4, -2)  {$w_{b,1,\#_2}$};

  \node[mynode] (ejka) at (7, 2)   {$w_{a,|q_a|}$};
  \node[mynode] (ejkb) at (7, -0.1)  {$w_{b,|q_b|}$};

  \node[mynode] (ejkash1) at (8, 1)   {$w_{a,|q_a|,\#_1}$};
  \node[mynode] (ejkbsh1) at (8, -1)  {$w_{b,|q_b|,\#_1}$};

  \node[mynode] (ejkash2) at (9, 0.1)   {$w_{a,|q_a|,\#_2}$};
  \node[mynode] (ejkbsh2) at (9, -2)  {$w_{b,|q_b|,\#_2}$};

  \node[mynode] (wnfa) at (7.5, 4)   {$w_{\vbl{ntF}}$};
  \node[mynode] (wnfb) at (4, -4)   {$w_{\vbl{ntF}}$};

      \draw[-] (wroot) -- (wa);
  \draw[-] (wroot) -- (wb);
  
  \draw[-] (wa) -- (wj0a);
  \draw[-] (wb) -- (wj0b);
  
  \draw[-] (wj0a) -- (wj1a);
  \draw[-] (wj0b) -- (wj1b);

  \draw[-] (wj1a) -- (wj1ash);
  \draw[-] (wj1b) -- (wj1bsh);

  \draw[-,dotted] (wj1ash) -- (ejka);
  \draw[-,dotted] (wj1bsh) -- (ejkb);

  \draw[-] (ejka) -- (ejkash1);
  \draw[-] (ejkb) -- (ejkbsh1);
  \draw[-] (ejkash1) -- (ejkash2);
  \draw[-] (ejkbsh1) -- (ejkbsh2);

  \draw[-,line width=0.1,dashed] (wj0a) -- (wnfa);
  \draw[-,line width=0.1,dashed] (wj0b) edge[bend right=18] (wnfb);

  \draw[-,line width=0.1,dashed] (wj1a) -- (wnfa);
  \draw[-,line width=0.1,dashed] (wj1b) edge[bend right=10] (wnfb);

  \draw[-,line width=0.1,dashed] (wj1ash) -- (wnfa);
  \draw[-,line width=0.1,dashed] (wj1bsh) -- (wnfb);

  \draw[-,line width=0.1,dashed] (ejka) -- (wnfa);
  \draw[-,line width=0.1,dashed] (ejkb) -- (wnfb);
  \draw[-,line width=0.1,dashed] (ejkash1) -- (wnfa);
  \draw[-,line width=0.1,dashed] (ejkbsh1) -- (wnfb);

    \draw[-,line width=0.1,dashed] (ejkash2) edge[bend right] (wnfa);
  \draw[-,line width=0.1,dashed] (ejkbsh2) -- (wnfb);
  \draw[-,line width=0.1,dashed] (ejkbsh2) -- (wnfb);

  \draw[-,line width=0.1,dashed] (wa) edge[bend left=45] (wnfa);
  \draw[-,line width=0.1,dashed] (wb) -- (wnfb);

\end{tikzpicture}
\caption{The epistemic state $s\block{q_a}{q_b}$. The state $w_\vbl{ntF}$ appears twice for clarity.}
\end{figure}
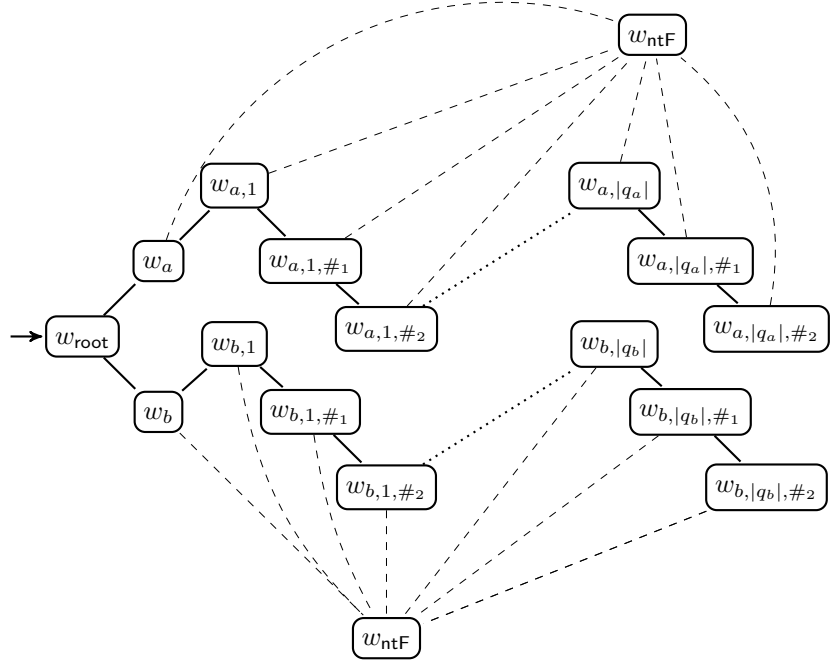

We also define two variations of $s\block{q_a}{q_b}$ that encodes the same block but without the last separator or pair of separators.

\begin{definition}
  Let $q_a, q_b \in \{0,1\}^*$. We define the epistemic states 
  $s\block{q_a}{q_b} - \#_2 $ and 
  $s\block{q_a}{q_b} - \#_1 $ to be the submodels of 
  $s\block{q_a}{q_b}$ induced respectively by 
  $W\block{q_a}{q_b} \setminus \{w_{a,|q_{a}|,\#_2},w_{b,|q_{b}|,\#_2}\}$ and by 
  $W\block{q_a}{q_b} \setminus \{w_{a,|q_{a}|,\#_1},w_{b,|q_{b}|,\#_1},w_{a,|q_{a}|,\#_2},w_{b,|q_{b}|,\#_2}\}$. 
                                                                                                                  \end{definition}

\begin{definition}
  Let $q_1, q_2 \in \{0,1\}^*$. We define the epistemic state $s\block{q_a}{q_b}\{0,1,\#\}$ to be 
  $s\block{q_a}{q_b}\{0,1,\#\} = ((W'\block{q_a}{q_b},R'\block{q_a}{q_b},V'\block{q_a}{q_b}),w_{\vbl{root}})$, where
  \begin{align*}
                  W'\block{q_a}{q_b} = W\block{q_a}{q_b} &\cup &\{ & w_{\vbl{stg1}}, w_{\vbl{0}}(\xx), w_{\vbl{1}}(\xx), w_{\vbl{\#_1}}(\xx), w_{\vbl{\#_2}}(\xx), w_{\vbl{lp}} \mid \xx \in \{\aa,\bb\} \}; \\
            R'\block{q_a}{q_b} = 
        R\block{q_a}{q_b} &\cup  &\{                   &
      (w_{\vbl{root}},w_{\vbl{stg1}}), (w_{\vbl{stg1}},w_{\vbl{root}}), \\
      &&&
      (w_{\vbl{0}}(\xx),w_{\vbl{\#_1}}(\xx)),  (w_{\vbl{1}}(\xx),w_{\vbl{\#_1}}(\xx)), \\
      &&&
      (w_{\vbl{\#_1}}(\xx),w_{0}(\xx)), (w_{\vbl{\#_1}}(\xx),w_{1}(\xx)), \\
      &&&
      (w_{\vbl{\#_1}}(\xx),w_{\vbl{\#_2}}(\xx)),  
      (w_{\vbl{\#_2}}(\xx),w_{\vbl{\#_1}}(\xx)),\\
      &&&
      (w_{\vbl{\#_2}}(\xx),w_{\vbl{0}}(\xx)), (w_{\vbl{\#_2}}(\xx),w_{\vbl{1}}(\xx)), \\
      &&&
      (w_{\vbl{0}}(\xx),w_{\vbl{\#_2}}(\xx)),  (w_{\vbl{1}}(\xx),w_{\vbl{\#_2}}(\xx)), \\
      &&&
      (w_{\vbl{0}}(\xx),w_{\vbl{ntF}}(\xx)),  
      (w_{\vbl{1}}(\xx),w_{\vbl{ntF}}(\xx)), \\
      &&&
      (w_{\vbl{\#_1}}(\xx),w_{\vbl{ntF}}(\xx)),  
      (w_{\vbl{\#_2}}(\xx),w_{\vbl{ntF}}(\xx)) ,\\
      &&&
      (w_{\vbl{ntF}}(\xx),w_{\vbl{0}}(\xx)),  
      (w_{\vbl{ntF}}(\xx),w_{\vbl{1}}(\xx)), \\
      &&&
      (w_{\vbl{ntF}}(\xx),w_{\vbl{\#_1}}(\xx)),  
      (w_{\vbl{ntF}}(\xx),w_{\vbl{\#_2}}(\xx)) 
      \mid \xx \in \{\aa,\bb\}
      \}  \\
      &
      \cup 
      &
      \{ 
       &(w_{\xx,|q_{\xx}|,\#_2},w_{\vbl{bt}}(\xx)), (w_{\vbl{bt}}(\xx),w_{\xx,|q_{\xx}|,\#_2})
                                                            \mid ~ \xx \in \{\aa,\bb\}, ~ q_\xx \neq \varepsilon, \vbl{bt} \in \{0,1\}
      \}  \\
      &
      \cup 
      &
      \{ 
       & (w_{\xx},w_{\vbl{bt}}(\xx)), (w_{\vbl{bt}}(\xx),w_{\xx}) \mid ~ \xx \in \{\aa,\bb\}, ~ q_\xx = \varepsilon, ~ \vbl{bt} \in \{0,1\} \}
        \\
      &
      \cup 
      &
      \{ 
       & (w,w) \mid w \in 
       \{  w_{\vbl{stg1}}, w_{\vbl{0}}(\xx), w_{\vbl{1}}(\xx), w_{\vbl{\#_1}}(\xx), w_{\vbl{\#_2}}(\xx), w_{\vbl{lp}} \mid \xx \in \{\aa,\bb\} \}
       \}
       \\
      &
      \cup 
      &
      \{ 
       & (w,w_\vbl{lp}), (w_\vbl{lp},w)\mid w \in 
       \{   w_{\vbl{0}}(\xx), w_{\vbl{1}}(\xx), w_{\vbl{\#_1}}(\xx), w_{\vbl{\#_2}}(\xx) \mid \xx \in \{\aa,\bb\} \}
       \}
      ;   \end{align*}
            and 
      for every $w \in W\block{q_a}{q_b}$, $V'\block{q_a}{q_b}(w) = V\block{q_a}{q_b}$,
      $V'\block{q_a}{q_b}(w_{\vbl{stg1}}) = \{\vbl{stg1}\}$, and for 
      $p \in \{ \vbl{lp},\vbl{0},\vbl{1},\vbl{\#_1},\vbl{\#_2} \}$ and 
      $\xx \in \{\aa,\bb\}$, 
       $V'\block{q_a}{q_b}(w_p(\xx)) = \{p,\xx\}$.
      \end{definition}

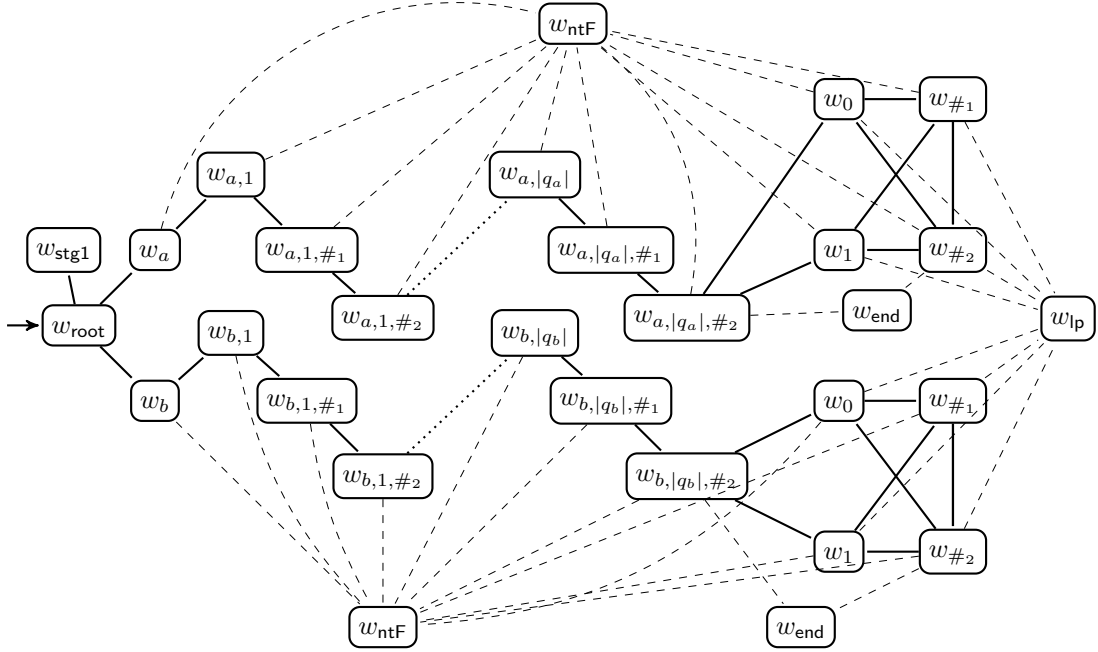
\begin{figure}
  \centering 
\begin{tikzpicture}[->,>=stealth',shorten >=1pt,auto,node distance=2.8cm, thick]

    \node[mynodesh,initial] (wroot)   at (0, 0)      {$w_\vbl{root}$};
  \node[mynodesh] (wst)  at (-0.2, 1)      {$w_{\vbl{stg1}}$};
  \node[mynodesh] (wa)   at (1, 1)    {$w_a$};
  \node[mynodesh] (wb)   at (1, -1)   {$w_b$};
  \node[mynodesh] (wj0a) at (2, 2)    {$w_{a,1}$};
  \node[mynodesh] (wj0b) at (2, -0.1)   {$w_{b,1}$};
  \node[mynode] (wj1a) at (3, 1)   {$w_{a,1,\#_1}$};
  \node[mynode] (wj1b) at (3, -1)  {$w_{b,1,\#_1}$};
  \node[mynode] (wj1ash) at (4, 0.1)   {$w_{a,1,\#_2}$};
  \node[mynode] (wj1bsh) at (4, -2)  {$w_{b,1,\#_2}$};

  \node[mynode] (ejka) at (6, 2)   {$w_{a,|q_a|}$};
  \node[mynode] (ejkb) at (6, -0.1)  {$w_{b,|q_b|}$};

  \node[mynodel] (ejkash1) at (7, 1)   {$w_{a,|q_a|,\#_1}$};
  \node[mynodel] (ejkbsh1) at (7, -1)  {$w_{b,|q_b|,\#_1}$};

  \node[mynodel] (ejkash2) at (8, 0.1)   {$w_{a,|q_a|,\#_2}$};
  \node[mynodel] (ejkbsh2) at (8, -2)  {$w_{b,|q_b|,\#_2}$};

  \node[mynodesh] (w0a)   at (10, 3)  {$w_0$};
  \node[mynodesh] (w1a)   at (10, 1) {$w_1$};
  \node[mynodesh] (whash1a) at (11.5, 3)   {$w_{\#_1}$};
  \node[mynodesh] (whash2a) at (11.5, 1)   {$w_{\#_2}$};
  \node[mynodesh] (wenda) at (10.5, 0.2)   {$w_{\vbl{end}}$};
  
  \node[mynodesh] (w0b)   at (10, -1)  {$w_0$};
  \node[mynodesh] (w1b)   at (10, -3) {$w_1$};
  \node[mynodesh] (whash1b) at (11.5, -1)   {$w_{\#_1}$};
  \node[mynodesh] (whash2b) at (11.5, -3)   {$w_{\#_2}$};
  \node[mynodesh] (wendb) at (9.5, -4)   {$w_{\vbl{end}}$};
  
  \node[mynode] (wnfa) at (6.5, 4)   {$w_{\vbl{ntF}}$};
  \node[mynode] (wnfb) at (4, -4)   {$w_{\vbl{ntF}}$};
  
  \node[mynodesh] (wlp)   [below right = 0.3 and 0.7 of whash2a]  {$w_\vbl{lp}$};

    \draw[-] (wroot) -- (wst);
  \draw[-] (wroot) -- (wa);
  \draw[-] (wroot) -- (wb);
  
  \draw[-] (wa) -- (wj0a);
  \draw[-] (wb) -- (wj0b);
  
  \draw[-] (wj0a) -- (wj1a);
  \draw[-] (wj0b) -- (wj1b);

  \draw[-] (wj1a) -- (wj1ash);
  \draw[-] (wj1b) -- (wj1bsh);

  \draw[-,dotted] (wj1ash) -- (ejka);
  \draw[-,dotted] (wj1bsh) -- (ejkb);

  \draw[-] (ejka) -- (ejkash1);
  \draw[-] (ejkb) -- (ejkbsh1);
  \draw[-] (ejkash1) -- (ejkash2);
  \draw[-] (ejkbsh1) -- (ejkbsh2);

  \draw[-] (ejkash2) -- (w0a);
  \draw[-] (ejkash2) -- (w1a);
  \draw[-] (ejkbsh2) -- (w0b);
  \draw[-] (ejkbsh2) -- (w1b);

  \draw[-] (w0a) -- (whash1a);
  \draw[-] (w1a) -- (whash1a);

  \draw[-] (whash1a) -- (whash2a);
  \draw[-] (whash1a) -- (whash2a);

  \draw[-] (whash2a) -- (w0a);
  \draw[-] (whash2a) -- (w1a);

  \draw[-,line width=0.1,dashed] (whash2a) -- (wenda);

  \draw[-] (w0b) -- (whash1b);
  \draw[-] (w1b) -- (whash1b);

  \draw[-] (whash1b) -- (whash2b);
  \draw[-] (whash1b) -- (whash2b);

  \draw[-] (whash2b) -- (w0b);
  \draw[-] (whash2b) -- (w1b);

  \draw[-,line width=0.1,dashed] (whash2b) -- (wendb);
      \draw[-,line width=0.1,dashed] (ejkash2) -- (wenda);
  \draw[-,line width=0.1,dashed] (ejkbsh2) -- (wendb);

  \draw[-,line width=0.1,dashed] (w0a) -- (wnfa);
  \draw[-,line width=0.1,dashed] (w1a) -- (wnfa);

  \draw[-,line width=0.1,dashed] (w0b) edge[bend left=23] (wnfb);
  \draw[-,line width=0.1,dashed] (w1b) 
  --
    (wnfb);

  \draw[-,line width=0.1,dashed] (whash1a) -- (wnfa);
  \draw[-,line width=0.1,dashed] (whash2a) -- (wnfa);

  \draw[-,line width=0.1,dashed] (whash1b) -- (wnfb);
  \draw[-,line width=0.1,dashed] (whash2b) -- (wnfb);

  \draw[-,line width=0.1,dashed] (wj0a) -- (wnfa);
  \draw[-,line width=0.1,dashed] (wj0b) edge[bend right=15] (wnfb);

  \draw[-,line width=0.1,dashed] (wj1a) -- (wnfa);
  \draw[-,line width=0.1,dashed] (wj1b) edge[bend right=10] (wnfb);

  \draw[-,line width=0.1,dashed] (wj1ash) -- (wnfa);
  \draw[-,line width=0.1,dashed] (wj1bsh) -- (wnfb);

  \draw[-,line width=0.1,dashed] (ejka) -- (wnfa);
  \draw[-,line width=0.1,dashed] (ejkb) -- (wnfb);
  \draw[-,line width=0.1,dashed] (ejkash1) -- (wnfa);
  \draw[-,line width=0.1,dashed] (ejkbsh1) -- (wnfb);

    \draw[-,line width=0.1,dashed] (ejkash2) edge[bend right] (wnfa);
  \draw[-,line width=0.1,dashed] (ejkbsh2) -- (wnfb);
  
  \draw[-,line width=0.1,dashed] (wa) edge[bend left=45] (wnfa);
  \draw[-,line width=0.1,dashed] (wb) -- (wnfb);

  \draw[-,line width=0.1,dashed] (w0a) 
  --
    (wlp);
  \draw[-,line width=0.1,dashed] (w1a) 
    --
  (wlp);

  \draw[-,line width=0.1,dashed] (w0b) 
  --
    (wlp);
  \draw[-,line width=0.1,dashed] (w1b) 
    --
  (wlp);

  \draw[-,line width=0.1,dashed] (whash1a) 
  --
    (wlp);
  \draw[-,line width=0.1,dashed] (whash2a) 
    --
  (wlp);

  \draw[-,line width=0.1,dashed] (whash1b) 
  --
    (wlp);
  \draw[-,line width=0.1,dashed] (whash2b) 
    --
  (wlp);
\end{tikzpicture}
\caption{The epistemic state $s\block{q_a}{q_b} \{0,1,\#\}$. 
}
\end{figure}

\subsection{The Epistemic Actions}

We use the following shorthand formulas:
\begin{align*}
      \logicize{nxt}(0) & = \logicize{nxt}(1) = \#_1,
  &
  \logicize{nxt}(a) & = \logicize{nxt}(b) = \logicize{nxt}(\#_2) = 0 \lor 1,
  \\
      \logicize{nxt}(\#_1) & = \#_2, 
                  & 
  \logicize{symb} & = 0 \lor 1 \lor \#_1 \lor \#_2,
  \\ 
  \logicize{last} & = \bar{K} \vbl{end} \land K \neg \vbl{lp},
                                                          &
  \logicize{loop}_\xx & = \xx \land \bar{K} \vbl{lp}, \quad \text{ and}
        \end{align*}\begin{align*}
  \logicize{tail} & = 
  (\aa \lor \bb) \land 
  \left( 
    (\neg \logicize{symb} \land K
  \neg \vbl{nxt}(\aa)) \lor  
  \bigvee_{\vbl{d} \in \{0,1,\#_1,\#_2\}} (\vbl{d} \land 
  K
  \neg \vbl{nxt}(\vbl{d}))
  \right),
\end{align*}
for every 
$\xx \in \{\aa,\bb\}$.

We now adjust the epistemic actions that add blocks to the encoding of a sequence of blocks. 

For every $1 \leq i \leq n$, we define 
$
\ad_i = (E_i,R_i,e_s), $
where 
        \begin{align*}
                        E_i = \{& e_s,  
                                    e_\xx,  e_{\xx,lst},
                                    e_{st}, e_{\vbl{end}}, e_{\vbl{ntF}} ,
                  e_{01,\xx}                   & \mid & \xx \in \{\aa,\bb\} \}
                  \\
                  & \cup
                  \{
                    e_{\xx,j}, e_{\xx,j,\#_1}, e_{\xx,j,\#_2}
                    & \mid &
                    \xx \in \{\aa,\bb\},
                    ~~ 0 < j \leq  |\xx_i| 
                  \}; 
                  \\ 
                                          R_i &\text{ is the reflexive, symmetric closure of} \\ 
            R_i^- = \{&                             (e_s,e_{st}), (e_s,e_{\xx}),  
                (e_{s},e_{\xx,lst}), 
                \\ 
                &
                (e_\xx,e_{\xx}), 
                (e_\xx,e_{\vbl{ntF}}), 
                (e_\xx,e_{\xx,lst}),\\ 
                &
                (e_{\xx,lst},e_{\vbl{ntF}}), 
                                \\ 
                &
                (e_{01,\xx},
                                e_{01,\xx}),
                                (e_{01,\xx},e_{\vbl{end}}), 
                                                (e_{01,\xx},e_{\vbl{ntF}})
                                &\mid &
                \xx \in \{\aa,\bb\} 
                \}\\
                & \cup \\ 
                \{ & 
                (e_{\xx,lst},e_{\xx,1}), \\ 
                &(e_{\xx,j},e_{\xx,j,\#_1}), (e_{\xx,j},e_{\vbl{ntF}}), \\ 
                &(e_{\xx,j,\#_1},e_{\xx,j,\#_2}), (e_{\xx,j,\#_1},e_{\vbl{ntF}})\\  
                &(e_{\xx,j,\#_2},e_{\vbl{ntF}}), \\  
                &(e_{\xx,|\xx_i|,\#_2},e_{01,\xx}),
                (e_{\xx,|\xx_i|,\#_2},e_{\vbl{end}})
                                                                                &\mid &
                \xx \in \{\aa,\bb\}, ~~ 0 < j \leq |\xx_i|
                \}\\ 
               & \cup \\  &
                \{
                (e_{\xx,j,\#_2},e_{\xx,j+1})
                  &\mid &
                \xx \in \{\aa,\bb\}, ~~ 0 < j  < |\xx_i| 
                \}\\ 
               & \cup \\  &
                \{
                  (e_{\xx,lst},e_{01}), (e_{\xx,lst},e_{\vbl{end}})
                  &\mid &   
                  \xx \in \{\aa,\bb\}, ~~  |\xx_i| = 0 
                \};
                          \end{align*}
            and the precondition for each event in $E_i$ is given below:
            \begin{align*}
                        pre(e_s) =~& \vbl{root} \land \bar{K} \vbl{stg1} 
                        &
                        pre(e_{st}) =~& \vbl{stg1} \\ 
                                    pre(e_\xx) =~& 
                                      \xx 
                          \land \neg \vbl{last}
                                                &
            pre(e_{\xx,lst}) =~& 
                        \xx 
            \land \vbl{last}
            \\ 
                                    pre(e_{\vbl{ntF}}) =~& \vbl{ntF}
                        &
                                    pre(e_{\xx,j}) =~& \xx_i[j] \land \logicize{loop}_\xx
                        \\ 
                                    pre(e_{\xx,j,\#_1}) =~& \#_1 \land \logicize{loop}_\xx
                                    &
            pre(e_{\xx,j,\#_2}) =~& \#_2 \land \logicize{loop}_\xx
                        \\ 
                                    pre(e_{\vbl{end}}) =& \vbl{end}
                        &
                                    pre(e_{01,\xx}) =& \logicize{loop}_\xx \lor \vbl{lp} 
                                                            ,  
                                \end{align*}
        for each $ \xx \in \{\aa,\bb\}$ and $ j \leq |\xx_i|$. 

                The action is illustrated in \Cref{fig:addRS}.

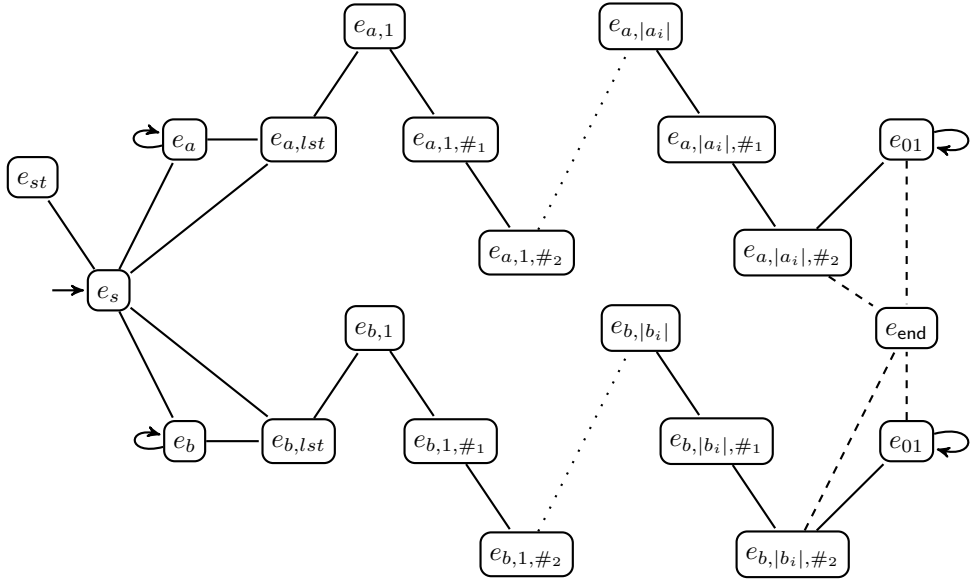
\begin{figure}
  \centering 
\begin{tikzpicture}[->,>=stealth',shorten >=1pt,auto,node distance=2.8cm, thick]

    \node[mynode,initial] (es)   at (0, 0)      {$e_s$};
  \node[mynode] (est)  at (-1, 1.5)      {$e_{st}$};
  \node[mynode] (ea)   at (1, 2)    {$e_a$};
  \node[mynode] (eb)   at (1, -2)   {$e_b$};
  \node[mynode] (ej0a) at (2.5, 2)    {$e_{a,lst}$};
  \node[mynode] (ej0b) at (2.5, -2)   {$e_{b,lst}$};
  \node[mynode] (ej1a) at (3.5, 3.5)   {$e_{a,1}$};
  \node[mynode] (ej1b) at (3.5, -0.5)  {$e_{b,1}$};
  \node[mynode] (ej1ash) at (4.5, 2)   {$e_{a,1,\#_1}$};
  \node[mynode] (ej1bsh) at (4.5, -2)  {$e_{b,1,\#_1}$};
  \node[mynode] (ej1ash2) at (5.5, 0.5)   {$e_{a,1,\#_2}$};
  \node[mynode] (ej1bsh2) at (5.5, -3.5)  {$e_{b,1,\#_2}$};

  \node[mynode] (ejka) at (7, 3.5)   {$e_{a,|a_i|}$};
  \node[mynode] (ejkb) at (7, -0.5)  {$e_{b,|b_i|}$};

  \node[mynode] (ejkash) at (8, 2)   {$e_{a,|a_i|,\#_1}$};
  \node[mynode] (ejkbsh) at (8, -2)  {$e_{b,|b_i|,\#_1}$};
  \node[mynode] (ejkash2) at (9, 0.5)   {$e_{a,|a_i|,\#_2}$};
  \node[mynode] (ejkbsh2) at (9, -3.5)  {$e_{b,|b_i|,\#_2}$};

  \node[mynode] (epsilon01a) at (10.5, 2)  {$e_{01}$};
  \node[mynode] (epsilon01b) at (10.5, -2)  {$e_{01}$};

  \node[mynode] (eend) at (10.5, -0.5)  {$e_{\vbl{end}}$};

    \draw[-] (es) -- (est);
  \draw[-] (es) -- (ea); 
  \draw[-] (es) -- (eb);

  \draw[-] (ea) edge[loop left] (ea);
  \draw[-] (eb) edge[loop left] (eb);
  
  \draw[-] (es) -- (ej0a);
  \draw[-] (es) -- (ej0b);

  \draw[-] (ea) -- (ej0a);
  \draw[-] (eb) -- (ej0b);
  
  \draw[-] (ej0a) -- (ej1a);
  \draw[-] (ej0b) -- (ej1b);

  \draw[-] (ej1a) -- (ej1ash);
  \draw[-] (ej1b) -- (ej1bsh);

  \draw[-,loosely dotted] (ej1ash2) -- (ejka);
  \draw[-,loosely dotted] (ej1bsh2) -- (ejkb);

  \draw[-] (ejka) -- (ejkash);
  \draw[-] (ejkb) -- (ejkbsh);
  
  \draw[-] (ej1ash) -- (ej1ash2);
  \draw[-] (ej1bsh) -- (ej1bsh2);
  \draw[-] (ejkash) -- (ejkash2);
  \draw[-] (ejkbsh) -- (ejkbsh2);
  
  \draw[-] (ejkash2) -- (epsilon01a);
  \draw[-] (ejkbsh2) -- (epsilon01b);
  
    \draw[-] (epsilon01a) edge[loop right] (epsilon01a);
  \draw[-] (epsilon01b) edge[loop right] (epsilon01b);
  
  \draw[-,dashed] (epsilon01a) -- (eend);
  \draw[-,dashed] (epsilon01b) -- (eend);
  \draw[-,dashed] (ejkash2) -- (eend);
  \draw[-,dashed] (ejkbsh2) -- (eend);

\end{tikzpicture}
\caption{The epistemic action $\ad_i$, where the designated event, $e_s$ is marked with an arrow. 
For a cleaner figure, 
we omit event $e_{\vbl{ntF}}$.}
\label[figure]{fig:addRS}
\end{figure}

The remaining reduction proceeds in the same way as the one in \Cref{sec:K}. The difference is that the accessibility relation for $\remv_{\vbl{bt}}$ must become its reflexive and symmetric closure, and that $\vbl{bt}$ should also range over $\#_1$ and $\#_2$. 
Then, as in \Cref{sec:multi}, one needs to also take into account the separator symbols  $\#_1$ and $\#_2$.

\section{A Single Agent with Transitivity}
\label{sec:transitivity}

The case of (single-agent) transitive frames requires a slightly different construction. If we use a similar approach as in the other cases to add blocks to the encoding of a block sequence, it is hard to then distinguish between a state in a final loop and one that is part of the block encoding.
Every state accessible from the loop is also accessible from the encoding state due to transitivity, but it is also the case that every state accessible from an encoding state must have an analogous state accessible from the loop, to be able to then unfold the loop into the encoding states.

Our solution is to switch the order between the loops and the encoding states.

We proceed to adjust our reduction to the case of a single agent with an accessibility relation that is reflexive and transitive (a preorder).

We use the following set of propositional variables: 
\[
\prop = \{ 0, 1, \#, \vbl{a}, \vbl{b}, \vbl{root}, \vbl{stg1}, \vbl{empty}, 
\vbl{lp}  \}.
\]

\subsection{The Epistemic States}

The initial epistemic state $\initial = ((W_0,R_0,V_0),w_{\vbl{root}})$, where
\begin{align*}
        W_0 = \{ &
    w_{p} \mid p \in {\{ \vbl{root}, \vbl{empty}, \vbl{stg1}, \aa, \bb\}}\} 
        \cup 
        \{ w_{p}(\xx) \mid p \in {\{ \vbl{0}, \vbl{1}, \vbl{\#},         \vbl{lp}
    \}}, ~ \xx \in {\{\aa,\bb\}} \};\\
                R_0 &\text{ is the reflexive and transitive closure of: }
    \\ 
    R_0^-
    = \{ 
        &(w_{\vbl{root}},w_{\vbl{empty}}), (w_{\vbl{root}},w_{\vbl{stg1}}), 
        \\ 
    &
        (w_{\vbl{root}},w_{\vbl{bt}}(\xx)), (w_{\vbl{bt}}(\xx),w_{\#}(\xx)), 
    \\ 
    &
                            (w_{\#}(\xx),w_{\vbl{bt}}(\xx)),
    (w_{\#}(\xx),w_{\xx}),
                                                  \\ 
    &
    (w_{\vbl{bt}}(\xx),w_{\vbl{lp}}),
              (w_{\#}(\xx),w_{\vbl{lp}})
          \qquad
      \mid \xx \in \{\aa,\bb\}, ~ \vbl{bt} \in \{0,1\} \}; \end{align*}
        and 
    for every $w_p, w_{p'}(\xx) \in W_0$, $V_0(w_p) = \{p\}$ and $V_0(w_{p'}(\xx)) = \{p',\xx\}$.
    The epistemic state $\initial$ is illustrated in \Cref{figure:initialT}.

\begin{figure}
  \centering 
\begin{tikzpicture}[->,>=stealth',shorten >=1pt,auto,node distance=2.5cm, thick]

    \node[mynode,initial] (r0)   at (0, 0)    {$w_{\vbl{root}}$};
  \node[mynode] (stg)   at (0, 1.5)    {$w_{\vbl{stg1}}$};
  \node[mynode] (empty)   at (-1.5, 1)    {$w_{\vbl{empty}}$};
  \node[mynode] (wa)   at (9, 1)  {$w_{\vbl{a}}$};
  \node[mynode] (wb)   at (9, -1) {$w_{\vbl{b}}$};

  \node[mynode] (w0a)   at (3, 1)  {$w_0$};
  \node[mynode] (w1a)   at (4, 2.5) {$w_1$};
  \node[mynode] (whash1a) at (6, 1)   {$w_{\#}$};
      
  \node[mynode] (w0b)   at (3, -1)  {$w_0$};
  \node[mynode] (w1b)   at (4, -2.5) {$w_1$};
  \node[mynode] (whash1b) at (6, -1)   {$w_{\#}$};

    \draw[->] (r0) -- (stg);
  \draw[->] (r0) -- (empty);
  \draw[->] (whash1a) -- (wa);
  \draw[->] (whash1b) -- (wb);

  \draw[->] (r0) -- (w0a);
  \draw[->] (r0) -- (w1a);
  \draw[->] (r0) -- (w0b);
  \draw[->] (r0) -- (w1b);

  \draw[->] (w0a) -- (whash1a);
  \draw[->] (w1a) -- (whash1a);

  \draw[->] (w0b) -- (whash1b);
  \draw[->] (w1b) -- (whash1b);

  \draw[->] (whash1a) -- (w0a);
  \draw[->] (whash1a) -- (w1a);

  \draw[->] (whash1b) -- (w0b);
  \draw[->] (whash1b) -- (w1b);

  \draw[->] (w1a) -- (w0a);
  \draw[->] (w0a) -- (w1a);

  \draw[->] (w1b) -- (w0b);
  \draw[->] (w0b) -- (w1b);

  \node[mynodesh] (wlp) at (5, 0)   {$w_{\vbl{lp}}$};

    \draw[->,dashed] (w0a) edge   (wlp);
  \draw[->,dashed] (w1a) edge  (wlp);

  \draw[->,dashed] (w0b) edge  (wlp);
  \draw[->,dashed] (w1b) edge  (wlp);

  \draw[->,dashed] (whash1a) edge   (wlp);
    
  \draw[->,dashed] (whash1b) edge  (wlp);
    
\end{tikzpicture}
\caption{The initial epistemic state $\initial$ for the case of a single agent with a reflexive and transitive accessibility relation. We do not draw all the accessibility pairs; the accessibility relation is the reflexive and transitive closure of the one in the figure.}
\label[figure]{figure:initialT}
\end{figure}
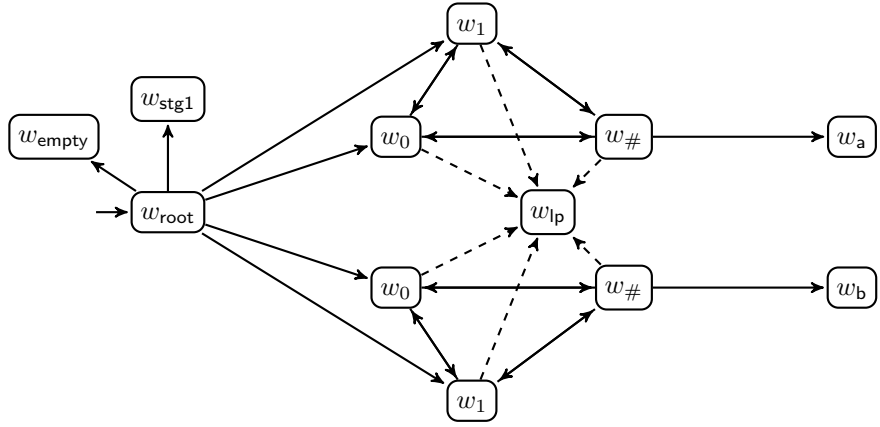

\begin{definition}
Let $q_\aa, q_\bb \in \{0,1\}^*$.
We define 
$s\block{q_\aa}{q_\bb}$ to be 
the epistemic state 
$s\block{q_\aa}{q_\bb} = ((W\block{q_\aa}{q_\bb},R\block{q_\aa}{q_\bb},V\block{q_\aa}{q_\bb}),w_{\vbl{root}})$, where
    \begin{align*}
      W\block{q_\aa}{q_\bb} = \{ &w_p             &&\mid p \in \{ \vbl{root}, \vbl{a}, \vbl{b} \}
            \} ~
      \cup \\ 
      \{ & w_{\xx,j}, 
            w_{\xx,j,\#}
                        &&\mid 
            ~ \xx \in \{\aa,\bb\}  \text{ and } 1 \leq j \leq |q_\xx|
            \};
        \\ 
    R\block{q_\aa}{q_\bb} &\text{ is the reflexive, transitive closure of } \\  \{
        &(w_{\vbl{root}},w_{\xx})     && \mid 
    \xx \in \{\aa,\bb\}
    \} ~\cup 
     \\ 
     \{& 
     (w_{\vbl{root}},w_{\xx,j}),
    (w_{\xx,j},w_{\xx,j,\#}), (w_{\xx,j,\#},w_{\xx})
    && \mid 
    \xx \in \{\aa,\bb\}, 
    1 \leq j \leq |q_{\xx}|
        \} ~\cup  \\ 
     \{& 
    (w_{\xx,j,\#},w_{\xx,j+1})
    && \mid 
    \xx \in \{\aa,\bb\}, 
    1 \leq j < |q_{\xx}|
        \} 
                        ; \text{ and }
    \end{align*}
        for every 
        $w_p \in W\block{q_\aa}{q_\bb}$, 
        $V(w_p) = \{p\}$, 
        and for every $1 \leq j \leq |q_{\xx}|$, 
        $V(w_{\xx,j}) = \{q_\xx[j],\xx\}$ and $V(w_{\xx,j,\#}) = \{\#,\xx\}$.
        \end{definition}

As usual, we 
also define a variations of $s\block{q_\aa}{q_\bb}$ that encodes the same block but without the last separator.

\begin{definition}
  Let $q_\aa, q_\bb \in \{0,1\}^*$. We define the epistemic state 
  $s\block{q_\aa}{q_\bb} - \# $ and 
    to be the submodel of 
  $s\block{q_\aa}{q_\bb}$ induced  by 
  $W\block{q_\aa}{q_\bb} \setminus \{w_{a,|q_\aa|,\#},w_{b,|q_\bb|,\#}\}$. 
                                                                                                                      \end{definition}

\begin{definition}
  Let $q_1, q_2 \in \{0,1\}^*$. We define the epistemic state $s\block{q_\aa}{q_\bb}\{0,1,\#\}$ to be 
  $s\block{q_\aa}{q_\bb}\{0,1,\#\} = ((W'\block{q_\aa}{q_\bb},R'\block{q_\aa}{q_\bb},V'\block{q_\aa}{q_\bb}),w_{\vbl{root}})$, where
  \begin{align*}
                  W'\block{q_\aa}{q_\bb} = W\block{q_\aa}{q_\bb} \cup &\{  w_{\vbl{stg1}}, w_{\vbl{0}}(\xx), w_{\vbl{1}}(\xx), w_{\vbl{\#}}(\xx), w_{\vbl{lp}} \mid \xx \in \{\aa,\bb\} \}; \\
            R'\block{q_\aa}{q_\bb} & \text{ is the reflexive, transitive closure of }\\  
        R\block{q_\aa}{q_\bb} \cup  \{&                         (w_{\vbl{root}},w_{\vbl{stg1}}),  
      (w_{\vbl{root}},w_{\vbl{0}}(\xx)), (w_{\vbl{root}},w_{\vbl{1}}(\xx)),\\
      &
      (w_{\vbl{0}}(\xx),w_{\vbl{\#}}(\xx)),  (w_{\vbl{1}}(\xx),w_{\vbl{\#}}(\xx)), \hfill \\
      &
      (w_{\vbl{\#}}(\xx),w_{0}(\xx)), (w_{\vbl{\#}}(\xx),w_{1}(\xx)), \\
                                                            &
      (w_{\vbl{0}}(\xx),w_{\vbl{lp}}),  
      (w_{\vbl{1}}(\xx),w_{\vbl{lp}}), 
      (w_{\#}(\xx),w_{\xx})
                                                                              \mid \xx \in \{\aa,\bb\}
      \}  \\
      \cup 
            \{ &
       (w_{\#}(\xx),w_{\xx,1})                                                             \mid ~ \xx \in \{\aa,\bb\}, ~ q_\xx \neq \varepsilon       \}  \\
      \cup 
                                                                                                      \{ &
        (w,w_\vbl{lp})
             \mid w \in 
       \{   w_{\vbl{0}}(\xx), w_{\vbl{1}}(\xx), w_{\vbl{\#}}(\xx) \mid \xx \in \{\aa,\bb\} \}
       \}
      ;   \end{align*}
            and 
      for every $w \in W\block{q_\aa}{q_\bb}$, $V'\block{q_\aa}{q_\bb}(w) = V\block{q_\aa}{q_\bb}$,
      $V'\block{q_\aa}{q_\bb}(w_{\vbl{stg1}}) = \{\vbl{stg1}\}$, 
      $V'\block{q_\aa}{q_\bb}(w_{\vbl{lp}}) = \{\vbl{lp}\}$, and for 
      $p \in \{ \vbl{0},\vbl{1},\vbl{\#} \}$ and 
      $\xx \in \{\aa,\bb\}$, 
       $V'\block{q_\aa}{q_\bb}(w_p(\xx)) = \{p,\xx\}$.
      \end{definition}

\begin{figure}
  \centering 
\begin{tikzpicture}[->,>=stealth',shorten >=1pt,auto,node distance=2.8cm, thick]

        \node[mynodesh,initial] (wroot)                                     {$w_\vbl{root}$};
  \node[mynodesh] (wst)     [above = 0.8 of wroot]                    {$w_{\vbl{stg1}}$};
  
    \node[mynodesh] (w1a)     [above right = 0.5 and 1 of wroot]      {$w_1$};
  \node[mynodesh] (w0a)     [above = 0.8 of w1a]                      {$w_0$};
  \node[mynodesh] (whash1a) [below right = 0.1 and 1 of w0a]          {$w_{\#}$};

    \node[mynodesh] (wlp)     [below right = 0.5 and 1 of w1a]                    {$w_\vbl{lp}$};

    \node[mynodesh] (w0b)     [below right = 0.5 and 1 of wroot]      {$w_0$};
  \node[mynodesh] (w1b)     [below = 0.8 of w0b]                      {$w_1$};
  \node[mynodesh] (whash1b) [below right = 0.1 and 1 of w0b]          {$w_{\#}$};
  
    \node[mynodesh] (wj0a)    [right = 0.75 of whash1a]                  {$w_{a,1}$};
  \node[mynodesh] (wj0b)    [right = 0.75 of whash1b]                  {$w_{b,1}$};
  
  \node[mynode]   (wj1a)    [right = 0.6 of wj0a]                       {$w_{a,1,\#}$};
  \node[mynode]   (wj1b)    [right = 0.6 of wj0b]                       {$w_{b,1,\#}$};
  
  \node[mynode]   (ejka)    [right = 1 of wj1a]                     {$w_{a,|q_\aa|}$};
  \node[mynode]   (ejkb)    [right = 1 of wj1b]                     {$w_{b,|q_\bb|}$};

  \node[mynodel]  (ejkash1) [right = 0.6 of ejka]                       {$w_{a,|q_\aa|,\#}$};
  \node[mynodel]  (ejkbsh1) [right = 0.6 of ejkb]                       {$w_{b,|q_\bb|,\#}$};

    \node[mynodesh] (wa)      [right = 0.75 of ejkash1]                  {$w_\aa$};
  \node[mynodesh] (wb)      [right = 0.75 of ejkbsh1]                  {$w_\bb$};

    \draw[->] (wroot) -- (wst);
  
    \draw[->] (wroot) -- (w0a);
  \draw[->] (wroot) -- (w1a);
  \draw[->] (wroot) -- (w0b);
  \draw[->] (wroot) -- (w1b);

    \draw[->] (w0a) -- (whash1a);
  \draw[->] (w1a) -- (whash1a);
  \draw[->] (whash1a) -- (w1a);
  \draw[->] (whash1a) -- (w0a);
  \draw[->] (w1a) -- (w0a);
  \draw[->] (w0a) -- (w1a);
  
  \draw[->] (w0b) -- (whash1b);
  \draw[->] (w1b) -- (whash1b);
  \draw[->] (whash1b) -- (w1b);
  \draw[->] (whash1b) -- (w0b);
  \draw[->] (w1b) -- (w0b);
  \draw[->] (w0b) -- (w1b);
  
    \draw[->] (whash1a) -- (wj0a);
  \draw[->] (whash1b) -- (wj0b);

    \draw[->] (wj0a) -- (wj1a);
  \draw[->] (wj0b) -- (wj1b);

  \draw[->,dotted] (wj1a) -- (ejka);
  \draw[->,dotted] (wj1b) -- (ejkb);
  
  \draw[->] (ejka) -- (ejkash1);
  \draw[->] (ejkb) -- (ejkbsh1);
  
    \draw[->] (ejkash1) -- (wa);
  \draw[->] (ejkbsh1) -- (wb);

    \draw[->,line width=0.1,dashed] (w0a) -- (wlp);
  \draw[->,line width=0.1,dashed] (w1a) -- (wlp);
  \draw[->,line width=0.1,dashed] (w0b) -- (wlp);
  \draw[->,line width=0.1,dashed] (w1b) -- (wlp);
  \draw[->,line width=0.1,dashed] (whash1a) -- (wlp);
  \draw[->,line width=0.1,dashed] (whash1b) -- (wlp);
\end{tikzpicture}
\caption{The epistemic state $s\block{q_\aa}{q_\bb} \{0,1,\#\}$. 
}
\end{figure}

\subsection{The Epistemic Actions}

We use the following shorthand formulas:
\begin{align*}
      \logicize{nxt}(0) & = \logicize{nxt}(1) = \#,
  &
        \logicize{nxt}(\#) & = 0 \lor 1, \\
            \logicize{nxt}(a) & = \logicize{nxt}(b) = 0 \lor 1,
    & 
  \logicize{symb} & = 0 \lor 1 \lor \#,
      \\ 
  \logicize{tail} & = 
  \logicize{symb}
    \land 
        \bigvee_{\vbl{d} \in \{0,1,\#\}} (\vbl{d} \land 
  K
  \neg \vbl{nxt}(\vbl{d}))
            , 
                  \text{ and}
                &
  \logicize{loop}_\xx &= \xx \land \bar{K} \vbl{lp} 
\end{align*}
for every 
$\xx \in \{\aa,\bb\}$.

\begin{remark}
  We note that the definition of $\logicize{tail}$ is the main reason we need the separator symbol $\#$.
  Otherwise, it would be harder to express that a state is the last encoding state when the accessibility relation is reflexive.
\end{remark}

We now show how to adjust the epistemic actions $\ad_i$ for this case. 

For every $1 \leq i \leq n$, we define 
$
\ad_i = (E_i,R_i,e_s), $
where 
        \begin{align*}
                        E_i = \{& e_s,  
                                    e_\xx,  
                                                      e_{st}, 
                                    e_{01,\xx}                   & \mid & \xx \in \{\aa,\bb\} \}
                  \\
                  & \cup
                  \{
                    e_{\xx,j}, e_{\xx,j,\#}
                    & \mid &
                    \xx \in \{\aa,\bb\},
                    ~~ 0 < j \leq  |\xx_i| 
                  \}; 
                  \\ 
                                          R_i &\text{ is the reflexive, transitive closure of} \\ 
            R_i^- = \{&                             (e_s,e_{st}), (e_s,e_{\xx}),  
                (e_{s},e_{01,\xx})
                                                                                                                                                                                                                                                                                                                                &\mid &
                \xx \in \{\aa,\bb\} 
                \}\\
                & \cup \\ 
                \{ & 
                (e_{s},e_{\xx,1}), 
                                                (e_{\xx,j},e_{\xx,j,\#}), (e_{\xx,j},e_{\xx})
                                                                                                                                                                &\mid &
                \xx \in \{\aa,\bb\}, ~~ 0 < j \leq |\xx_i|
                \}\\ 
               & \cup \\  &
                \{
                (e_{\xx,j,\#},e_{\xx,j+1})
                  &\mid &
                \xx \in \{\aa,\bb\}, ~~ 0 < j  < |\xx_i| 
                \}
                                                                                                                ;
                          \end{align*}
            and the precondition for each event in $E_i$ is given below:
            \begin{align*}
                        pre(e_s) =~& \vbl{root} \land \bar{K} \vbl{stg1} 
                        &
                        pre(e_{st}) =~& \vbl{stg1} \\ 
                                    pre(e_\xx) =~& 
                                      \xx 
                          \land \neg \vbl{loop}_\xx
                                                                                                                                                                        &
                                    pre(e_{\xx,j}) =~& \xx_i[j] \land \logicize{loop}_\xx
                        \\ 
                                    pre(e_{\xx,j,\#}) =~& \# \land \logicize{loop}_\xx
                                                                                                                                    &
                                    pre(e_{01,\xx}) =& \logicize{loop}_\xx \lor \vbl{lp} 
                                                            ,  
                                \end{align*}
        for each $ \xx \in \{\aa,\bb\}$ and $ j \leq |\xx_i|$. 

        Action $\ad_i$ corresponds to adding a block at the end of a sequence of blocks. The precondition of $e_s$ ensures that $\bar{K}\vbl{stg1}$ is true at the epistemic state and therefore we are still in the first stage. 
        The action is illustrated in \Cref{fig:addT}.

\begin{figure}
  \centering 
\begin{tikzpicture}[->,>=stealth',shorten >=1pt,auto,node distance=2.8cm, thick]
  \node[mynodesh,initial] (wroot)                                     {$e_s$};
  \node[mynodesh] (wst)     [above = 0.8 of wroot]                    {$e_{st}$};
  
    \node[mynodesh] (w0a)     [above right = 0.5 and 1.2 of wroot]      {$e_{01,\aa}$};

    \node[mynodesh] (w0b)     [below right = 0.5 and 1.2 of wroot]      {$e_{01,\bb}$};
  
    \node[mynodesh] (wj0a)    [right = 1 of w0a]                  {$e_{\aa,1}$};
  \node[mynodesh] (wj0b)    [right = 1 of w0b]                  {$e_{\bb,1}$};
  
  \node[mynode]   (wj1a)    [right = 0.6 of wj0a]               {$e_{\aa,1,\#}$};
  \node[mynode]   (wj1b)    [right = 0.6 of wj0b]               {$e_{\bb,1,\#}$};
  
  \node[mynode]   (ejka)    [right = 1.1 of wj1a]               {$e_{\aa,|\aa_i|}$};
  \node[mynode]   (ejkb)    [right = 1.1 of wj1b]               {$e_{\bb,|\bb_i|}$};

  \node[mynodel]  (ejkash1) [right = 0.6 of ejka]               {$e_{\aa,|\aa_i|,\#}$};
  \node[mynodel]  (ejkbsh1) [right = 0.6 of ejkb]               {$e_{\bb,|\bb_i|,\#}$};

    \node[mynodesh] (wa)      [right = 0.8 of ejkash1]            {$e_\aa$};
  \node[mynodesh] (wb)      [right = 0.8 of ejkbsh1]            {$e_\bb$};

    \draw[->] (wroot) -- (wst);
  
    \draw[->] (wroot) -- (w0a);
    \draw[->] (wroot) -- (w0b);

    \draw[->] (w0a) -- (wj0a);
  \draw[->] (w0b) -- (wj0b);

    \draw[->] (wj0a) -- (wj1a);
  \draw[->] (wj0b) -- (wj1b);

  \draw[->,dotted] (wj1a) -- (ejka);
  \draw[->,dotted] (wj1b) -- (ejkb);
  
  \draw[->] (ejka) -- (ejkash1);
  \draw[->] (ejkb) -- (ejkbsh1);
  
    \draw[->] (ejkash1) -- (wa);
  \draw[->] (ejkbsh1) -- (wb);

              \end{tikzpicture}
\caption{The epistemic action $\ad_i$, where the designated event, $e_s$ is marked with an arrow. 
}
\label[figure]{fig:addT}
\end{figure}
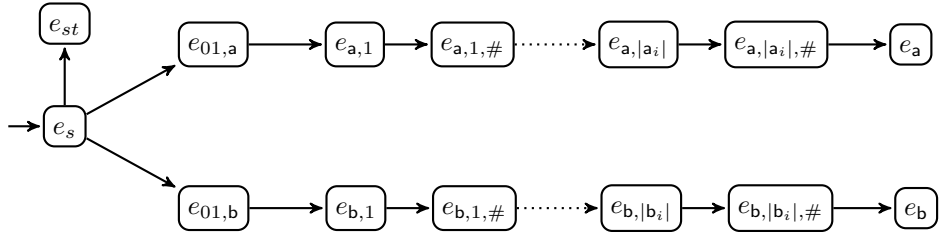

The remaining reduction proceeds in the same way as usual (see \Cref{sec:multi,sec:K,sec:rs}).

\section{Conclusions}
\label[section]{sec:conclusions}
We have shown that the plan existence problem for the case of preconditions of modal depth 1 and no postconditions is undecidable (except for the case of a single agent with negative introspection).
This result completes the picture of (un)decidability of plan existence with respect to the modal depth of the pre- and postconditions. 

The situation may seem grim, as it appears that even very restricted epistemic preconditions suffice to make the plan existence problem undecidable.
However, a more optimistic view is that the modal depth of the preconditions and postconditions is not the right measure to use to determine the hardness of the problem.
In fact, several approaches restrict the problem successfully to ensure decidability.

\paragraph*{Decidability}
Usually, one solves the plan existence problem for a decidable case by showing that the class of relevant epistemic states is finite of finitely representable (\emph{e.g.}~\cite{aucher:hal-01098740,bol-et-al-del-based-ep}).
To actively restrict the plan space, one can bound the length of the plan as in the bounded plan existence problem~\cite{bol-et-al-del-based-ep} or restrict the epistemic actions to be \emph{separable}~\cite{Bolander01012011,bol-et-al-del-based-ep} --- \emph{i.e.}~no two preconditions are consistent with each other.
This assumption ensures that the state space does not grow arbitrarily. In fact, we note that it was crucial for our reduction that the epistemic actions are \emph{not} separable, to allow the plan to encode arbitrarily long sequences of PCP blocks.

The following decidability result is based on known ideas that have the same effect.

\paragraph*{A decidable case: an agent with negative introspection}

When the accessibility relations are constrained to be an equivalence relation, every epistemic state is bisimilar to one where every state is accessible and every set of propositional formulas is satisfied in at most one state.
Then, \Cref{lem:K-is-preservedK} yields that the application of an epistemic action is equivalent to removing a subset of the states, and therefore we can consider only plans that have length at most as much as the number of states in the initial epistemic state.
This observation yields the following theorem.

\begin{theorem}\label[theorem]{thm:s5}
    The plan existence problem for the case of a single agent with reflexive, symmetric, and transitive accessibility relations, is in $\mathsf{NP}$.
\end{theorem}

We can also use a very similar argument for the slightly more general case of a single agent that has a Eucledian accessibility relation \cite{Halpern2007Characterizing}. 

Furthermore, it is not hard to see that in this case, the plan existence problem is $\mathsf{NP}$-hard by a reduction from SAT.
A propositional formula $\varphi$ with $P$ the set of its propositional variables can be turned into an initial state whose worlds are $\{0, p \mid p \in P\}$, where $0 \notin P$;  $p \in L(w)$ if and only if $w=p$ for each $p\in P$ and world $w$; and $0$ is the designated state. 
For each $p \in P$ we can introduce an epistemic action that removes world $p$. Then, the goal is $\varphi'$, where $\varphi'$ results from $\varphi$ by replacing each $p \in P$ with $\bar{K} p$.

\begin{corollary}\label[corollary]{cor:5}
    The plan existence problem for the case of a single agent with accessibility relations that are restricted to be at least Eucledian, is $\mathsf{NP}$-complete.
\end{corollary}

\bibliographystyle{plain}
\bibliography{bibliography}

\end{document}